\ifx\XeTeXversion\undefined\pdfoutput=1\fi
\documentclass{iopjournal}
\graphicspath{{figures/}}
\usepackage[T1]{fontenc}
\usepackage{lmodern}
\usepackage{amsmath}
\usepackage{amssymb}
\usepackage[numbers,sort&compress]{natbib}
\usepackage{booktabs}
\usepackage{multirow}
\usepackage{amsthm}
\newtheorem{theorem}{Theorem}
\newtheorem{corollary}{Corollary}
\newtheorem{lemma}{Lemma}

\newcommand{\warpax}{{\normalfont\textsc{warpax}}}
\newcommand{\dd}{\mathrm{d}}

\newcommand{\rodalWECmissVSfive}{15.6}  %
\newcommand{\rodalDECmissVSfive}{28.5}  %

\newcommand{\rodalVortFrac}{0.000}  %
\newcommand{\natarioVortFrac}{0.508}  %
\newcommand{\natarioExpFrac}{0.000}  %
\newcommand{\alcubierreVortFrac}{0.333}  %
\newcommand{\vdbVortFrac}{0.440}  %
\newcommand{\rodalNECscaling}{0.768}  %
\newcommand{\vdbTransitionVS}{0.38}  %
\newcommand{\kappaVorticity}{0.060}  %

\renewcommand{\headrulewidth}{0pt}
\renewcommand{\articletype}[1]{\vspace*{-8mm}\noindent{\scriptsize\sf\bfseries\MakeUppercase{#1}}}

\begin{document}

\articletype{Paper}

\title{Observer-robust energy condition verification
for warp drive spacetimes}

\author{An T. Le$^{1,2,3}$}

\affil{$^1$Center for Environmental Intelligence, VinUniversity, Hanoi, Vietnam}

\affil{$^2$College of Engineering and Computer Sciences, VinUniversity, Hanoi, Vietnam}

\affil{$^3$Intelligent Autonomous Systems, TU Darmstadt, Germany}

\email{an.lt@vinuni.edu.vn; an@robot-learning.de}

\begin{abstract}
Energy-condition tests for warp-drive spacetimes must account for all admissible
causal directions at each point. We use the classical S-lemma to express the
null, weak, and strong conditions as $4\times4$ linear matrix inequalities;
the dominant condition requires two such tests. These criteria require neither
Hawking--Ellis classification nor a rapidity cutoff. For the null condition
in a fixed orthonormal tetrad, the optimized multiplier margin equals half
the minimum normalized null-energy contraction. Interval evaluation of the metric and curvature
provides pointwise certificates and global bounds on its minimum over
the bubble wall. On flat unit-lapse slices, the momentum constraint relates
Eulerian momentum to shift vorticity. For smooth shifts on all of
Euclidean space with bounded vorticity, momentum vanishes identically precisely
for a gradient plus rigid rotation. For shifts linear
in speed, integrated negative Eulerian energy scales exactly quadratically
when finite, with a fixed profile and integration domain on these slices.
We implement the tests in \warpax{}, a JAX toolkit, and compare four warp-drive
geometries at matched parameters across subluminal and superluminal speeds.
At the reference parameters, the global bounds establish null-energy violation
in all four bubble walls.
The ideal irrotational Rodal profile has zero Eulerian momentum and is
everywhere Type~I; Type-IV regions are detected in the other sampled walls.
For Rodal, the Eulerian reading misses about $73\%$ of the sampled
wall weak-energy violations. These numerical type labels and sampled
fractions are distinct from the interval bounds. Finite-segment null-geodesic
integrals and flat-space quantum-inequality estimates supplement the pointwise
analysis.
\end{abstract}

\keywords{warp drive, energy conditions, Hawking--Ellis
classification, general relativity, automatic differentiation, JAX}

\section{Introduction}
\label{sec:intro}

Alcubierre's 1994 warp drive metric \cite{alcubierre1994} demonstrated
that general relativity admits solutions describing effective
superluminal travel by deforming the spacetime geometry around a
compact region, a ``warp bubble''.  The bubble center follows a timelike worldline while the bubble
moves at an arbitrary coordinate velocity $v_s$ through the
asymptotically Minkowski exterior.  The metric takes the
Arnowitt--Deser--Misner (ADM) $3+1$ form
\begin{equation}
  \dd s^2 = -\dd t^2 + (\dd x - v_s f(r_s)\,\dd t)^2 + \dd y^2 + \dd z^2\,,
  \label{eq:alcubierre}
\end{equation}
where $v_s=\dot x_s(t)$ and $f(r_s)$ is a smooth top-hat profile,
approximately unity inside the bubble and tending to zero outside.  Here
$r_s = [(x - x_s(t))^2 + y^2 + z^2]^{1/2}$ is the distance from the
bubble center.

A central obstacle for the warp metrics considered here is their requirement
for exotic matter: matter whose stress-energy tensor violates the
classical energy conditions.  Throughout this paper, the spacetime
signature is $(-+++)$ and we work in geometric units ($G = c = 1$).
The four energy conditions, namely the null (NEC), weak
(WEC), strong (SEC), and dominant (DEC), are pointwise inequalities on
the stress-energy tensor contracted with timelike or null vectors:
\begin{alignat}{2}
  \text{NEC:} & \quad T_{ab}\, k^a k^b \geq 0
    & \quad & \forall\; \text{null } k^a\,,
    \label{eq:nec} \\
  \text{WEC:} & \quad T_{ab}\, u^a u^b \geq 0
    & \quad & \forall\; \text{timelike } u^a\,,
    \label{eq:wec} \\
  \text{SEC:} & \quad \left(T_{ab} - \tfrac{1}{2}\,T\,g_{ab}\right)
    u^a u^b \geq 0
    & \quad & \forall\; \text{timelike } u^a\,,
    \label{eq:sec} \\
  \text{DEC:} & \quad T_{ab}\, u^a u^b \geq 0 \;\text{and}\;
    {-T^a{}_b\, u^b}\;\text{future causal}
    & \quad & \forall\; \text{future timelike } u^a\,.
    \label{eq:dec}
\end{alignat}
Here $T=g^{ab}T_{ab}$ is the trace, and the DEC permits zero energy flux.
These conditions must hold for \emph{all} admissible observers at each
spacetime point \cite{santiago2021,hawking1973}.  An energy condition either holds or fails at a point independently of the
observer.  A single-frame \emph{diagnostic}, such as $T_{ab}u^au^b$ for a
particular $u^a$, depends on that observer.

For a Type~I stress-energy tensor, the energy conditions reduce to inequalities
on the invariant rest-frame energy density $\rho$ and principal pressures $p_i$;
the eigenvalues of $T^a{}_b$ are $-\rho$ and the $p_i$.  A Type~IV tensor has no
causal eigenvector and violates the NEC
\cite{hawking1973,santiago2021,martinmoruno2018core,martinmoruno2017lnp}.
Section~\ref{sec:classification} states these classical criteria and derives
the momentum discriminant under its principal-axis hypothesis.
Theorem~\ref{thm:lmi} gives a complementary formulation: $4\times4$ linear
matrix inequalities decide all four conditions at every algebraic type,
including Types~II and~III.  These criteria require
neither a rapidity cap nor a timelike coordinate vector $\partial_t$.
Corollary~\ref{cor:margin} relates the optimized multiplier to the null-cone
minimum in a fixed tetrad.  Interval evaluation of the metric and its derivatives
then gives bounds on that minimum.  The pointwise interval tests decide all
four conditions at each of $4257$ sampled wall points per drive
(Appendix~\ref{app:interval-census}); the separate global search encloses the
wall infimum of the normalized null contraction (Appendix~\ref{app:enclosures}).

The algebraic type is related to the Eulerian momentum density $j$.
When its direction is a principal axis of the spatial stress, the sign of
$\Delta=(\rho_n+S_\parallel)^2-4|j|^2$ determines the Type-I/II/IV transition.
Transverse stress coupling requires the separate reduction of
Corollary~\ref{cor:conformal}.  On time-independent flat unit-lapse slices,
$8\pi|j|=\tfrac12|\nabla\times(\nabla\times\beta)|$.
For a smooth shift on $\mathbb R^3$ with bounded vorticity, vanishing momentum
characterizes gradient shifts up to rigid rotation (Lemma~\ref{lem:gradient}).
This result concerns momentum; the observed vorticity/type association is
empirical.  For speed-linear shifts, the same constraints give an exact
quadratic law for the integrated negative Eulerian energy on a fixed domain
(Lemma~\ref{lem:integrated}).  A quadratic law for the pointwise eigenvalue
margin follows under the additional assumption $j=0$.

\smallskip\noindent We compare four drives at matched parameters: Alcubierre,
Nat\'ario, Van~den~Broeck, and Rodal, using the laboratory-frame conventions of
Appendix~\ref{app:metric_definitions}.  The ideal irrotational Rodal geometry is
Type~I everywhere at every speed (Lemma~\ref{lem:rodal}).  The sampled Alcubierre
and Nat\'ario walls contain Type-IV regions; Van~den~Broeck's wall Type-IV
fraction first crosses $50\%$ at $v_s\approx\vdbTransitionVS$.
Within Rodal's wall the
single-frame Eulerian diagnostic does not register ${\approx}74\%$ of the
dominant and ${\approx}73\%$ of the weak energy-condition violations that boosted
observers see; the drive's positive-energy characterization rests on its
predominantly positive invariant energy density~\cite{rodal2026}.
Its irrotational shift establishes Type~I, but does not establish energy-condition satisfaction.  The composite diagnostic of Section~\ref{sec:construction-exoticity} uses a
specified slice, wall mask, numerical caps and floors, and affine normalization.
Finite-segment null-geodesic and quantum-inequality diagnostics are reported in
Section~\ref{sec:averaged-quantum}.

Prior warp-drive tooling (WarpFactory~\cite{warpfactory2024}) samples
${\sim}10^3$ observer directions per point;
we evaluate the Type~I eigenvalue criteria and Type-IV labels directly,
using curvature computed by automatic differentiation
(Section~\ref{sec:related}).  Shell constructions such as the
regularized WarpShell~\cite{bobrick2021,fell2021} sit outside the matched
benchmark set; recent positive-energy
constructions~\cite{fuchs2024constvel,garattini2025desitter} are compared
separately in Section~\ref{sec:construction-exoticity}.

\subsection{Related work}
\label{sec:related}

\paragraph{Energy conditions in warp-drive spacetimes.}
Hawking and Ellis~\cite{hawking1973} give the standard algebraic classification of
$T^a{}_b$ into four types; Mart\'in-Moruno and Visser give its modern treatment
and the conditions under which each energy condition can be
violated at each type~\cite{martinmoruno2017lnp,martinmoruno2018core,
martinmoruno2018type3,martinmoruno2017rainich,martinmoruno2021backreaction}.
Santiago, Schuster and Visser~\cite{santiago2021} proved that
within the generic Nat\'ario framework physically reasonable warp drives
violate the weak, strong and dominant conditions; their null-energy conclusion
also requires the subsidiary conditions stated in that work.  The qualification ``physically reasonable''
is theirs, a hypothesis on the localization of the warp field.  Their conclusions are
observer-independent but not framework-independent, and Barzegar, Buchert and
Vigneron~\cite{barzegar2026classification} examine the routes that would evade
its hypotheses.  An irrotational shift on time-independent Euclidean slices with unit lapse
gives Type-I stress-energy, as in Rodal's construction~\cite{rodal2026}.
Rodal also reports an increase in integrated negative energy when shift
vorticity is introduced.  That integrated result does not determine the
pointwise algebraic type.
Olum~\cite{olum1998} showed that superluminal travel, under his
destination-surface definition and generic-condition hypothesis, requires
negative energy.  Lobo and Visser~\cite{lobovisser2004}
established the unavoidable energy-condition violation of the Alcubierre and
Nat\'ario classes.  The applicability of Olum's result requires care: Shoshany and
Snodgrass~\cite{shoshany2024ctc} argue that its hypotheses, in particular the
generic condition, fail for a wide class including Alcubierre and
Van~den~Broeck.  Kontou and Sanders~\cite{kontou2020} review the modern status of
the conditions.  The strong condition is violated by a positive cosmological
constant, for which $\rho+3p=-\Lambda/4\pi<0$.
The Bobrick--Martire taxonomy of
\emph{physical warp drives}~\cite{bobrick2021} and the positive-energy
construction of Fell and Heisenberg~\cite{fell2021} examine how the geometric
assumptions restrict the allowed stress-energy.

\paragraph{Warp-drive geometries.}
The foundational constructions remain Alcubierre's 1994
metric~\cite{alcubierre1994}, Nat\'ario's zero-expansion
variant~\cite{natario2002} and Van den Broeck's modified-volume
construction~\cite{vandenbroeck1999}; the Krasnikov tube~\cite{krasnikov1998}
replaces the moving bubble with a static corridor.  The Lentz
soliton~\cite{lentz2021} was proposed as a positive-energy alternative but was
subsequently shown to violate the NEC~\cite{santiago2021} and the WEC in the
Eulerian frame~\cite{celmaster2025}.  The Rodal metric~\cite{rodal2026} completes
our benchmark set; Rodal's earlier work~\cite{rodal2023invariants,rodal2024natario}
analyzes curvature invariants for Alcubierre and Nat\'ario in the same $3+1$
formalism, and later notes argue against low-energy warp drives on metamaterial
gravitational-coupling grounds~\cite{rodal2025infeasibility} and develop
photon-propagation branch diagnostics on black-hole
backgrounds~\cite{rodal2026screening}.  Among the
2024--2026 constructions, Fuchs \emph{et al.}~\cite{fuchs2024constvel} introduce a
constant-velocity shell-based positive-energy drive in the Bobrick--Martire
family; Garattini and Zatrimaylov~\cite{garattini2024bh} show that embedding a
warp drive on a black-hole background alleviates energy-condition violations, and
generalize to a de~Sitter background~\cite{garattini2025desitter}.  Their de~Sitter analysis expresses the weak and null conditions
``up to a total divergence term that averages to zero''.  These divergence
identities do not establish pointwise energy-condition satisfaction.
Abell\'an, Bol\'ivar and Vasilev~\cite{abellan2024lapse} treat the lapse as an
additional degree of freedom, a nontrivial lapse accommodating a heat-flux
matter content, complementary to
the unit-lapse, flat-slice family analyzed here.  Clough, Dietrich and
Khan~\cite{clough2024gw} evolve a warp-drive containment failure numerically, complementing the
pointwise analysis here, and White
\emph{et al.}~\cite{white2025nacelle} segment the bubble
into interior-flat Gaussian nacelles, reshaping the
energy distribution rather than the all-observer character of the wall.
Huey~\cite{huey2024} proposes a
superluminal drive sourced by non-compact thin membranes whose distributional
stress-energy satisfies the WEC; such constructions sit outside both the
Santiago--Schuster--Visser hypotheses and the smooth-metric scope used here.
Santos-Pereira \emph{et al.}~\cite{santospereira2025matching} study Darmois
junction conditions for matching Alcubierre to Minkowski.

Two classification programs bear directly on what is measured here.  Barzegar
and Buchert~\cite{barzegar2024restrictions} identify shared restrictions of the
current warp-drive spacetimes (flow-orthogonal foliations, vanishing spatial
curvature) and establish, for a one-dimensional shift at $\Lambda=0$, the
relation $E=-\Omega^2/8\pi\le0$ between the coordinate vorticity and the Eulerian
energy density.  Fell and Heisenberg~\cite{fell2021} relate vanishing shift
vorticity to vanishing Eulerian momentum.  Section~\ref{sec:shift-vorticity} compares momentum and numerical type
labels for the sampled geometries.  Barzegar, Buchert and
Vigneron~\cite{barzegar2026classification} develop a general classification
and no-go results involving global hyperbolicity, asymptotic flatness and
positive energy.  Under their R-Warp hypotheses a harmonic-gradient shift implies Minkowski;
this gives no conclusion for nonharmonic Rodal (Lemma~\ref{lem:rodal}).  They also question the consistency of the
Fuchs shell's assumed matter source with its metric.  Our pointwise tests of
$G_{ab}/8\pi$ do not resolve that source-consistency question.  Buchert and
Frackowiak~\cite{buchert2026realizations} construct new warp solutions and
analyze their coordinate-acceleration and vorticity structure kinematically.
These geometric classifications complement the Hawking--Ellis classification
of the induced stress-energy.

\paragraph{Numerical methods.}
Computational tools for general relativity span symbolic algebra
(xAct, including xPerm~\cite{martingarcia2008}, and
SageManifolds~\cite{gourgoulhon2018}) and
numerical-relativity frameworks (the Einstein Toolkit~\cite{einsteintoolkit2024},
which evolves spacetime geometries); see
Ref.~\cite{maccallum2018} for a survey.  The pointwise tests combine causal-cone algebra
(Theorem~\ref{thm:lmi}) with interval arithmetic
(Appendix~\ref{app:interval-lmi}).  Bergqvist and
Senovilla~\cite{bergqvist2001nullcone} characterize the tensors with the dominant
property as the convex cone generated by superenergy tensors of simple forms,
and Mart\'in-Moruno and Visser develop the complementary polynomial
route~\cite{martinmoruno2017rainich,martinmoruno2018core}.  The S-lemma supplies the pointwise linear matrix
inequality used here; the underlying S-lemma and Lorentz-cone results are
standard~\cite{yakubovich1971,polik2007slemma,loewy1975icecream,
hildebrand2007lorentz,hildebrand2011lorentzII}.  
Interval arithmetic with directed rounding encloses numerical errors.  Appendix~\ref{app:enclosures} propagates such enclosures
through $g\to\Gamma\to R^a{}_{bcd}\to G_{ab}\to T_{ab}$.
A strict bound can establish the sign of an energy-condition contraction;
a saturated point can remain inconclusive.

The warp-drive community's dedicated toolkit, WarpFactory, of Helmerich
\emph{et al.}~\cite{warpfactory2024,warpfactory_toolkit2024,warpfactory_docs2024},
evaluates the conditions over a discrete
sample of ${\sim}10^3$ observer directions and velocities and uses fourth-order
central finite differences for the curvature chain.  \warpax{} uses
forward-mode automatic differentiation in JAX~\cite{jax2018} for the
curvature chain, as in diffjeom~\cite{coogan2024diffjeom} and the neural
\emph{Einstein Fields} representation~\cite{cranganore2025einsteinfields}.
Bara~\cite{bara2025} discusses prospective astrophysical applications of
automatic differentiation.  \warpax{} combines the algebraic energy-condition
tests with continuous observer optimization~\cite{optimistix2024} and
geodesic integration via Diffrax~\cite{kidger2022diffrax}, in an
Equinox~\cite{kidger2021equinox} module structure.  Numerical Type-IV detection uses the mixed tensor
eigensystem, with the tolerances stated in Section~\ref{sec:classification}.

Section~\ref{sec:methods} presents the methods, Section~\ref{sec:results-ec} the
energy-condition results, Section~\ref{sec:convergence} the convergence
validation, Section~\ref{sec:discussion} the implications and
Section~\ref{sec:conclusion} the conclusions; the \warpax{} implementation
(Appendix~\ref{app:toolkit}) and the metric definitions
(Appendix~\ref{app:metric_definitions}) are given in the appendices.

\subsection{Scope and limitations}
\label{sec:scope}

The energy-condition criteria are exact algebraic statements.  Their numerical
evaluation has a separate accuracy limit.  Type-I eigenvalue inequalities
apply when the type is established; numerical Hawking--Ellis labels depend on
eigenvalue and eigenvector tolerances.  The LMI criteria of
Theorem~\ref{thm:lmi} apply independently of those labels.  On the full grids,
their binary64 evaluation uses the noise floor~\eqref{eq:noise-floor} and
returns \emph{inconclusive} within it.

An exact multiplier or dual certificate establishes a verdict for its stated
input tensor.  Interval evaluation of the metric and curvature also
encloses that input, giving pointwise verdicts and global bounds for the
specified wall null deficit (Appendices~\ref{app:enclosures}
and~\ref{app:interval-lmi}).  Sampled type fractions, miss rates, and
basin-local extrema retain their numerical status.  Their grid stability
does not supply a rigorous error bound.

The rapidity-capped observer search measures observer-dependent contractions.
A negative contraction supplies a violation witness; a nonnegative search
result cannot establish satisfaction for all boosts.  Throughout this paper,
\emph{observer-robust} refers to the frame-independent truth of the pointwise
energy condition.  Numerical severities depend on the stated normalization
and, for the timelike search, the cap.  The single-frame comparison uses the
Eulerian normal $n^a$, which is unit timelike at every speed.  We study
$v_s\in[0.1,2.5]$, including the luminal transition.

For the Alcubierre form $g_{00} = -1 + (v_s f)^2$, so $\partial_t$ is null exactly on
the locus
\begin{equation}
  v_s\,f(r_s) = 1,
  \label{eq:ergo-locus}
\end{equation}
and spacelike where $v_s f > 1$.  Since $f(0)=1$ and $f'(r_s) < 0$ for every
$r_s > 0$, the shape function is strictly decreasing, so \eqref{eq:ergo-locus} has a
solution only for $v_s \ge 1$, and at $v_s = 1$ that solution is the single point
$r_s = 0$, not a surface and not the wall's inner edge.  The locus lies inside the
active wall mask $0.1 \le f \le 0.9$ precisely when $f = 1/v_s$ falls in that band,
that is for
\begin{equation}
  10/9 \le v_s \le 10,
  \label{eq:ergo-range}
\end{equation}
and not for every $v_s > 1$.  At this locus $\partial_t$ ceases to define a
timelike observer, but the spacetime remains Lorentzian.  The determinant is
$\det g=-1$ for Alcubierre, Nat\'ario, and Rodal, and $\det g=-B^6<0$ for
Van~den~Broeck.  The Eulerian normal
$n^a = (1/\alpha)(\partial_t - \beta^i \partial_i)$ stays unit timelike
($\alpha = 1$, $g^{00} = -1$ at all $v_s$).  The eigenvalues of $T^a{}_b$ are
boost invariants and vary continuously as a multiset across the transition.
The Hawking--Ellis classification therefore remains well defined.
The computed curvature and
energy-condition quantities are finite at $v_s \in \{1.0, 1.5, 2.0, 2.5\}$, and we compute the
type and eigenvalue structure there
(Section~\ref{sec:results-ec}, Table~\ref{tab:velocity_type}).  The
continuous optimizer, used only at non-Type~I points, parameterizes
future-directed unit-timelike vectors of the metric at each point, so its
observer domain follows the local light cones through the transition rather
than a fixed background cone.  We also guard against the principal
numerical risk of the superluminal eigenproblem, a spurious Type~IV from an
ill-conditioned non-symmetric $T^a{}_b$: the Type-IV labels reported here
are compared on selected samples using standard and generalized eigensolvers
and a $50$-digit recomputation. These comparisons do not certify every grid
label or convergence of a spatial fraction.  What is decided at
$v_s>1$ is the energy-condition verdict, by Theorem~\ref{thm:lmi}; the algebraic
labels there are checked, not certified.  Neither establishes physical realizability or the global causal properties of
a superluminal construction.

We compute $T_{ab}=G_{ab}/8\pi$ from prescribed metrics, the Synge
$G$-method.  The resulting energy-condition tests characterize the induced
classical tensor.  They do not establish a realizable matter source,
well-posed matter evolution, or global causal viability.  These additional
requirements are discussed by Barzegar, Buchert, and
Vigneron~\cite{barzegar2026classification}.

Pointwise NEC violation does not establish violation of the averaged null
energy condition (ANEC) on a complete null geodesic~\cite{grahamolum2007}.
Section~\ref{sec:averaged-quantum} reports finite-segment null-energy
integrals with selected-ray step refinement.  No omitted-tail or validated
trajectory-error bound is available.  The quantum-inequality comparison is
a flat-space estimate~\cite{ford1995,pfenning1997}; applying a curved-space
bound to these worldlines requires a separate treatment
\cite{fewster2000qi,fewstereveson1998,kontou2020}.

\section{Methods}
\label{sec:methods}

The calculation starts from a prescribed metric and its induced stress-energy.
We then test the energy conditions over all causal directions.  Algebraic
classification helps interpret the tensor, while the matrix criteria provide
a test even when the type is difficult to determine numerically.  Observer
searches show which measurements reveal a violation.

\subsection{Stress-energy via automatic differentiation}
\label{sec:autodiff}

The Einstein field equations in geometric units ($G = c = 1$) are
\begin{equation}
  G_{ab} = 8\pi\, T_{ab}\,,
  \label{eq:efe}
\end{equation}
where the Einstein tensor $G_{ab} = R_{ab} - \frac{1}{2} R\, g_{ab}$
is constructed from the Ricci tensor $R_{ab}$ and scalar curvature
$R = g^{ab} R_{ab}$.  We use the total effective tensor $T_{ab}=G_{ab}/8\pi$;
for a cosmological background this includes $-\Lambda g_{ab}/8\pi$.
It is not the material tensor obtained after subtracting that contribution.
Computing $T_{ab}$ from a given metric requires
the full curvature chain:
\begin{equation}
  g_{ab} \;\xrightarrow{\partial}\;
  \Gamma^a{}_{bc} \;\xrightarrow{\partial}\;
  R^a{}_{bcd} \;\xrightarrow{\text{contract}}\;
  R_{ab},\, R \;\xrightarrow{\text{EFE}}\; T_{ab}\,.
  \label{eq:chain}
\end{equation}
The Christoffel symbols of the second kind are
\begin{equation}
  \Gamma^a{}_{bc} = \frac{1}{2}\, g^{ad}
    \left(\partial_c\, g_{bd} + \partial_b\, g_{cd}
    - \partial_d\, g_{bc}\right)\,,
  \label{eq:christoffel}
\end{equation}
and the Riemann curvature tensor is
\begin{equation}
  R^a{}_{bcd} = \partial_c\,\Gamma^a{}_{bd}
    - \partial_d\,\Gamma^a{}_{bc}
    + \Gamma^a{}_{ce}\,\Gamma^e{}_{bd}
    - \Gamma^a{}_{de}\,\Gamma^e{}_{bc}\,.
  \label{eq:riemann}
\end{equation}

\warpax{} implements this chain with JAX's forward-mode automatic
differentiation.  The metric function $g_{ab}(x^\mu)$ maps a coordinate 4-vector
to a $4 \times 4$ symmetric matrix.  Forward-mode differentiation yields
$\partial_c\, g_{ab}$ as exact derivatives of the implemented
program, up to roundoff, with no step size and no finite-difference truncation
error; nesting it (forward-over-forward) gives the second derivatives
needed for the Riemann tensor.  For the metric map $\mathbb R^4\to\mathbb R^{16}$, forward mode
requires four input directions, whereas reverse mode requires one pass per
scalar output to form the full Jacobian.  The derivative
convention places the differentiation index last,
$(\partial_c\, g)_{abc} = \partial g_{ab} / \partial x^c$, which fixes the index order in the contractions.

\paragraph{Three energy-density quantities.}
Three energy densities recur throughout this paper.  The \emph{invariant energy
density} $\rho$ at a Type-I point is minus the timelike eigenvalue of
$T^a{}_b$; it equals $T_{ab}u^au^b$ on a unit timelike eigenvector.  The \emph{observer-dependent energy density}
$T_{ab}\,u^a u^b$ is what an observer with four-velocity $u^a$ measures.  The
\emph{Eulerian energy density} $T_{ab}\,n^a n^b$ is the case of $n^a$ the unit
normal to the ADM spatial hypersurface, the single-frame baseline used in this paper.

\paragraph{Notation.}
$R_b$ denotes the warp bubble radius throughout, distinct from the Ricci scalar
$R$ of the curvature expressions.

\subsection{Hawking--Ellis classification}
\label{sec:classification}

The algebraic classification of $T^a{}_b$ \cite{hawking1973} identifies whether
a local rest frame exists.  When it does, energy-condition tests reduce to
inequalities on the rest-frame density and principal pressures.
The four types describe the eigenstructure of
$T^a{}_b=g^{ac}T_{cb}$:

\begin{itemize}
  \item \textbf{Type~I:} Four real eigenvalues with a timelike eigenvector.
    This class includes perfect fluids.  For signature
    $(-+++)$ they are
    $(-\rho, p_1, p_2, p_3)$; the timelike eigenvector is identified
    by $g_{ab}\,v^a v^b < 0$, with $\rho = -\lambda_{\text{timelike}}$,
    $p_i = \lambda_{\text{spacelike},i}$.
  \item \textbf{Type~II:} A
    defective $2\times 2$ null Jordan block, a repeated eigenvalue with a single null eigenvector in that block.
    Canonical example: null dust
    $T_{ab} = \Phi^2 k_a k_b$ with $k^a$ null
    \cite{martinmoruno2018core}.
  \item \textbf{Type~III:} $3\times 3$ null Jordan structure, a
    null Jordan block of size three \cite{martinmoruno2018type3}.
  \item \textbf{Type~IV:} One complex-conjugate eigenvalue pair and two real
    eigenvalues, with no causal eigenvector.  Numerical labels depend on the
    classification tolerances (Section~\ref{sec:scope}).
\end{itemize}

A Type~IV tensor has a negative contraction along some null direction
and therefore violates the null energy condition
\cite{hawking1973,martinmoruno2017lnp,martinmoruno2018core}.  The violating
direction is exhibited in closed form for the momentum-sourced case below.  The defective
null cases Type~II and Type~III are \emph{not} decided by that contraction;
Appendix~\ref{app:slemma} gives tests that also apply to these types.

In an orthonormal frame adapted to the slice normal $n$, write
$\rho_n=T(n,n)$, $j_i=-T(n,e_i)$, and $S_{ij}=T(e_i,e_j)$.
For $j\ne0$, let $\hat\jmath=j/|j|$ and $S_\parallel=S(\hat\jmath,\hat\jmath)$.
Type~II can occur on a continuous transition between Types~I and~IV.  Under the
splitting hypothesis, $\hat\jmath$ is a principal axis of $S$ and the momentum
plane decouples.  Its block is
$A=\bigl(\begin{smallmatrix}-\rho_n & j\\ -j & S_\parallel\end{smallmatrix}\bigr)$,
which is self-adjoint with respect to $\eta_2=\mathrm{diag}(-1,1)$; here $j$
denotes the momentum component along $\hat\jmath$.
With $\Delta=(\rho_n+S_\parallel)^2-4j^2$, the type is
\begin{equation}
  \Delta>0 \;\Rightarrow\; \text{I},\qquad
  \Delta<0 \;\Rightarrow\; \text{IV},\qquad
  \Delta=0,\ j\neq0 \;\Rightarrow\; \text{II},\qquad
  \Delta=0,\ j=0 \;\Rightarrow\; \text{I (degenerate)} .
  \label{eq:momentum-trichotomy}
\end{equation}
At $\Delta=0$ with $j\neq0$ the repeated eigenvalue $(S_\parallel-\rho_n)/2$ leaves
$A-\lambda I$ of rank one with the single null eigendirection $(1,\pm1)$:
a defective null $J_2$, hence Type~II.  At $\Delta=0$ with $j=0$ the block reduces to $-\rho_n
\mathbb{1}$ and stays diagonalizable.  Since $\Delta$ is continuous, every path from
$\Delta>0$ to $\Delta<0$ at fixed $j\neq0$ crosses the Type~II locus, so Types~I
and~IV do not exhaust the possibilities.  Types~II and~III are unstable under
small perturbations of $T^a{}_b$, while Type~IV and nondegenerate Type~I are
stable~\cite{martinmoruno2018core}.  The Type-II locus has codimension one in
the momentum-block parameter space away from $j=0$.
Its preimage in a spacetime need not have zero volume, and a finite grid need
not sample a transition locus.
Type~III is excluded under the splitting hypothesis: in a
true $2\oplus2$ split the momentum block admits Jordan blocks of size at most two
and the transverse block is Euclidean-self-adjoint, hence diagonalizable, so no $J_3$
can be assembled; where the split fails, in the transverse channel of
Corollary~\ref{cor:conformal}, it is not.  An explicit
Lorentz-self-adjoint example is $T_{ab}=\ell_a m_b+m_a\ell_b$ with $\ell^a=(1,1,0,0)$
null, $m^a=(0,0,1,0)$ unit spacelike and $\ell\cdot m=0$, whose mixed form
\begin{equation}
  T^a{}_b=
  \begin{pmatrix}
    0 & 0 & 1 & 0\\
    0 & 0 & 1 & 0\\
    -1 & 1 & 0 & 0\\
    0 & 0 & 0 & 0
  \end{pmatrix}
  \label{eq:typeIII-example}
\end{equation}
satisfies $(T^a{}_b)^3=0$ with $(T^a{}_b)^2\ne0$.  Its eigenvalues all vanish, its
eigenspace is spanned by the null $\ell^a$ and the spacelike $(0,0,0,1)$, and the
Jordan form is $J_3\oplus J_1$: Segre type $[3,1]$, Hawking--Ellis Type~III.  Its
Eulerian data are $\rho_n=0$, $\hat\jmath^a=(0,0,1,0)$, $|j|=1$ and
$S_\parallel=0$ with $S_{xy}=1$, so $\hat\jmath$ is not a principal axis of $S$ and
$\Delta=-4<0$: the momentum discriminant would report Type~IV at a point that is
Type~III.  The splitting hypothesis cannot be omitted from the Type~III exclusion
or the $\Delta<0$ criterion in \eqref{eq:momentum-trichotomy}.  Gergely's $2+1+1$ analysis of the proper kinetic gravity braiding
tensor gives the same discriminant trichotomy for a different
stress-energy~\cite{gergely2026braiding}.  The numerical classifier alone does not establish that a sampled near-degenerate
point lies on the exact locus.

The Hawking--Ellis label can change across $\Delta=0$.  In a fixed tetrad the
multiplier margin $\max_\sigma\lambda_{\min}(\widehat T+\sigma\eta)$ is
half the normalized null-cone minimum (Corollary~\ref{cor:margin}) and is
$1$-Lipschitz in the operator norm of $\widehat T$: Weyl's eigenvalue inequality
bounds the change for each $\sigma$, and taking the maximum preserves that bound.
Its sign determines the NEC
verdict, which can change at zero margin.  Theorem~\ref{thm:lmi} applies
independently of the type.  We use the Type-II convention
$\widehat T=\bigl(\begin{smallmatrix}\mu+f&-f\\-f&-\mu+f\end{smallmatrix}\bigr)
\oplus[p_2]\oplus[p_3]$ in an orthonormal tetrad.
For example, $(\mu,f,p_2,p_3)=(0,1,-2,0)$ has a positive null part but null-cone minimum
$-4/3$ (Appendix~\ref{app:slemma}).  Positive null dust
$T_{ab}=\Phi k_ak_b$, $\Phi>0$, instead satisfies all four conditions:
$T(u,u)=\Phi(k\!\cdot\!u)^2\ge0$, its trace vanishes, and its flux is future null
or zero.

The mixed tensor $T^a{}_b = g^{ac}T_{cb}$ is generally non-symmetric, so
\warpax{} uses a general eigensolver; eigenvalues count as real when
$|\mathrm{Im}\,\lambda_i| < \epsilon s$ ($\epsilon=10^{-10}$,
$s=\max(|\mathrm{Re}\,\lambda_j|,\|T^a{}_b\|_\infty)$)
(Section~\ref{sec:scope} compares selected Type-IV labels using three
solver approaches).  Scaling the tolerance with the tensor preserves the test under
$T\to cT$ for nonzero $c$.  A unit floor would instead impose an absolute
cutoff on weak stress tensors and could suppress a nonreal eigenvalue pair.  The vacuum shortcut requires $T^a{}_b=0$ and assigns Type~I with
$\rho=p_i=0$.  A small real part of the spectrum is insufficient: pure
momentum flux has eigenvalues $\pm iq$, and a nonzero nilpotent tensor has
an entirely zero spectrum.  The separate WarpShell diagnostic in
Appendix~\ref{app:slemma} includes near-vacuum labels, which are counted
separately there.
For Type~I, all energy conditions reduce to algebraic
inequalities in the rest-frame density $\rho$ and principal pressures $p_1,p_2,p_3$:
\begin{alignat}{2}
  \text{NEC:} & \quad \rho + p_i \geq 0
    & \quad & \forall\, i\,,
    \label{eq:nec-type1} \\
  \text{WEC:} & \quad \rho \geq 0 \;\text{and}\;
    \rho + p_i \geq 0
    & \quad & \forall\, i\,,
    \label{eq:wec-type1} \\
  \text{SEC:} & \quad \rho + p_i \geq 0 \;\;\forall\, i
    \;\;\text{and}\;\;
    \rho + {\textstyle\sum_i} p_i \geq 0\,,
    & &
    \label{eq:sec-type1} \\
  \text{DEC:} & \quad \rho \geq |p_i|
    & \quad & \forall\, i\,.
    \label{eq:dec-type1}
\end{alignat}
These inequalities are \emph{exact}: for Type~I stress-energy a point
satisfies (or violates) an energy condition if and only if the
corresponding eigenvalue inequality holds, whichever observer
is chosen.  The classification decision itself is tolerance-based in floating
point (Section~\ref{sec:scope}), so the numerical verdict depends on
eigenvalue accuracy, degeneracy tolerances, and the timelike eigenvector test.
\paragraph{Deciding outside Type~I.}
Decompose $T$ in the Eulerian frame into the Eulerian energy density
$\rho_n = T_{ab}\,n^a n^b$ (equal to the Type-I rest-frame density $\rho$
when the momentum vanishes), the momentum density
$j^i = -h^i{}_a\,T^a{}_b\,n^b$ and
the spatial stress $S_{ij} = h_i{}^a h_j{}^b\,T_{ab}$, with
$h^a{}_b = \delta^a{}_b + n^a n_b$ the projector onto the slice.  In the timelike
plane spanned by $n^a$ and the momentum direction $\hat{\jmath}^a$ the mixed tensor
is the block
$\left(\begin{smallmatrix} -\rho_n & |j| \\ -|j| & S_\parallel \end{smallmatrix}\right)$,
$S_\parallel = S(\hat{\jmath},\hat{\jmath})$, whose eigenvalues
$\tfrac12(S_\parallel-\rho_n) \pm \tfrac12\sqrt{\Delta}$ are complex exactly when the
momentum density dominates,
\begin{equation}
  \Delta = (\rho_n + S_\parallel)^2 - 4|j|^2 < 0,
  \label{eq:discriminant}
\end{equation}
that is $2|j| > |\rho_n + S_\parallel|$.  Such a pair makes the point Type~IV
when $\hat{\jmath}^a$ is a
principal axis of $S_{ij}$,
and one of the Eulerian null vectors
$k^a = n^a \pm \hat{\jmath}^a$ then witnesses the violation,
\begin{equation}
  T_{ab}\,k^a k^b = \rho_n + S_\parallel - 2|j| < 0,
  \label{eq:null-witness}
\end{equation}
with no rapidity and no optimizer.  A Type-IV point whose complex pair lies outside
that plane, for example one opened by a conformal spatial factor, still has
no causal eigenvector, so the classification guarantees its null-energy violation
even where this witness is nonnegative.  At Type~II a nonnegative value of the same
contraction does not establish the NEC: it probes the momentum plane only, leaving a
transverse violation undetected, and on the canonical block
$(\mu,f,p_2,p_3)=(0,1,-2,0)$ it is exactly zero while the null-cone minimum is $-4/3$.
Type~II is therefore decided by Appendix~\ref{app:slemma}, not by
\eqref{eq:null-witness}.  The discriminant
\eqref{eq:discriminant} is compared with the full eigensolver in
Section~\ref{sec:shift-vorticity}.

\paragraph{Observer interpretation of the Type-I criteria.}
\label{par:per-observer}
Let $\{e_0^a,e_i^a\}$ be an orthonormal eigenframe at a Type-I point,
with $T(e_0,e_0)=\rho$ and $T(e_i,e_j)=p_i\delta_{ij}$.
For $u^a=\cosh\zeta\,e_0^a+\sinh\zeta\,\hat s^a$,
\begin{equation*}
 T(u,u)=\rho+\sinh^2\zeta\,[\rho+p(\hat s)],\qquad
 p(\hat s)=\sum_i p_i(\hat s\cdot e_i)^2.
\end{equation*}
Nonnegativity for every rapidity and direction requires and suffices for
$\rho\ge0$ and $\rho+p_i\ge0$, the WEC criteria.
For a normalized null vector $k=e_0+\hat s$, $T(k,k)=\rho+p(\hat s)$;
the NEC therefore requires no separate sign condition on $\rho$.
Applying the timelike test to $T_{ab}-\tfrac12Tg_{ab}$ gives the SEC,
while future-causal energy flux gives $\rho\ge|p_i|$ for the DEC
\cite{hawking1973,martinmoruno2017lnp}.  The rest frame $e_0$ need not
coincide with the Eulerian normal $n$.

\paragraph{The momentum discriminant.}
\label{par:discriminant}
The complex pair of \eqref{eq:discriminant} has imaginary part
$\tfrac12\sqrt{-\Delta}$.

At a point with $j=0$, $T^a{}_b$ is block diagonal, $(-\rho)\oplus S$ with $S$ symmetric, so its spectrum
is real and the point is Type~I. For a flat-slice unit-lapse drive whose shift
is $\beta^i=v_s F^i(\mathbf{x}-v_st\mathbf e_x)$ with fixed smooth $F$, the scalings
$\rho_n,S_{ij}\propto v_s^2$ and $|j|\propto v_s$ are exact for the Einstein
tensor rather than leading order, as follows from the $3+1$ equations.
The extrinsic curvature $K_{ij}=\partial_{(i}\beta_{j)}$ is linear in $v_s$; on a
flat slice the Hamiltonian constraint $16\pi\rho_n=K^2-K_{ij}K^{ij}$ is quadratic in
$K$, and the momentum constraint $8\pi j^i=D_j(K^{ij}-\gamma^{ij}K)$ is linear in it.
In the evolution equation for $K_{ij}$, with $\alpha=1$ and
$\gamma_{ij}=\delta_{ij}$, every surviving term is quadratic in $K$ or is
$\pounds_\beta K$, while $\partial_tK_{ij}=-v_s\partial_xK_{ij}$, so
$S_{ij}-\tfrac12\delta_{ij}(S-\rho_n)$ is quadratic as well and with $\rho_n$ so is
$S_{ij}$.  When $j=0$, this also implies
$\min_i(\rho+p_i)=-C v_s^2$ at a fixed comoving point, or after taking an infimum
over a fixed comoving wall domain.  Here $C$ is independent of speed; $C>0$
requires a null violation in that domain.  With $j\ne0$, diagonalization also
involves the speed-linear momentum terms and the quadratic margin law does not
follow. Under the splitting hypothesis, a complex pair occurs wherever the
momentum density dominates,
$2|j|>|\rho_n+S_\parallel|$; writing $|j|=d\,v_s$ and $\rho_n+S_\parallel=a\,v_s^2$
this is $\Delta=v_s^2(a^2 v_s^2-4d^2)<0$, so at a point with $d>0$, $a\ne0$ and $v_s>0$ the
momentum-plane pair closes at the exact speed
$v_s=2d/|a|$, above which the decoupled tensor is Type~I
(Section~\ref{sec:invariant-benchmark}). The null-witness deficit
$\rho_n+S_\parallel-2|j|$ then has a linear-in-$|j|$ term absent in the
irrotational case.

Because $a$ and $d$ are independent of $v_s$, the field $v_\star(\mathbf{x})=2d/|a|$ is a
property of the fixed shift profile.  Where
$a\neq0$ the momentum channel is open for $v_s<v_\star(\mathbf{x})$
and closed above it; where $a=0$ with $d>0$ it never closes.
For $d=0$ the momentum channel is absent, and we set $v_\star=0$.
Its predicted volume fraction is
\begin{equation}
  F_{\rm mom}(v_s)=\frac{\operatorname{vol}\{\,v_\star>v_s\,\}}{\operatorname{vol}(\text{wall})}.
  \label{eq:closing-fraction}
\end{equation}
$F_{\rm mom}$ bounds the Type-IV fraction only under the splitting hypothesis, with
$\hat\jmath^a$ a principal axis of $S_{ij}$: $\Delta<0$ is sufficient for Type~IV
under that hypothesis and not otherwise, since switching on
the transverse coupling of Corollary~\ref{cor:conformal} can restore a real spectrum to the
momentum pair.  On the block of that corollary the cubic
discriminant is even in $m$ with leading coefficient $+4m^6$, so at fixed
$\Delta<0$ and nonvanishing minor there is a finite $|m|$ past which $D>0$: three
real eigenvalues, Type~I, with $\Delta<0$.  At $(\rho_n,S_x,S_y,|j|)=(1,1,0,2)$,
where $\Delta=-12$, the discriminant is negative at $m=3$ and positive at $m=4$.
Thus $F_{\rm mom}\le F_{\rm IV}$ requires the splitting premise.
Section~\ref{sec:closing-speed} compares aggregate numerical fractions; that
comparison does not check the premise pointwise.

\begin{corollary}[Type-IV boundary with transverse stress coupling]
\label{cor:conformal}
Let the momentum direction $\hat{\jmath}^a$ fail to be a principal axis of the spatial
stress $S_{ij}$, so a transverse coupling $m$ enters the block along $n^a$,
$\hat{\jmath}^a$, and the coupled transverse axis $\hat t^a$; and let the remaining
spatial axis $\hat z^a$ decouple from that block, $S(\hat{\jmath},\hat z)=S(\hat
t,\hat z)=0$. The decoupling is a hypothesis and not a normalization: without it the
characteristic polynomial of $T^a{}_b$ is quartic and the cubic discriminant of the
block does not decide the type. Its cubic discriminant
$D(\rho_n,S_x,S_y,|j|,m)$ is even in $m$ and satisfies
$D|_{m=0}=(\mathrm{minor})^2\,\Delta$, so a transverse complex pair (Type~IV) opens
where $D<0$, its boundary lying on the algebraic surface $D=0$, and it can open with
$\Delta\ge0$: the momentum discriminant of Section~\ref{sec:classification} is
sufficient under the splitting hypothesis, but is not necessary for
Type~IV without that hypothesis. On the witness family $\rho_n=S_x=1$, $|j|=\tfrac12$, $S_y=0$ the transverse
pair is complex on the band $m_\star<|m|<m_2$ between the two positive roots of $D$,
with $m_\star\approx0.782$ and $m_2\approx2.015$, and at $m=1$ an Eulerian null vector
of the orthonormal frame exhibits the null violation directly.
\end{corollary}

\noindent The witness family is a representative stress ratio chosen to exhibit the
transverse channel, not a sampled Van~den~Broeck wall point; the overall scale is
irrelevant to the type.
\begin{proof}
In the orthonormal Eulerian frame $(n^a,\hat{\jmath}^a,\hat t^a,\hat z^a)$, with $\hat t^a$
the transverse axis coupled to $\hat{\jmath}^a$ by $m=S(\hat{\jmath},\hat t)$ and $\hat z^a$
decoupled (real eigenvalue $S_z$), the mixed tensor on the $(n,\hat{\jmath},\hat t)$ block is
\begin{equation*}
  T^a{}_b\big|_{\rm block}=\begin{pmatrix} -\rho_n & |j| & 0\\ -|j| & S_x & m\\ 0 & m & S_y\end{pmatrix},
  \qquad S_x=S_\parallel.
\end{equation*}
Its characteristic polynomial is $\lambda^3+c_2\lambda^2+c_1\lambda+c_0$, with
\begin{align*}
 c_2&=\rho_n-S_x-S_y,\\
 c_1&=S_xS_y-m^2-\rho_n(S_x+S_y)+|j|^2,\\
 c_0&=\rho_n(S_xS_y-m^2)-|j|^2S_y.
\end{align*}
The cubic discriminant is
\begin{equation*}
  D=18c_2c_1c_0-4c_2^3c_0+c_2^2c_1^2-4c_1^3-27c_0^2 ,
\end{equation*}
a polynomial even in $m$. The surface $D=0$ also contains real-eigenvalue
degeneracies away from the Type-I/Type-IV transition. Setting $m=0$ gives
$D|_{m=0}=(\mathrm{minor})^2\,\Delta$ with
$\mathrm{minor}=S_xS_y-S_y^2-|j|^2+(S_x-S_y)\rho_n$ and $\Delta=(\rho_n+S_x)^2-4|j|^2$,
recovering Section~\ref{sec:classification} on a principal axis. For the witness
family $\rho_n=S_x=1$, $|j|=\tfrac12$, $S_y=0$, the discriminant
becomes, in $u=m^2$, $D=4u^3-18u^2+\tfrac{27}{4}u+\tfrac{27}{16}$, whose two positive
roots $u_\star\approx0.611$ and $u_2\approx4.059$ bound the Type-IV band
$m_\star=\sqrt{u_\star}\approx0.782<|m|<\sqrt{u_2}=m_2\approx2.015$; for $m$ in this band the
block carries a complex pair, and at $m=1$ the null vector $k^a=n^a+(\hat{\jmath}^a-\hat t^a)/\sqrt2$ gives
$T_{ab}\,k^ak^b=\rho_n+\tfrac12 S_x+\tfrac12 S_y-\sqrt2\,|j|-m=\tfrac12-\tfrac1{\sqrt2}<0$.
\end{proof}

The momentum and transverse reductions require their respective decoupling
hypotheses.  Table~\ref{tab:shift_vorticity}(b) compares the measured imaginary
part with $\tfrac12\sqrt{\max(0,-\Delta)}$ at a selected off-axis point; agreement or
disagreement at that point does not establish a wall-wide mechanism.
The numerical type labels use the full four-dimensional tensor.
For Garattini--Zatrimaylov, the type is established analytically as Type~I (Section~\ref{sec:construction-exoticity}).

The algebraic slack, $\min_i(\rho+p_i)$ and its WEC/SEC/DEC analogs,
is the smallest left-hand side of the corresponding Type-I inequalities.
Its sign determines satisfaction, and its magnitude measures proximity to the
boundary in the eigenframe.  The rapidity-capped
minimum $\min_{\zeta \leq \zeta_{\max}} T_{ab}\,u^a u^b$ depends on the chosen
normal and cap.  A negative value witnesses WEC violation; a nonnegative capped
value does not prove WEC.  At a NEC-violating point the unrestricted timelike
density tends to $-\infty$ along suitable boosts.

\subsection{Pointwise matrix criteria}
\label{sec:lmi}

The energy conditions are inequalities on causal vectors, so they can be
tested without first assigning an algebraic type.  Fixing the time component
in an orthonormal frame reduces the causal directions to a sphere or ball.
The S-lemma, a standard result
for quadratic inequalities, replaces the test over this set by a matrix
inequality with one scalar multiplier.  Appendix~\ref{app:slemma} gives the
proof and explicit certificates.

Let $\{n^a,e_i^a\}$ be an orthonormal tetrad with $n^a$ the unit slice normal, and
decompose the stress-energy in it as $\rho_n=T_{ab}n^an^b$, $b_i=-T_{ab}n^ae_i^b$,
$S_{ij}=T_{ab}e_i^ae_j^b$; thus $b_i=j_i$ in this frame.
Every future unit timelike $u^a$ is
$u^a=\gamma(n^a+w^ie_i^a)$ with $|w|<1$ and $\gamma=(1-|w|^2)^{-1/2}$, and every
future nonzero null $k^a$ is $k^a=c\,(n^a+s^ie_i^a)$ with $|s|=1$, $c>0$.  Then
\begin{equation}
  T_{ab}u^au^b=\gamma^2 q(w),\qquad
  q(w)=\rho_n-2\,b\!\cdot\!w+w^{\mathsf T}Sw ,
  \label{eq:lmi-quadratic}
\end{equation}
with $\gamma^2>0$.  This positive prefactor leaves the sign unchanged.
The conditions therefore reduce to quadratic inequalities on compact sets:
the closed unit ball for the timelike conditions, by continuity, and the
unit sphere for the null condition.  The S-lemma decides these
inequalities without a Hawking--Ellis classification.

Write $\eta=\mathrm{diag}(-1,1,1,1)$ for the tetrad-frame metric and
$\widehat T$ for the tetrad-frame component matrix of $T_{ab}$, so that
$\widehat T_{00}=\rho_n$, $\widehat T_{0i}=-b_i$, $\widehat T_{ij}=S_{ij}$.  Define
the one-parameter family
\begin{equation}
  M(\sigma)\;=\;\widehat T+\sigma\,\eta\;=\;
  \begin{pmatrix} \rho_n-\sigma & -b^{\mathsf T}\\[2pt] -b & S+\sigma\,\mathbb{1}\end{pmatrix},
  \label{eq:lmi}
\end{equation}
the tetrad-frame matrix of $T_{ab}+\sigma g_{ab}$.  The notation
$M\succeq0$ means positive semidefinite: $z^{\mathsf T}Mz\ge0$ for every real
vector $z$.

\begin{theorem}[Decision by a linear matrix inequality]
\label{thm:lmi}
Let $T_{ab}$ be a real symmetric tensor at a point of a four-dimensional
Lorentzian spacetime.  Define $\Theta_{ab}=T_{ab}-\tfrac12 T g_{ab}$ with
$T=g^{ab}T_{ab}$, and let
$(T^2)_{ab}=T_{ac}g^{cd}T_{db}$.  Then, with no hypothesis on the Hawking--Ellis
type of $T^a{}_b$,
\begin{enumerate}
  \item the NEC holds at the point $\iff$ $M(\sigma)\succeq0$ for some
    $\sigma\in\mathbb{R}$;
  \item the WEC holds $\iff$ $M(\sigma)\succeq0$ for some $\sigma\ge0$;
  \item the SEC holds $\iff$ the WEC holds for $\Theta_{ab}$, i.e.\ $\iff$
    $\widehat\Theta+\sigma\eta\succeq0$ for some $\sigma\ge0$;
  \item the DEC holds $\iff$ the WEC holds \emph{and} the WEC holds for
    $-(T^2)_{ab}$.
\end{enumerate}
\end{theorem}

At Type~I points we use the eigenvalue inequalities
\eqref{eq:nec-type1}--\eqref{eq:dec-type1}, which are already necessary and
sufficient there.  At other types we use Theorem~\ref{thm:lmi}; a negative
Eulerian null contraction \eqref{eq:null-witness} also gives a direct
violation witness.  Where the metric and curvature are evaluated with interval
arithmetic, these tests include bounds on the numerical error in $T_{ab}$
(Appendix~\ref{app:interval-lmi}).

\subsection{Observer parameterization}
\label{sec:observer}

To display observer-dependent contractions, we boost the Eulerian normal
for timelike vectors and parameterize the celestial sphere for null vectors.
The LMI criterion itself requires no observer search.

\paragraph{Timelike observers.}
Any future-directed unit timelike vector $u^a$ at a point can be written as a Lorentz
boost of the Eulerian normal $n^a$:
\begin{equation}
  u^a = \cosh\zeta\; n^a + \sinh\zeta\; \hat{s}^a(\theta, \phi)\,,
  \label{eq:rapidity}
\end{equation}
where $\zeta \geq 0$ is the rapidity (boost magnitude),
$(\theta, \phi)$ specify the spatial boost direction, and
$\hat{s}^a(\theta, \phi)$ is a unit spatial vector constructed from an
orthonormal tetrad adapted to the ADM decomposition.  The tetrad is
built by applying Gram--Schmidt orthogonalization to the coordinate
basis vectors projected into the spatial hypersurface orthogonal to
$n^a$.

\paragraph{Null directions.}
A future-directed null direction has the normalized representative
\begin{equation}
  k^a = n^a + \hat{s}^a(\theta, \phi)\,,
  \label{eq:null-observer}
\end{equation}
where $(\theta, \phi)$ parameterize the null direction on the celestial
sphere.  This fixes $k_a n^a = -1$, removing the null-vector rescaling
ambiguity.  The resulting NEC margin $T_{ab}\, k^a k^b$ is defined relative to $n^a$.
Its minimum is a two-dimensional optimization problem over
$(\theta, \phi)$.

In the rapidity parameterization, $\zeta = 0$
recovers the Eulerian observer, so the optimization always
includes the Eulerian result as a baseline.  The WEC, SEC, and DEC involve
three-dimensional optimization over $(\zeta, \theta, \phi)$, while the
NEC requires only $(\theta, \phi)$.  The true timelike
observer manifold has unbounded rapidity ($\zeta \in [0, \infty)$);
in practice, we optimize over $\zeta \in [0, \zeta_{\max}]$ with a
fixed cap $\zeta_{\max} = 5$ (corresponding to $\gamma \approx 74$;
see Section~\ref{sec:optimization}).

\subsection{Observer optimization}
\label{sec:optimization}

The observer diagnostic searches for the minimum
\begin{equation}
  \text{margin}^* = \min_{(\zeta, \theta, \phi) \in \mathcal{O}}
    \; m(T_{ab}, u^a(\zeta, \theta, \phi))\,,
  \label{eq:opt}
\end{equation}
where $m$ is the signed margin function: $m \geq 0$ means the
condition holds for the tested vector, and $m < 0$ indicates a violation.  For the NEC,
$m = T_{ab}\, k^a k^b$ with optimization over $(\theta, \phi)$ only.
For the DEC, the diagnostic at each spacetime point returns
\begin{equation}
  m_{\text{DEC}} = \min\bigl(
    m_{\text{flux}},\;
    m_{\text{future}},\;
    m_{\text{WEC}}
  \bigr)\,,
  \label{eq:dec-diagnostic}
\end{equation}
where $m_{\text{flux}} = -g_{ab}\, J^a J^b$ (nonnegative when the
energy flux $J^a = -T^a{}_b\, u^b$ is causal),
$m_{\text{future}} = -J_a n^a$ (positive when the causal $J^a$ is
future-directed, with $n^a$ the future-pointing unit normal), and
$m_{\text{WEC}} = T_{ab}\, u^a u^b$ (nonnegative when the energy
density is nonnegative).  A negative value of any component
indicates a DEC violation for that observer.
The three sub-terms have mixed dimensions, so we use them only for sign-based
detection ($m_{\text{DEC}}^*<0$), not as a combined severity.  For Type~I
stress-energy the algebraic condition (equation~\ref{eq:dec-type1}) guarantees
$J^a$ causal and future-directed for every future-directed $u^a$, so the
algebraic slack $\min_i(\rho-|p_i|)$ is the primary DEC result; the three-term
optimizer diagnostic is used only at the residual non-Type~I points.

Optimization uses BFGS \cite{nocedal2006} in Optimistix
\cite{optimistix2024} over an unconstrained boost vector
$\mathbf z\in\mathbb{R}^3$, with direction $\mathbf z/|\mathbf z|$
and rapidity $\zeta=\zeta_{\max}\tanh(|\mathbf z|/\zeta_{\max})$.
At $\mathbf z=0$ the observer is Eulerian.  The soft cap is
$\zeta_{\max}=5$ ($\gamma\approx 74$); null directions use a stereographic
parameterization of $S^2$.  Multi-start
($N_{\text{starts}}$ seeds: Eulerian, six axis-aligned,
and random) is evaluated in parallel, batched over the starts.

\paragraph{Closed-form worst observer at Type~I points.}
At a Type~I point the rest-frame minimum has a closed form.  In the eigenframe
$T^a{}_b=\mathrm{diag}(-\rho,p_1,p_2,p_3)$, an observer boosted from the
rest frame with rapidity $\zeta$ along principal axis $i$ measures
\begin{equation}
  \rho_{\rm obs}(\zeta)
    = \rho\cosh^2\zeta + p_i\sinh^2\zeta
    = \rho + (\rho+p_i)\sinh^2\zeta\,,
  \label{eq:worst-observer}
\end{equation}
monotone in $\sinh^2\zeta$.  If $\rho+p_i\ge0$ for every axis the rest
frame already measures the least energy density; if
$\rho+p_{i^*}<0$ for $i^*=\arg\min_i(\rho+p_i)$, the worst boost lies
along the principal eigenvector $e_{i^*}$ of the most-violating
pressure, every observer boosted along that axis past the threshold rapidity
$\sinh^2\zeta_{\rm th} = \rho/|\rho+p_{i^*}|$ (for $\rho>0$) measures
negative energy density, and $\rho_{\rm obs}$ is unbounded below as
$\zeta\to\infty$.  The closed form therefore supplies the worst
direction, the violation threshold, and the asymptotic
sign; no finite observer attains the unrestricted infimum.  Here $\zeta$ is
measured from the Type-I rest frame.  It agrees with the Eulerian rapidity
used by the numerical cap only when that rest frame coincides with $n^a$.

\subsection{Geodesic integration}
\label{sec:geodesics}

\warpax{} integrates the geodesic equation
\begin{equation}
  \frac{\dd^2 x^\mu}{\dd \lambda^2}
    + \Gamma^\mu{}_{\alpha\beta}\,
      \frac{\dd x^\alpha}{\dd \lambda}\,
      \frac{\dd x^\beta}{\dd \lambda} = 0
  \label{eq:geodesic}
\end{equation}
with either of two schemes.  The toolkit provides the Tsitouras 5(4) Runge--Kutta method
(Tsit5) \cite{tsitouras2011} via Diffrax \cite{kidger2022diffrax}, with adaptive
step-size control (PID controller, $\mathrm{rtol}=\mathrm{atol}=10^{-10}$).  Every
geodesic finite-segment value reported here uses the extended-phase-space
symplectic scheme of Section~\ref{sec:averaged-quantum}, with a fourth-order
Yoshida composition.  A numerical tangent can drift from the null cone;
these runs monitor that drift and step-refine each selected ray.
The coordinate-path diagnostic samples an explicitly prescribed path and does
not solve the geodesic equation.

\section{Results: Energy conditions}
\label{sec:results-ec}

We evaluate energy conditions on $50^3$ spatial grids for each metric,
with the time coordinate fixed at $t = 0$ (for constant-velocity
bubbles whose metric depends on $x - v_s t$, the $t = 0$ slice is
representative up to spatial translation).  The grid extends $\pm 5R_b$ for metrics with $R_b = 1$
(Alcubierre, Van~den~Broeck, Nat\'ario) and $\pm 3R_b$ for
the large-radius Rodal metric ($R_b = 100$), centered on
the bubble.

Automatic differentiation evaluates derivatives of the implemented metric
without finite-difference truncation error.  Roundoff and profile
regularization still affect the computed tensor.  Spatial grid refinement
tests the sampling of wall extrema and thresholded fractions; geodesic
step refinement is a separate test.

On the
uniform grids above the Alcubierre/VdB/Nat\'ario walls span only ${\sim}1.35$
cells (10--90\% width $\Delta r\approx 2.2/\sigma$;
Table~\ref{tab:params}(b)). The central benchmark
(Section~\ref{sec:invariant-benchmark}) therefore evaluates all four metrics
at matched parameters ($R_b=1$, $\sigma=8$) on one family of wall-clustered
graded grids that resolve the $10$--$90\%$ wall by $5.9$ to $8.9$ cells on a radial
crossing ($\Delta r/\Delta x_{\rm wall}$), a factor of $4.4$ to $6.6$ over the $1.35$ cells of
the uniform grid: box $\pm 3R_b$, clustering strength $a_\star=2$, three levels
$N=80,100,120$ giving $5.9$, $7.6$, $8.9$ cells at wall spacings $\Delta x_{\rm
wall}=0.0465, 0.0362, 0.0309$, every level clearing the four-cell resolution
criterion.  (The spacing is the \emph{clustered} one at the wall, not the nominal
$6R_b/N$.)
We refine sampled wall extrema off-grid within their selected basins and
report the grid dependence of the thresholded volume fractions
(Section~\ref{sec:convergence}).  These calculations measure numerical stability;
the global enclosures are obtained separately.
Table~\ref{tab:params}(a) lists the per-metric native parameters;
unless a caption states otherwise, every figure and table below uses them,
and matched-parameter results use $R_b=1$, $\sigma=8$,
$v_s=0.5$ on the wall-resolved graded grids just described.

\begin{table}[tbp]
  \centering
  \footnotesize
  \caption{\emph{(a)}~Per-metric parameters and grid configuration; all grids are
    $50^3$ at $t = 0$, centered on the bubble.
    \protect\protect\linebreak\vspace*{3pt} This table lists native per-metric parameters for the uniform-grid diagnostics; the primary cross-metric comparison instead evaluates every retained metric at common parameters ($R_b = 1$, $\sigma = 8$) on smoothly graded compact grids (Section~\ref{sec:invariant-benchmark}).

    $^{\dagger}$Additional Van~den~Broeck parameters: $\tilde{R}=1$,
    $\alpha_{\text{vdb}}=0.5$, $\sigma_B=8$.
    \emph{(b)}~Wall resolution analysis on the same grids.  Transition width is the
    10--90\% threshold distance of each metric's shape function, root-solved from
    the profile rather than taken from the large-$\sigma R_b$ limit
    $\Delta r = 2\,\mathrm{atanh}(0.8)/\sigma$ ($72.53$ against $73.24$ at the
    native Rodal parameters).  The \emph{Convergence} column records the
    resolution-support tier on the native uniform grid (Stability-only: minimum
    margin stable under refinement; -- : no study).}
  \label{tab:params}

  \emph{(a)}\;\begin{tabular}{l c c c c}
    \toprule
    Metric & $R_b$ & $\sigma$ & Domain & $\zeta_{\max}$ \\
    \midrule
    Alcubierre      & 1   & 8.0  & $(\pm 5)^3$   & 5 \\
    Van~den~Broeck$^{\dagger}$  & 1   & 8.0  & $(\pm 5)^3$   & 5 \\
    Nat\'ario       & 1   & 8.0  & $(\pm 5)^3$   & 5 \\
    Rodal           & 100 & 0.03 & $(\pm 300)^3$ & 5 \\
    Schwarzschild   & $M\!=\!1$ & -- & $(\pm 20)^3$ & 5 \\
    \bottomrule
  \end{tabular}

  \medskip
  \emph{(b)}\;\begin{tabular}{@{}lccccc@{}}
    \toprule
    Metric & Wall width & $\Delta x$ & Cells & Resolved & Convergence \\
    \midrule
    Alcubierre & 0.2747 & 0.2041 & 1.35 & No & Stability-only \\
    Nat\'ario & 0.2747 & 0.2041 & 1.35 & No & -- \\
    Van~Den~Broeck & 0.2747 & 0.2041 & 1.35 & No & -- \\
    Rodal & 72.5335 & 12.2449 & 5.92 & Yes & Stability-only \\
    \midrule
    Schwarzschild & -- & 0.82 & -- & -- & -- \\
    \bottomrule
\end{tabular}

\end{table}

\noindent A wall spanning fewer than four cells \emph{on a
uniform grid} is excluded from the uniform-grid comparison (this removes the
native thin-shell and soliton cases); the matched benchmark instead resolves the
wall on the graded grids above, whose observed resolution dependence is
reported in Section~\ref{sec:convergence}.
Cell counts are quoted as the \emph{worst} case over every crossing the axis makes
of the wall, measured on the grid actually used: a closed wall is crossed twice per
axis, and reporting the better crossing would overstate the resolution.

For each grid point we compute the Eulerian baseline: the ADM normal for
WEC/SEC/DEC and six axis-aligned null directions for NEC.  The all-observer
criteria are the Type-I eigenvalue inequalities and, at other types, the
LMI tests of Theorem~\ref{thm:lmi}.  A negative axial null contraction
\eqref{eq:null-witness} supplies a direct violation witness.  All numerical
verdicts remain subject to the accuracy limits of Section~\ref{sec:scope}.
Away from the wall the stress can be small without being Type~I.  An
unrestricted grid fraction also depends on the size of the exterior box; the primary cross-metric comparisons therefore use wall-restricted fractions
(Section~\ref{sec:invariant-benchmark}).

\subsection{Energy conditions across the luminal transition}
\label{sec:invariant-benchmark}

The matched-parameter benchmark ($R_b=1$, $\sigma=8$) compares the
all-observer energy conditions and algebraic types on graded bubble-wall
grids across the luminal transition.  Exact Hawking--Ellis types and
Type-I eigenvalue slacks are Lorentz invariants; their numerical estimates
retain the accuracy limits of Section~\ref{sec:scope}.  Their volume fractions refer to the specified constant-$t$ slice
and wall mask.  None requires $\partial_t$ to be timelike, so the
analysis is identically defined below, at, and above $v_s=1$.

\paragraph{Wall algebraic types.}
Table~\ref{tab:invariant_benchmark} restricts to the active wall
($f \in [0.1, 0.9]$, proper-volume weighted).  The Alcubierre, Nat\'ario, and
Van~den~Broeck walls are \emph{Type-IV dominated}
($81$--$99\%$ at $v_s = 0.5$), with Type-I fractions $1.3\%$, $16.6\%$, $18.9\%$: an established Type-IV tensor has no timelike eigenvector and violates all
four energy conditions~\cite{martinmoruno2018core}.  The Rodal samples are
$100\%$ Type~I, consistent with the exact result for its ideal profile.
Their energy-condition tests use the Type-I eigenvalue inequalities.
The wall type fractions come from the standalone benchmark grid; their
resolution dependence is reported in Section~\ref{sec:convergence}.

\begin{table}[tbp]
  \centering
  \footnotesize
  \caption{Matched-parameter ($R_b=1$, $\sigma=8$, $v_s=0.5$),
    wall-restricted, proper-volume-weighted invariant benchmark on the
    wall-resolved graded grid ($N=100$, ${\sim}7.6$ cells across the 10-90\%
    wall).  The invariant peak NEC deficit $\min(\rho+p_i)$ is over Type~I points,
    refined off-grid; for the Type-IV-dominated walls it samples only the residual
    real-eigenvalue points and is not a severity bound.
    ``Missed by Eulerian'' is the fraction of all-observer violations at which the
    Eulerian frame reports none, at $v_s=0.5$.}
  \label{tab:invariant_benchmark}
  \begin{tabular}{@{}l cc c ccc@{}}
  \toprule
  & Type~I & Type~IV & $\min(\rho+p_i)$ & \multicolumn{3}{c}{Missed by Eulerian (\%)} \\
  \cmidrule(lr){5-7}
  Metric & (\%) & (\%) & (Type~I) & WEC & NEC & DEC \\
  \midrule
  Alcubierre & 1.3 & 98.7 & -0.158 & 0.0 & 0.0 & 0.0 \\
  Natário & 16.6 & 83.4 & -1.955 & 0.0 & 0.0 & 0.0 \\
  Van den Broeck & 18.9 & 81.1 & -0.521 & 54.8 & 14.1 & 51.3 \\
  Rodal & 100.0 & 0.0 & -0.194 & 72.6 & 13.2 & 74.0 \\
  \bottomrule
\end{tabular}

\end{table}

\paragraph{The all-observer reading of the positive-energy result.}
Rodal recently put forward a warp drive with predominantly positive
invariant energy density and global Type~I structure, reported at
$v/c=1$~\cite{rodal2026}.  Lemma~\ref{lem:rodal} derives that structure for the analytic profile.
The sampled comparison illustrates how observer choice enters: the Eulerian frame reports no violation
at ${\approx}73\%$ of the wall points where a boosted observer sees a
weak-energy violation, and at ${\approx}74\%$ for the dominant energy
condition (Table~\ref{tab:invariant_benchmark}).  The reference verdicts use the Type-I eigenvalue inequalities.
The fractions are sampled and proper-volume weighted, with a range of about
two percentage points across the resolution ladder.
Global Type~I fixes the algebraic character, not energy-condition satisfaction: the
irrotational wall still has a residual null (hence weak, dominant, and strong)
violation, reduced but nonzero in Rodal's own construction~\cite{rodal2026} and
consistent with the drive-by-drive NEC violation we report below.
For the ideal irrotational Rodal profile, $\rho_n=\rho$ because $j=0$.
This equality of densities does not equate a single-observer test with the
all-observer energy conditions.  At a Type-IV point there is no timelike
eigenvector and hence no rest-frame energy density.

\paragraph{Type and severity across $v_s=1$.}
Table~\ref{tab:velocity_type} and Figure~\ref{fig:velocity_type} track the
wall type structure from deep subluminal to superluminal.  Across the sampled transition, Rodal remains $100\%$ Type~I to $v_s = 2.5$,
while the Alcubierre and Nat\'ario walls remain Type-IV dominated
(Van~den~Broeck only above its transition, below).  The Type-IV fraction decreases overall with speed for Alcubierre and Nat\'ario,
with small sampled fluctuations. The wall Type~I fraction
reaches $24.9\%$
(Alcubierre), $43.8\%$ (Nat\'ario), and $23.2\%$ (Van~den~Broeck) at
$v_s=2.5$ (Figure~\ref{fig:velocity_type}), even as the
sampled Type-I NEC slacks deepen.
Van~den~Broeck initially shows the opposite trend, passing through a velocity-driven
Type-I$\,\to\,$Type-IV transition that takes its wall Type-IV fraction from
$7.1\%$ at $v_s=0.1$ to a peak of $92.6\%$ at $v_s=0.7$, its wall Type-IV fraction
crossing $50\%$ at $v_s\approx\vdbTransitionVS$; the decline sets in only past that
peak.  For the everywhere-Type-I Rodal geometry the
invariant NEC and DEC margins (Figure~\ref{fig:velocity_type}b) deepen
monotonically from $\min(\rho+p_i)\approx-0.0077$ at $v_s=0.1$ to $-4.799$
at $v_s=2.5$, following the quadratic law $\min(\rho+p_i)=-\rodalNECscaling\,v_s^2$
($R^2=1$ over the full sweep $v_s\in[0.1,2.5]$; the subluminal-only fit of
Table~\ref{tab:speed_scaling}(b) returns the same coefficient to the displayed precision,
$|\min(\rho+p_i)|/v_s^2=0.7679023$ at every sampled speed).
For the ideal Rodal profile, irrotationality gives $j=0$ and the quadratic law
holds exactly.  The separately regularized implementation has the small residual
curl discussed in Appendix~\ref{app:rodal_ablation}; its sampled agreement with
that law is numerical evidence.  The strictly negative coefficient is
consistent with the Santiago--Schuster--Visser no-go~\cite{santiago2021}, whose null
conclusion holds under the subsidiary conditions its authors state; the $v_s^2$ speed
dependence is our calculation for this family, not a feature of the theorem.

\begin{table}[tbp]
  \centering
  \footnotesize
  \caption{Wall-restricted Hawking--Ellis Type~I / Type~IV fractions (\%,
    proper-volume weighted) at subluminal, luminal, and superluminal speeds.
    Computed frame-independently from $T^a{}_b$ on the wall-resolved benchmark grid
    ($N=100$, wall resolved to $7.6$ cells); Type-IV labels cross-checked
    (Section~\ref{sec:scope}); the full sweep to $v_s=2.5$ is shown in
    Figure~\ref{fig:velocity_type}.}
  \label{tab:velocity_type}
  \begin{tabular}{@{}l cc cc cc@{}}
  \toprule
  & \multicolumn{2}{c}{$v_s=0.5$ (sub)} & \multicolumn{2}{c}{$v_s=1.0$ (luminal)} & \multicolumn{2}{c}{$v_s=2.0$ (super)} \\
  \cmidrule(lr){2-3}\cmidrule(lr){4-5}\cmidrule(lr){6-7}
  Metric & Type~I & Type~IV & Type~I & Type~IV & Type~I & Type~IV \\
  \midrule
  Alcubierre & 1.3 & 98.7 & 5.7 & 94.3 & 16.7 & 83.3 \\
  Natário & 16.6 & 83.4 & 28.4 & 71.6 & 41.3 & 58.7 \\
  Van den Broeck & 18.9 & 81.1 & 21.0 & 79.0 & 25.4 & 74.6 \\
  Rodal & 100.0 & 0.0 & 100.0 & 0.0 & 100.0 & 0.0 \\
  \bottomrule
\end{tabular}

\end{table}

\begin{figure}[tbp]
  \centering
  \includegraphics[width=\textwidth]{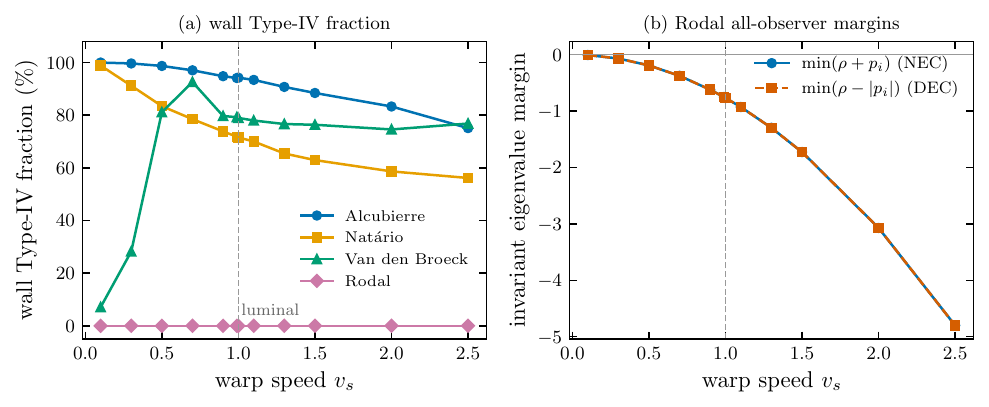}
  \caption{Type structure and energy-condition slacks at $R_b=1$, $\sigma=8$
    on the $N=100$ wall grid; $v_s=1$ is dashed.
    (a)~Proper-volume-weighted Type-IV fractions on $t=0$,
    $f\in[0.1,0.9]$. These aggregate trends do not test the local splitting
    hypotheses. Rodal's ideal profile is Type~I by Lemma~\ref{lem:rodal}.
    (b)~Rodal's sampled NEC and DEC minima coincide at the displayed precision
    and follow the quadratic law for its ideal irrotational profile
    (equation~\ref{eq:wall-deficit}). Type~I does not imply energy-condition
    satisfaction. Lines connect sampled speeds.}
  \label{fig:velocity_type}
\end{figure}

\subsection{Shift vorticity and algebraic type}
\label{sec:shift-vorticity}

The benchmark shows different algebraic types even at matched speeds and
wall widths.  To interpret this difference, we relate the Eulerian momentum
to spatial derivatives of the shift.  This identifies a sufficient condition
for Type~I and the additional stress information needed when momentum is nonzero.

For flat unit-lapse slices, the Eulerian momentum density obeys
$8\pi|j|=\tfrac12|\nabla\times(\nabla\times\beta)|$
(Lemma~\ref{lem:jcurl}), where $\beta_i=g_{0i}$ is the ADM shift.
Under the smoothness, entire-domain and bounded-vorticity hypotheses of
Lemma~\ref{lem:gradient}, it vanishes everywhere exactly for gradients plus
rigid rotation. This characterizes momentum, not algebraic type: the
discriminant also involves $\rho_n$ and $S$, and its type implication requires
the splitting assumption of
Section~\ref{sec:classification}.  On the slice we split the covariant gradient of the
shift into expansion, shear, and vorticity parts, as a fluid
velocity field is decomposed,
\begin{equation}
  D_i\beta_j = \tfrac{1}{3}\,\theta_\beta\,\gamma_{ij}
             + \sigma^\beta_{ij} + \omega^\beta_{ij},
  \qquad
  \omega^\beta_{ij} = D_{[i}\beta_{j]} = \partial_{[i}\beta_{j]},
  \label{eq:shift-decomp}
\end{equation}
where the vorticity part is one half of the exterior derivative of the shift one-form,
so it follows from first metric derivatives and transforms tensorially under
time-independent spatial coordinate changes. The wall summaries use
$\langle X\rangle_W=(\int_W X\,dV)/(\int_W dV)$ on $t=0$,
$W=\{0.1\le f\le0.9\}$, and the dimensionless vorticity fraction
\[
  \mathcal{R}_\omega =
  \frac{\langle\omega_\beta^2\rangle_W}
  {\langle\theta_\beta^2/3+\sigma_\beta^2+\omega_\beta^2\rangle_W}.
\]
Thus the wall integrals are formed before taking the ratio. Expansion and shear
fractions use the same denominator. For a fixed comoving wall, a speed-linear
shift and a speed-independent spatial metric make these fractions independent
of nonzero $v_s$, provided the denominator is nonzero.
Here $\sigma_\beta^2=\sigma^\beta_{ij}\sigma_\beta^{ij}$ and
$\omega_\beta^2=\omega^\beta_{ij}\omega_\beta^{ij}$ denote full contractions.

Table~\ref{tab:shift_vorticity}(a) and Figure~\ref{fig:shift-vorticity} report
the wall-restricted, proper-volume-weighted decomposition.  Rodal's analytic
shift is irrotational by Lemma~\ref{lem:rodal}; the implemented regularization
gives $\mathcal{R}_\omega = \rodalVortFrac$.  The sufficient condition for
Type~I is classical: an irrotational shift on a unit-lapse, flat slice gives
$j=0$ and a symmetric spatial stress
block~\cite{santiago2021,martinmoruno2018core}.  The remaining drives have
nonzero shift vorticity
($\mathcal{R}_\omega = \alcubierreVortFrac$ Alcubierre,
$\vdbVortFrac$ Van~den~Broeck, $\natarioVortFrac$ Nat\'ario) and all have
sampled Type-IV wall regions.  Nat\'ario has \emph{zero}
shift expansion (its computed expansion fraction is $\natarioExpFrac$)
and the largest vorticity fraction in this sample.

The association between $\mathcal{R}_\omega$ and type is empirical within this
family.  Lemma~\ref{lem:gradient} concerns vanishing momentum: $j$ depends on
$\nabla\times\omega$, whereas the decoupled momentum discriminant also involves
the energy density and parallel stress.  The transverse channel of
Corollary~\ref{cor:conformal} requires the full spatial stress.
Rigid rotation illustrates the distinction: it has $\mathcal{R}_\omega=1$,
$\nabla\times\omega=0$ and $j=0$, and describes Minkowski spacetime in rotating
coordinates.  The exact curvature coefficient in
Section~\ref{sec:construction-exoticity} likewise depends on derivatives of
$\omega$.  The vorticity fraction alone supplies no criterion for type,
energy-condition violation or curvature scaling.

The curl of the shift vorticity fixes $j$, while the Hamiltonian constraint
fixes $\rho_n$ from the symmetric shift gradient.  The spatial stress also
contains evolution and Lie-derivative terms.  Where $\hat\jmath$ is a principal
axis of $S$, the momentum plane decouples and
$2|j|>|\rho_n+S_\parallel|$ gives a Type-IV pair.  Without this splitting,
$\Delta<0$ still supplies an axial NEC-violation witness but does not determine
the algebraic type.  Table~\ref{tab:shift_vorticity}(b) compares the full
eigensolver's imaginary part with the decoupled-plane prediction at a specified
wall point for each drive.

Under the splitting hypothesis with $\Delta<0$, the positive imaginary part is
\begin{equation}
  \operatorname{Im}\lambda = \tfrac12\sqrt{-\Delta}
    = \tfrac12\sqrt{4|j|^2 - (\rho_n+S_\parallel)^2}\,,
  \label{eq:imag-discriminant}
\end{equation}
which gives $\operatorname{Im}\lambda\approx|j|$ where the momentum dominates.  On a flat unit-lapse
slice the Eulerian momentum density is fixed by the shift through an exact identity.

\begin{lemma}[Momentum density is the curl of the shift vorticity]
\label{lem:jcurl}
Let $\beta_i=g_{0i}$ be a $C^2$ shift on a region with $\alpha=1$ and a
time-independent flat spatial metric $\gamma_{ij}=\delta_{ij}$.
Write $h^a{}_b=\delta^a{}_b+n^an_b$ for the spatial projector.
The Eulerian momentum density $j^i=-h^i{}_a T^a{}_b n^b$ obeys,
in magnitude and to all orders in $\beta$,
\begin{equation}
  8\pi\,|j|=\tfrac12\,|\nabla\times(\nabla\times\beta)|=\tfrac12\,|\nabla\times\omega|,
  \qquad \omega=\nabla\times\beta.
  \label{eq:j-curl}
\end{equation}
Consequently an irrotational shift ($\omega=0$) and a spatially uniform rotation
($\nabla\times\omega=0$) both have $j=0$.  The momentum is determined by
$\nabla\times\omega$, rather than by the vorticity alone.
\end{lemma}
\begin{proof}
With $K_{ij}=-\tfrac12\pounds_n\gamma_{ij}$ and the $+\beta$ convention,
$K_{ij}=\tfrac12(\partial_i\beta_j+\partial_j\beta_i)$ and
$K=\nabla\!\cdot\!\beta$.  The momentum constraint gives
\begin{equation*}
  8\pi j_i=\partial_jK_{ij}-\partial_iK
  =\tfrac12\bigl[\nabla^2\beta-\nabla(\nabla\!\cdot\!\beta)\bigr]_i
  =-\tfrac12\bigl[\nabla\times(\nabla\times\beta)\bigr]_i .
\end{equation*}
Taking the Euclidean norm proves the identity.

\end{proof}

\noindent Santiago--Schuster--Visser~\cite{santiago2021}, equation~(4.8), give this identity as
$f_i=\tfrac1{16\pi}[\nabla\times(\nabla\times v)]_i$ in their flow convention.
They also establish that an irrotational shift gives Type-I stress-energy.
To identify all shifts with vanishing momentum, a local condition is not
enough: the following converse uses the entire slice and bounded vorticity.
It includes rigid rotation.

\begin{lemma}[Shifts with vanishing momentum]
\label{lem:gradient}
For $\alpha=1$ and a time-independent spatial metric $\gamma_{ij}=\delta_{ij}$,
let the shift $\beta$ be smooth on $\mathbb{R}^3$ with
bounded vorticity $\omega=\nabla\times\beta$.  Then
\begin{equation}
  j\equiv0
  \iff
  \beta=\nabla\phi+\Omega\times\mathbf{x}
  \quad\text{for a scalar }\phi\text{ and a constant vector }\Omega.
  \label{eq:gradient-iff}
\end{equation}
A stationary shift consisting of rigid rotation alone gives $K_{ij}=0$ and a flat spacetime.
If the vorticity tends to zero at infinity, then $\Omega=0$.
\end{lemma}
\begin{proof}
If $\beta$ has that form then $\omega=2\Omega$ is constant, so $\nabla\times\omega=0$
and $j\equiv0$ by Lemma~\ref{lem:jcurl}.  Conversely $j\equiv0$ gives
$\nabla\times\omega=0$, while $\nabla\!\cdot\!\omega=0$ holds for any curl, so
$\nabla^2\omega=\nabla(\nabla\!\cdot\!\omega)-\nabla\times(\nabla\times\omega)=0$:
each Cartesian component of $\omega$ is harmonic on $\mathbb{R}^3$ and bounded, hence
constant by Liouville's theorem.  Writing $\omega=2\Omega$ and using
$\nabla\times(\Omega\times\mathbf{x})=2\Omega$ leaves
$\nabla\times(\beta-\Omega\times\mathbf{x})=0$ on a simply connected domain, so that
difference is a gradient.  Finally
$\partial_i(\Omega\times\mathbf{x})_j=\epsilon_{jki}\Omega_k$ is antisymmetric in
$(i,j)$, so $K_{ij}=\partial_{(i}\beta_{j)}$ vanishes on it.
\end{proof}

\noindent The controlled differential-rotation experiment uses
$\beta=c(-y,x,0)e^{-(x^2+y^2+z^2)/2}$ at $(x,y,z)=(0.5,0.5,0)$,
with $c\in[0,0.4]$.  Fitting its imaginary eigenvalue against
$\omega=(\omega_{ij}\omega^{ij})^{1/2}=|\nabla\times\beta|/\sqrt2$
gives the empirical slope $\kappa\approx\kappaVorticity$.
This coefficient belongs to that profile and sample point.
Under the momentum-plane splitting hypothesis, equation~\eqref{eq:imag-discriminant}
instead gives the exact relation
$\operatorname{Im}\lambda=|j|\sqrt{1-(\rho_n+S_\parallel)^2/(4|j|^2)}$.
For $|\rho_n+S_\parallel|\ll2|j|$, its leading fractional difference from
$|j|$ is $(\rho_n+S_\parallel)^2/(8|j|^2)$.
An irrotational flat-slice unit-lapse shift has $j=0$ and Type~I
stress-energy, as used in Rodal's construction~\cite{rodal2026}.

\begin{lemma}[Type~I stress-energy for the irrotational profiles]
\label{lem:rodal}
Write the Rodal shift in the orthonormal spherical frame of
\eqref{eq:rodal-shift}, $\hat\beta_r=-v_sF(r_s)\cos\theta$,
$\hat\beta_\theta=v_sG(r_s)\sin\theta$.  Its only possibly nonvanishing curl
component is the azimuthal one, and
\begin{equation}
  (\nabla\times\beta)_\phi=\frac{v_s\sin\theta}{r_s}\bigl[(r_sG)'-F\bigr].
  \label{eq:rodal-curl}
\end{equation}
For $v_s\neq0$, the shift is curl-free on each sphere exactly when
$F=(r_sG)'$.  The profiles
\eqref{eq:shape-tanh} and \eqref{eq:rodal-G} satisfy that relation identically in
$(r_s,\sigma,R_b)$.  Hence $\nabla\times\beta\equiv0$, so $j\equiv0$ by
Lemma~\ref{lem:jcurl}, so $T^a{}_b$ is block diagonal $(-\rho_n)\oplus S_{ij}$
with $S$ symmetric: its spectrum is real and the slice normal is a timelike
eigenvector.  The drive is Hawking--Ellis Type~I at every point of every slice,
at every $v_s$.  The same holds for the Garattini--Zatrimaylov bubble on its
matching condition, whose shift is a sum of two radial gradients
(Appendix~\ref{app:metric_definitions}) and therefore has $j=0$ by Lemma~\ref{lem:jcurl}.
\end{lemma}
\begin{proof}
With $\beta_\phi=0$ and axial symmetry, the radial and polar curl
components vanish termwise and the azimuthal one is
$r_s^{-1}[\partial_{r_s}(r_s\hat\beta_\theta)-\partial_\theta\hat\beta_r]$, which
is \eqref{eq:rodal-curl}.  Writing
$\Delta_\sigma(r_s)=\ln\cosh\sigma(r_s-R_b)-\ln\cosh\sigma(r_s+R_b)$,
equation~\eqref{eq:rodal-G} is $G=-\Delta_\sigma/(2r_s\sigma\tanh\sigma R_b)$, so
$(r_sG)'=-\Delta_\sigma'/(2\sigma\tanh\sigma R_b)
=[\tanh\sigma(r_s+R_b)-\tanh\sigma(r_s-R_b)]/(2\tanh\sigma R_b)=f(r_s)=F$.  The
block-diagonal step is the $j=0$ case of the momentum decomposition of
Section~\ref{sec:classification}.
\end{proof}

\noindent The conclusion for the irrotational drive is Rodal's own
Proposition~2~\cite{rodal2026}, proved there under the same assumptions; the
relation $F=(r_sG)'$ is the irrotationality condition his construction is built to
satisfy, written in the conventions used here.  In this form the label is a
continuum property.  The numerical Rodal implementation uses small-radius
regularizations (Appendix~\ref{app:metric_definitions}); its residual curl
is a numerical effect; the exact zero applies to the analytic profile.

\noindent The harmonic-gradient implication of
Barzegar, Buchert and Vigneron~\cite{barzegar2026classification} applies under
its R-Warp hypotheses, including the stated asymptotic assumptions.  Rodal's
potential is nonharmonic, so that implication gives no conclusion for this drive.

The sign of $\Delta$ determines the momentum-plane type under the splitting
hypothesis.  The reported Type-IV labels are taken from the full
four-dimensional eigensolver, not from the sign of $\Delta$.
A separate unweighted count comparison uses $N=60$, clustering $a=1.2$,
$t=0$ and $v_s=0.5$: bounds $\pm3$ with $R_b=1$, $\sigma=8$ for the first
three drives, and bounds $\pm300$ with $R_b=100$, $\sigma=0.03$ for Rodal.
The indicator $\Delta<0$ agrees with the Type-IV label at $100\%$ of Rodal
wall points, $97.5\%$ of Alcubierre's, $88.4\%$ of Nat\'ario's and $76.2\%$ of
Van~den~Broeck's, with full-grid agreements $98.9\%$, $99.7\%$
and $81.6\%$ for Alcubierre, Nat\'ario and Van~den~Broeck, respectively.  The sign
$\Delta<0$ supplies a NEC-violation witness without deciding Type~IV.  Numerical
disagreement alone does not identify its cause or test the splitting hypothesis.
Nat\'ario and Van~den~Broeck have $368$ and $752$ wall points labeled Type~IV
with $\Delta\ge0$, respectively.
The sampled Type-IV fraction varies with speed and wall geometry: only
$7\%$ of the Van~den~Broeck wall is Type-IV at $v_s=0.1$, before its
Type-I$\,\to\,$Type-IV transition.  Maeda~\cite{maeda2020nonstatic}
showed that in a static spacetime solving the Einstein equations the stress-energy
is Type~I wherever the hypersurface-orthogonal Killing field is timelike.
Related restrictions hold on suitable Killing horizons and symmetry
axes~\cite{martinmoruno2021backreaction}.  An established Type-IV tensor
excludes staticity under those hypotheses; laboratory-coordinate time dependence
alone would not establish this.

\begin{table}[tbp]
  \centering
  \footnotesize
  \caption{\emph{(a)}~Wall-restricted, proper-volume-weighted decomposition of the
    ADM shift gradient (Eq.~\ref{eq:shift-decomp}) into expansion, shear, and
    vorticity fractions, with the resulting wall Type-IV range over
    $v_s\in[0.1,2.5]$.  The vorticity fraction $\mathcal{R}_\omega$ is
    independent of $v_s$.  \emph{(b)}~The momentum-density discriminant at the
    matched wall sample ($v_s=0.5$, off-axis wall point).  The ratio
    $2|j|/|\rho_n+S_\parallel|>1$ is equivalent to $\Delta<0$ and gives
    Type~IV under the momentum-plane splitting hypothesis.  The measured
    $\operatorname{Im}\lambda$ is set beside the momentum-plane prediction
    $\tfrac12\sqrt{\max(0,-\Delta)}$ at this sample; the wall-maximum is in
    Table~\ref{tab:diagnostic_convergence}.}
  \label{tab:shift_vorticity}

  \emph{(a)}\;\begin{tabular}{@{}l ccc c@{}}
  \toprule
  & \multicolumn{3}{c}{Shift-gradient fraction} & Wall Type~IV \\
  \cmidrule(lr){2-4}
  Metric & Expansion & Shear & Vorticity $\mathcal{R}_\omega$ & range (\%) \\
  \midrule
  Alcubierre & 0.111 & 0.556 & 0.333 & 75--100 \\
  Natário & 0.000 & 0.492 & 0.508 & 56--99 \\
  Van den Broeck & 0.193 & 0.367 & 0.440 & 7--93 \\
  Rodal & 0.485 & 0.515 & 0.000 & 0--0 \\
  \bottomrule
\end{tabular}

  \medskip
  \emph{(b)}\;\begin{tabular}{l c c c c}
  \toprule
  Metric & Wall type & $2|j|/|\rho_n+S_\parallel|$ & $\operatorname{Im}\lambda$ & $\tfrac12\sqrt{\max(0,-\Delta)}$ \\
  \midrule
  Rodal & I & 0.00 & 0.000 & 0.000 \\
  Van~den~Broeck & IV & 0.89 & 0.058 & 0.000 \\
  Alcubierre & IV & 16.09 & 0.079 & 0.079 \\
  Nat\'ario & IV & 6.89 & 0.578 & 0.585 \\
  \bottomrule
\end{tabular}

\end{table}

\begin{figure}[tbp]
  \centering
  \includegraphics[width=\columnwidth]{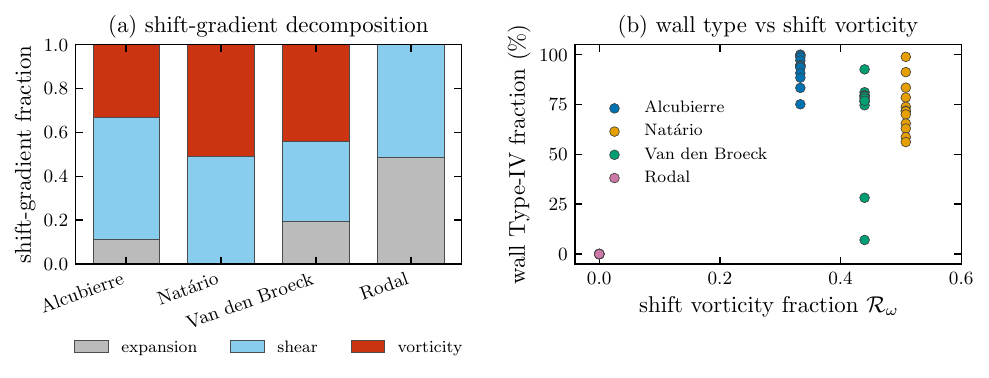}
  \caption{Shift-gradient decomposition and sampled type structure
    ($R_b=1$, $\sigma=8$, $t=0$, $f\in[0.1,0.9]$).
    (a)~Fractions of the wall-integrated gradient norm ($N=60$).
    Nat\'ario has zero expansion; Rodal's vorticity is below plot resolution,
    consistent with its ideal irrotational profile.
    (b)~Proper-volume-weighted Type-IV fractions ($N=100$) at the sampled
    speeds $v_s\in[0.1,2.5]$, paired with each drive's speed-independent
    $\mathcal{R}_\omega$. Alcubierre has a smaller $\mathcal{R}_\omega$ but a
    larger Type-IV fraction than Nat\'ario. The vorticity fraction alone does
    not determine type.}
  \label{fig:shift-vorticity}
\end{figure}

\subsection{Speed dependence of the momentum discriminant}
\label{sec:closing-speed}

The density and spatial stress scale quadratically with speed on the fixed
flat-slice profiles, while the momentum scales linearly.  Their competition
therefore changes with speed.  The discriminant expresses this change through
the local threshold $v_\star$, whose coefficients can be computed at one speed.

At every wall point we compute $a=(\rho_n+S_\parallel)/v_s^2$ and $d=|j|/v_s$ on
the wall-resolved benchmark grid, form the closing speed $v_\star=2d/|a|$, and
evaluate $F_{\rm mom}$ of \eqref{eq:closing-fraction} at the twelve speeds of the
velocity sweep.  One grid evaluation fixes the whole curve.

We compare the coefficients at $v_s=0.5$ and $v_s=1$; their relative
changes appear in Table~\ref{tab:closing_speed}. Flat unit-lapse slices with
the specified stationary or translating profile have the exact scaling of
Section~\ref{sec:classification}. Van~den~Broeck has a curved spatial metric,
so its single-speed field is included only as a numerical comparison.
For the ideal Rodal shift $j=0$, the momentum channel is absent.

Table~\ref{tab:closing_speed} compares $F_{\rm mom}$ with the sampled Type-IV
fraction. Its signed gap is a difference of aggregate fractions. The
calculation does not check principal-axis alignment or the pointwise
implication $\Delta<0\Rightarrow$ Type~IV at each speed. Consequently the gap
does not establish a bound or a volume assigned to a transverse mechanism.
The numerical closing field treats $|a|\le10^{-12}\max_{\rm wall}|a|$ as zero;
the exact criterion above distinguishes $a=0$ from every nonzero value.

\begin{table}[tbp]
  \centering
  \footnotesize
  \caption{The momentum channel as a speed-independent field.  $\delta a$ is the
    relative change in $a=(\rho_n+S_\parallel)/v_s^2$ between $v_s=0.5$ and
    $v_s=1$, a check of the scaling assumption.  $\tilde v_\star$ is
    the median finite closing speed over the wall.  $F_{\rm mom}$ (\%) is
    \eqref{eq:closing-fraction}, computed once at $v_s=1$ and read at each speed;
    ``gap'' is the measured wall Type-IV fraction minus $F_{\rm mom}$, in
    percentage points. Neither the sign of this gap nor coefficient agreement
    verifies the splitting hypothesis.}
  \label{tab:closing_speed}
  \begin{tabular}{@{}l cc cc cc cc cc@{}}
  \toprule
  & & & \multicolumn{2}{c}{$v_s=0.1$} & \multicolumn{2}{c}{$0.5$} & \multicolumn{2}{c}{$1.0$} & \multicolumn{2}{c}{$2.5$} \\
  \cmidrule(lr){4-5}\cmidrule(lr){6-7}\cmidrule(lr){8-9}\cmidrule(lr){10-11}
  Metric & $\delta a$ & $\tilde v_\star$ & $F_{\rm mom}$ & gap & $F_{\rm mom}$ & gap & $F_{\rm mom}$ & gap & $F_{\rm mom}$ & gap \\
  \midrule
  Alcubierre & 2e-15 & 2.48 & 99.9 & +0.1 & 95.5 & +3.3 & 84.6 & +9.6 & 48.8 & +26.3 \\
  Natário & 1e-14 & 1.24 & 98.4 & +0.4 & 73.8 & +9.6 & 57.8 & +13.7 & 31.5 & +24.7 \\
  Van den Broeck & 2e+00 & 1.74 & 100.0 & -92.9 & 99.0 & -17.9 & 90.8 & -11.9 & 35.2 & +41.6 \\
  Rodal & 0e+00 & 0.00 & 0.0 & +0.0 & 0.0 & +0.0 & 0.0 & +0.0 & 0.0 & +0.0 \\
  \bottomrule
\end{tabular}

\end{table}

\subsection{Single-frame and finite-observer diagnostics}
\label{sec:warpfactory}
\label{sec:individual}

The benchmark tests the energy conditions over all observers.  We now ask how
many of these violations a single observer or a finite observer sample misses,
using the same stress-energy calculation for both tests.

Two examples show the size of the effect.  For Alcubierre at $v_s=0.5$ the Eulerian
and all-observer NEC violation sets coincide on the sampled native grid
at the stated detection threshold (the sampled all-wall
NEC minimum is $-0.628$ over null directions and $-0.625$ for the six Eulerian
directions on the native uniform
grid, distinct from the Type-I invariant deficit $\min(\rho+p_i)=-0.158$ of
Table~\ref{tab:invariant_benchmark}).  The WEC contractions differ much more:
the minimum found at $\zeta_{\max}=5$ is ${\sim}{-}3460$, compared with
$-0.038$ for the Eulerian observer.  As Section~\ref{sec:optimization} explains,
this magnitude depends on the rapidity cap.  Van~den~Broeck, whose conformal factor acts on the spatial metric, is the only
drive with sampled Type-IV wall regions with missed violations across all four conditions on the
uniform-grid diagnostic (Table~\ref{tab:missed}); its SEC has the largest
conditional miss among its conditions ($f_{\text{miss}|\text{viol}}=22.6\%$;
Figure~\ref{fig:vdb-sec}), and for Alcubierre the SEC
is the only missed condition at all.

Figure~\ref{fig:fibonacci} compares a finite observer sample with the Type-I
algebraic criterion within the same implementation. It isolates the need for
an all-observer test without requiring an observer optimization for the reference
set.

We define the \emph{missed violation fraction} for a given energy
condition as
\begin{equation}
  f_{\text{miss}} = \frac{\#\{x \in \text{grid} :
    m_{\text{Euler}}(x) \geq 0 \;\wedge\; m_{\text{rob}}(x) < 0\}}
    {\#\{x \in \text{grid}\}}\,,
  \label{eq:fmiss}
\end{equation}
i.e.\ the fraction of grid points the Eulerian frame reports as satisfied
($m_{\text{Euler}}\geq 0$) but the observer-independent verdict marks violated (the exact algebraic slack
at Type~I, and at Type-IV points the null violation the label implies;
Section~\ref{sec:classification}).
The unconditional $f_{\text{miss}}$ (Table~\ref{tab:missed}) is a
\emph{domain-restricted diagnostic}: it scales with each metric's wall
thickness and grid volume and is \emph{not} a cross-metric severity measure.
The appropriate cross-metric statistic is the wall-restricted conditional miss
rate $f_{\text{miss}|\text{viol}}$, the fraction of robust-violated points the
Eulerian analysis misses, $(\text{missed})/(\text{total violated})$, which does not
dilute the failure rate over the satisfied points.  It is
reported on the matched benchmark
(Table~\ref{tab:matched_benchmark}(a)) and in
Section~\ref{sec:invariant-benchmark}; the matched and invariant computations, on
the same wall-resolved graded family, agree to within grid resolution (Rodal wall
DEC miss $73.6$--$74.0\%$).  Table~\ref{tab:missed} gives the
within-metric Eulerian-vs-robust comparison; note SEC misses are independent of
the others (the combination $\rho+\sum_i p_i$ can be negative while each
$\rho+p_i\geq 0$).

\begin{table}[tbp]
  \footnotesize
  \caption{Energy condition analysis at $v_s = 0.5$ on a $50^3$ grid
    ($\zeta_{\max} = 5$ throughout).
    \emph{Total}: percentage of grid points violating each condition on the robust
    margin of Section~\ref{sec:classification}.
    \emph{Missed}: percentage of grid points where the Eulerian analysis reports no
    violation the robust check finds, so the Eulerian violated
    fraction is the total minus the missed; $0.0\%$ means all violations are
    detected to the displayed precision, not that there are none.  Bold marks
    missed fractions $\geq 1\%$.}
  \label{tab:missed}
  \centering
  \begin{tabular}{@{}l cccc cccc@{}}
  \toprule
  & \multicolumn{4}{c}{Total violated (\%)}
  & \multicolumn{4}{c}{Missed by Eulerian (\%)} \\
  \cmidrule(lr){2-5} \cmidrule(lr){6-9}
  Metric & NEC & WEC & SEC & DEC
         & NEC & WEC & SEC & DEC \\
  \midrule
  Schwarzschild & 0.0 & 0.0 & 0.0 & 0.0
                 & 0.0 & 0.0 & 0.0 & 0.0 \\
  Alcubierre & 6.2 & 9.4 & 9.4 & 9.6
                 & 0.0 & 0.0 & \textbf{7.2} & 0.0 \\
  Van~Den~Broeck & 6.6 & 7.8 & 7.9 & 8.1
                 & 0.1 & 0.4 & \textbf{1.8} & 0.3 \\
  Nat\'ario & 7.7 & 12.0 & 12.0 & 12.2
                 & 0.0 & 0.0 & \textbf{3.2} & 0.0 \\
  Rodal & 87.0 & 87.0 & 87.0 & 99.9
                 & \textbf{1.6} & \textbf{15.6} & \textbf{28.0} & \textbf{28.5} \\
  \bottomrule
\end{tabular}

\end{table}

\begin{figure}[tbp]
  \centering
  \includegraphics[width=\textwidth]{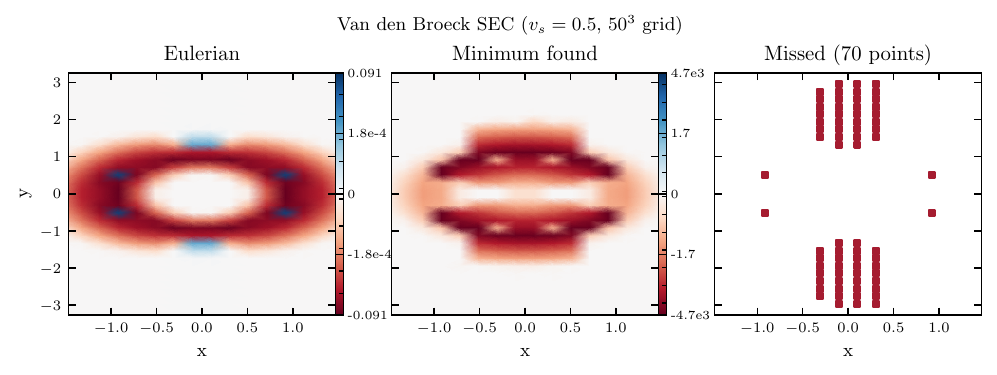}
  \caption{Van~den~Broeck SEC at $v_s=0.5$ on the $50^3$ native grid.
    Left: Eulerian contraction. Center: observer minimum found with
    $\zeta_{\max}=5$. Their symmetric logarithmic color scales differ;
    the capped magnitude is not an invariant severity.
    Right: all 70 missed sample points in the plane
    $z=5/49\approx0.102R_b$, marked in red without interpolation.
    At these points the Eulerian contraction is nonnegative and the boosted
    contraction is below $-10^{-10}$.
    The full-grid miss is $1.8\%$, or $22.6\%$ of detected violations
    (Table~\ref{tab:missed}); these percentages are not slice counts.}
  \label{fig:vdb-sec}
\end{figure}

\begin{figure}[tbp]
  \centering
  \includegraphics[width=0.7\textwidth]{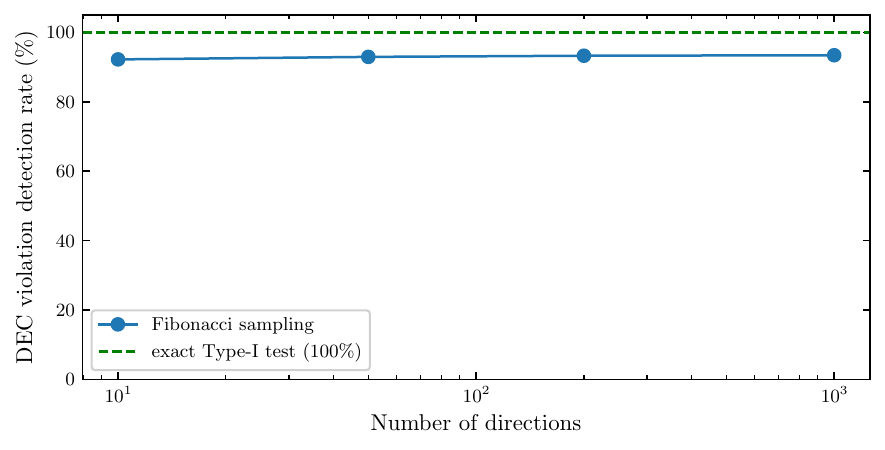}
  \caption{Finite-observer DEC detection for Rodal
    ($R_b=100$, $\sigma=0.03$, $v_s=0.5$; $25^3$ grid).
    Each Fibonacci direction uses ten rapidities in $[0,5]$.
    Detection rate (circles) reaches $93.4\%$ at $10^4$ observer samples per point.
    The dashed $100\%$ reference is the Type-I criterion on the
    $15\,624$ violating points of the $15\,625$ sampled;
    both diagnostics use threshold $-10^{-10}$.
    The residual miss combines angular and rapidity sampling, the cap, and
    threshold effects. It is not an isolated angular-sampling error.}
  \label{fig:fibonacci}
\end{figure}

\begin{table}[tbp]
  \footnotesize
  \caption{\textbf{Matched-parameter cross-metric benchmark.}
    \emph{(a)}~All retained metrics at common family parameters ($R_b = 1$,
    $\sigma = 8$) on identical compact bounds, at the finest resolution of the
    ladder of part~(b).  The reported quantity is the volume-weighted
    \emph{wall-restricted} conditional miss rate: Eulerian-satisfied but
    robust-violated points, as a fraction of robust-violated points within the wall
    $f \in [0.1, 0.9]$.  \emph{(b)}~Per-metric resolution validation of that
    benchmark: the wall-restricted DEC miss rate at three levels of the
    wall-resolved graded ladder ($N = 80, 100, 120$, giving $5.9$, $7.6$ and $8.9$
    cells across the 10--90\% wall) and the maximum absolute deviation from the three-grid mean, in
    percentage points.  With $\bar f$ the mean in percent, ``Stable'' means
    this deviation is at most $\max(0.5,0.05|\bar f|)$ percentage points.  The miss rate is a non-smooth thresholded fraction, so it
    has a stability spread and not a Richardson order.}
  \label{tab:matched_benchmark}
  \centering

  \emph{(a)}\;\begin{tabular}{@{}l c cccc@{}}
  \toprule
  & Wall & \multicolumn{4}{c}{Missed by Eulerian, wall-restricted (\%)} \\
  \cmidrule(lr){3-6}
  Metric & pts & NEC & WEC & SEC & DEC \\
  \midrule
  Alcubierre & 22688 & 0.0 & 0.0 & 18.7 & 0.0 \\
  Natário & 22688 & 0.0 & 0.0 & 0.0 & 0.0 \\
  Van den Broeck & 22688 & 16.7 & 57.4 & 8.6 & 53.1 \\
  Rodal & 22688 & 12.5 & 72.2 & 19.5 & 73.6 \\
  \bottomrule
\end{tabular}

  \medskip
  \emph{(b)}\;\begin{tabular}{@{}l ccc c c@{}}
  \toprule
  & \multicolumn{3}{c}{Wall DEC miss (\%) at $N=$} & Max dev & Stable \\
  \cmidrule(lr){2-4}
  Metric & 80 & 100 & 120 & (pp) & \\
  \midrule
  Alcubierre & 0.0 & 0.0 & 0.0 & 0.00 & \checkmark \\
  Natário & 0.0 & 0.0 & 0.0 & 0.00 & \checkmark \\
  Van den Broeck & 56.1 & 51.3 & 53.1 & 2.58 & \checkmark \\
  Rodal & 75.7 & 74.0 & 73.6 & 1.27 & \checkmark \\
  \bottomrule
\end{tabular}

\end{table}

\paragraph{Velocity dependence of the Eulerian comparison.}

Section~\ref{sec:invariant-benchmark} reports the sampled type structure
across the luminal transition.  In the subluminal Alcubierre observer
comparison, the Eulerian frame detects all sampled NEC/WEC violations to the
displayed precision (missed $=0.0\%$).  The capped contractions become more
negative as speed increases.

For Rodal at its native parameters
($R_b=100$, $\sigma=0.03$) the Eulerian frame misses
${\sim}\rodalDECmissVSfive\%$ of DEC and
${\sim}\rodalWECmissVSfive\%$ of WEC violations as a fraction of the full $50^3$ grid
(Table~\ref{tab:convergence_panel}(c)). This unconditional grid
fraction is not the wall-restricted conditional miss fraction of the
matched-parameter benchmark (${\approx}74\%$ DEC, ${\approx}73\%$ WEC;
Section~\ref{sec:invariant-benchmark}), which conditions on the violating points
within the active wall.  The discrepancy is systematic across the velocity range
and strongest for DEC and WEC.

\paragraph{Finite observer sampling and the algebraic criterion.}
WarpFactory \cite{warpfactory2024,warpfactory_toolkit2024,warpfactory_docs2024}
takes the minimum energy-condition value over a finite sample of observer
directions and velocities at each grid point.  The retained Fibonacci-lattice
comparison uses Rodal at $R_b=100$,
$\sigma=0.03$ and $v_s=0.5$, with ten rapidities between zero and five.
At $10^3$ directions per point it detects $93.4\%$ of the $15\,624$ points
whose Type-I algebraic DEC slack is below $-10^{-10}$
(Figure~\ref{fig:fibonacci}); the minimum algebraic slack is
$-2.0446\times10^{-6}$. The sampled diagnostic also uses threshold
$-10^{-10}$. This measures the shortfall of that finite observer search;
the calculation does not separate the effects of angular sampling,
rapidity sampling, the cap and numerical thresholds. The algebraic criterion
supplies the reference set without an observer search.

\subsection{Averaged and quantum conditions}
\label{sec:averaged-quantum}

The pointwise classical results do not determine averaged conditions or
semiclassical bounds.  We report finite-segment null-energy integrals and a
Ford--Roman reference estimate along a coordinate-static wall worldline
(Table~\ref{tab:averaged_quantum}, Figure~\ref{fig:averaged_quantum}).

\paragraph{Finite-segment null-energy diagnostics.}
Table~\ref{tab:averaged_quantum} reports two finite-window calculations at
$R_b=1$, $\sigma=8$, $v_s=0.5$.  In the coordinate-path diagnostic we sample
$(t,x,y,z)=(\lambda,-8+\lambda,b,0)$ and contract $T$ with a null-projected
direction at each node.  That projected vector need not be the path tangent.
This calculation has no complete-geodesic ANEC interpretation.
Its parameter interval is $[0,43.94822954822955]$ in units of $R_b$, with
$80$ initial impact parameters in $[10^{-3},2.5]$ and the selected-path
refinement described below.

For the geodesic diagnostic we integrate the canonical Hamiltonian
$H=\tfrac12g^{ab}p_ap_b$ with Tao's extended-phase-space method and a fourth-order
Yoshida composition~\cite{tao2016symplectic,christian2021fantasy}.
The initial position is $(0,-8R_b,b,0)$ with initial spatial direction $+\hat x$.
We fix the affine scale separately for every ray by
\begin{equation}
  -g_{ab}k^an^b=1
  \label{eq:anec-normalization}
\end{equation}
at its initial point.  Under a constant rescaling $k\mapsto c k$, $c>0$,
with the same geometric endpoints, $\int T(k,k)\,\dd\lambda$ scales by $c$.
Its magnitude therefore depends on this normalization.
Let $s>0$ be the initial spatial scale factor that enforces this normalization.
The affine bounds are $[0,L/s]$, with $L=43.5849609375$ for Alcubierre,
Nat\'ario, and Rodal, and $L=47.5927734375$ for Van~den~Broeck.
For Alcubierre, Nat\'ario, and Van~den~Broeck, $s=1$; for Rodal it is
$1.0000000009460461$ on the near-axis ray and $1.0054270723616685$ on the
minimum-found ray.  All lengths are in units of $R_b$.

Each metric's integration window is selected from a probe ray, with a factor-two
margin after its outward crossing of $r_s=3R_b$.  This is a window-selection
rule only.  The shape-function bound outside that radius does not bound the
stress contraction.  In particular, Rodal's angular profile has
$G(r_s)\sim R_b/[r_s\tanh(\sigma R_b)]$, so its shift and stress have algebraic
tails despite the radial $\tanh$ profile.  Neither omitted affine tail is bounded
here.  The reported integral is
$I_{\rm seg}=\int_{\lambda_-}^{\lambda_+}T_{ab}k^ak^b\,\dd\lambda$,
not an integral over a complete geodesic.

The impact-parameter search starts with $50$ equally spaced values on
$[10^{-3},5]$.  We exclude incomplete rays and rays failing the relative null
check below, refine a bracket of finite candidates locally, and compare that
result with the lowest finite coarse candidate.  The lower step-stable value is
reported as a \emph{minimum found}.  For Nat\'ario, $9$ of the $50$ fan rays
are excluded.  Its lowest retained candidate, at $b=1.429286$, is unbracketed;
the integral is step-stable but the impact-parameter search has no convergence
claim.  This search establishes neither a global minimum nor geodesic completeness.

At fixed impact parameter and affine endpoints, selected rays are recomputed
starting at step densities $16384$, $32768$ and $65536$ per reference span
$16/s$, where $s$ is the initial spatial scale.  Further doublings are used when
needed.  Table~\ref{tab:convergence_panel}(a) gives the densities and observed changes
for the minimum-found rays; the finest value is reported.  The near-axis rays
$b=10^{-3}$ use all three initial densities.  Their finest-pair absolute changes
are $7.90\times10^{-15}$, $2.07\times10^{-13}$, $1.26\times10^{-13}$ and
$2.18\times10^{-13}$ for Alcubierre, Nat\'ario, Van~den~Broeck and Rodal,
respectively.  Selected coordinate-path values use sample densities
$1024$, $2048$ and $4096$ per reference span $16$.

The null check requires $\max|g(k,k)|<10^{-6}D$, where
$D=\max_{\rm nodes}[(\max_{ab}|g_{ab}|)(\max_a|k^a|)^2]$; a zero $D$ is
replaced by one.  Thus the test is relative, whereas the table prints the
absolute residual.  A second monitor is the relative drift of
$E_K=-g(k,K)$, where $K=\partial_t+v_s\partial_x$ is Killing for the
translating benchmark metrics.  Both use every quadrature node.  The displayed
maxima cover retained finite fan rays and selected refinements.  They are
numerical diagnostics, not trajectory or quadrature error bounds.

\paragraph{Quantum inequality.}
At a fixed wall point the passing bubble presents a coordinate-static
observer with a temporary negative-energy pulse.  As a reference estimate, we use
the four-dimensional Minkowski inertial bound for a free massless minimally
coupled scalar field, with Lorentzian sampling,
$\int \rho(\tau)\,f(\tau)\,\dd\tau \ge -C\hbar/\tau_0^4$, with
$C=3/(32\pi^2)$ and
$f(\tau)=\tau_0/[\pi(\tau^2+\tau_0^2)]$ the normalized sampling function.
In geometric units $G=c=1$, $\hbar=\ell_P^2$; retaining it distinguishes
the Planck length from the bubble scale.  Balancing the wall pulse
against the bound gives a threshold set by the geometric mean of the
two scales,
\begin{equation}
  \tau_0^{\mathrm{th}} \;\simeq\; c_{\rm metric}\,\sqrt{\ell_P R_b},
  \label{eq:qi-threshold}
\end{equation}
where $c_{\rm metric}=(C/|R_b^2\rho_{\min}|)^{1/4}$ in the
short-window approximation $\rho(\tau)\simeq\rho_{\min}$.
Here $\rho=T_{00}/(-g_{00})$ is the density measured by
$u=(-g_{00})^{-1/2}\partial_t$, defined where $g_{00}<0$; it is generally
different from the Eulerian density $\rho_n$.
The plotted curves set $R_b=G=c=\hbar=1$ and show the dimensionless diagnostic
\[
  \mathcal M(\tau_0)=R_b^2\left[
    \int_{\tau_-}^{\tau_+}\rho(\tau)f(\tau)\,d\tau
    +\frac{C\hbar}{\tau_0^4}\right].
\]
Proper time is accumulated along the coordinate-static worldline, with
$d\tau=\sqrt{-g_{00}}\,dt$; the numerical window extends approximately
$10\tau_0$ on either side of its center. The Lorentzian kernel is not
renormalized on this finite window; we do not bound the omitted stress tails.
These Alcubierre and Rodal pulse diagnostics sample only selected worldlines.
Their zero crossings cannot be rescaled to macroscopic $R_b$ by
equation~\eqref{eq:qi-threshold}.
Curvature and acceleration corrections would be required for a quantum-energy
bound on these worldlines~\cite{kontou2020,kontou2024qei}.
Neither diagnostic establishes complete-geodesic ANEC violation or a
curved-spacetime quantum-energy bound.

\begin{table}[tbp]
  \centering
  \footnotesize
  \caption{Averaged and quantum diagnostics at matched parameters ($R_b=1$,
    $\sigma=8$, $v_s=0.5$).  \emph{(a)}~Finite-segment null-geodesic integrals, with the
    symplectic integrator at the finest step densities reported in
    Table~\ref{tab:convergence_panel}(a), with quadrature at every step.
    $\int T_{ab}k^a k^b\,\dd\lambda$ is given on
    $b=10^{-3}$ and at the minimum found (with its
    $b^\ast$); Nat\'ario's selected candidate is unbracketed.
    The maxima cover retained finite fan rays and selected refinements;
    $\max|g(k,k)|$ is the absolute off-cone deviation and
    $\max|\Delta E_K/E_K|$ the worst relative drift of the conserved Killing
    energy.  \emph{(b)}~The coordinate-path integral with projected null directions
    $\int T_{ab}k^ak^b\,\dd\lambda$ at $b=10^{-3}$ and at the minimum found over its coordinate-path scan,
    with selected-ray quadrature refinement, with the Minkowski ray $=0$ to ${<}10^{-8}$, beside the
    Ford--Roman columns: the most-negative static-observer wall energy density
    $R_b^2\rho_{\min}$ and the short-window coefficient $c_{\rm metric}$
    in equation~\eqref{eq:qi-threshold}, reported for the selected
    Alcubierre and Rodal worldlines (a dash elsewhere).}
  \label{tab:averaged_quantum}

  \emph{(a)}\;\begin{tabular}{@{}l rr cc l@{}}
  \toprule
  Metric & $b=0.001$ & minimum found ($b^\ast$) & $\max|g(k,k)|$ & $\max|\Delta E_K/E_K|$ & method \\
  \midrule
  Alcubierre & $-0.0772$ & $-0.1462$ ($0.41$) & $1.5e-10$ & $1.0e-10$ & symplectic \\
  Natário & $-0.1797$ & $-0.5404$ ($1.43$) & $1.1e-06$ & $6.7e-07$ & symplectic \\
  Van den Broeck & $-0.0536$ & $-0.0626$ ($0.60$) & $3.9e-11$ & $4.2e-11$ & symplectic \\
  Rodal & $+0.0424$ & $-0.0017$ ($2.57$) & $1.4e-10$ & $9.5e-11$ & symplectic \\
  \bottomrule
\end{tabular}

  \medskip
  \emph{(b)}\;\begin{tabular}{l rr rr}
  \toprule
  & \multicolumn{2}{c}{Finite-segment $\int T_{ab}k^ak^b\,\dd\lambda$} & \multicolumn{2}{c}{Short-window estimate} \\
  \cmidrule(lr){2-3}\cmidrule(lr){4-5}
  Metric & $b=0.001$ & min found ($b^\ast$) & $R_b^2\rho_{\rm static,min}$ & $c_{\rm metric}$ \\
  \midrule
  Alcubierre & $-0.161$ & $-0.793$ (0.86) & $-0.0876$ & $0.57$ \\
  Natário & $-0.588$ & $-3.51$ (0.82) & -- & -- \\
  Van den Broeck & $+0.0217$ & $-0.0987$ (0.82) & -- & -- \\
  Rodal & $+0.0439$ & $-0.0136$ (1.27) & $-0.00704$ & $1.08$ \\
  \bottomrule
\end{tabular}

\end{table}

\begin{figure}[tbp]
  \centering
  \includegraphics[width=\textwidth]{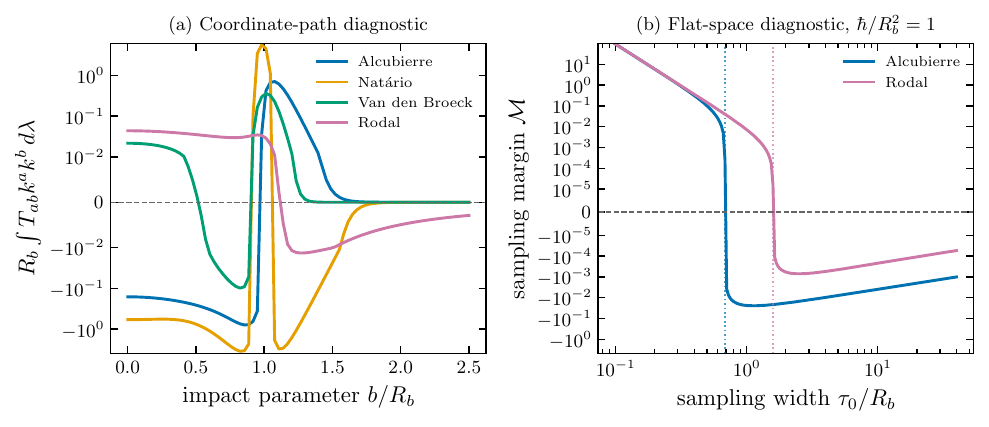}
  \caption{Finite-window diagnostics at $R_b=1$, $\sigma=8$, $v_s=0.5$.
    (a)~Coordinate-path integrals of the projected-null contraction versus
    impact parameter $b$; the projected vector need not be the path tangent.
    The strong dependence on $b$ motivates the separate geodesic calculations
    in Table~\ref{tab:averaged_quantum}(a).
    (b)~Sampling margin $\mathcal M(\tau_0)$ defined in the text for selected
    coordinate-static worldlines, with $\hbar=G=c=1$.
    Dotted lines mark its zero crossings; both ordinates use symmetric
    logarithmic scales. Neither panel establishes complete-geodesic ANEC
    violation or a curved-spacetime quantum bound.}
  \label{fig:averaged_quantum}
\end{figure}

\subsection{Kinematically standardized constructions and a slice diagnostic}
\label{sec:construction-exoticity}

To compare constructions with different length scales,
Table~\ref{tab:cross_matched} applies the same pointwise energy-condition
diagnostics to the regularized Fuchs shell~\cite{fuchs2024constvel} and the
Garattini--Zatrimaylov de~Sitter bubble~\cite{garattini2025desitter}, beside
Alcubierre and Rodal.  This panel standardizes kinematics only.

The standardized block uses $v_s=0.02$, characteristic radius $R_c=15\,$m,
and $10$--$90\%$ profile width $W=4.419943\,$m, chosen from the Fuchs sigmoid.
For the $\tanh$ profiles this gives $\sigma=0.497115373\,$m$^{-1}$.
All rows use $t=0$, the band $0.1\le f\le0.9$ and $r_{\max}/R_c=2$.
Each has its own axisymmetric grid ladder
$(N_r,N_\mu)=(32,32),(48,48),(64,64)$, clustered around its bubble center.
The Garattini--Zatrimaylov center is $r_0=v_s/H$, not the coordinate origin.
The reference block uses the parameter sets specified in
Appendix~\ref{app:metric_definitions}; only the Fuchs numerical parameters
are taken directly from its source.

Background curvature, shell compactness and matter content differ across the
rows.  The printed $R_c^2\min(\rho+p_i)$ values are numerical diagnostics of
the specified metrics and samples; the panel supplies no comparison of physical
matter content or stress severity.  For Garattini--Zatrimaylov, the matched-motion
shift is a sum of radial gradients, so Lemma~\ref{lem:rodal} gives exactly
$100.0\%$ Type~I and $0.0\%$ Type~II.  The repeated real eigenvalue does not
produce a Jordan block.  Its tabulated fractions use this exact result.
The Fuchs row refers to the covariant modification and $C^2$ regularization
specified in the appendix, at fixed smoothing widths.

\paragraph{Composite slice diagnostic.}
The composite in Table~\ref{tab:exoticity_index} combines three diagnostics on
$t=0$, at $R_b=1$, $\sigma=8$, $v_s=0.5$: the minimum Type-I NEC algebraic slack
on the $N=100$ sweep grid, the proper-volume Type-IV fraction on the wall
$0.1\le f\le0.9$, and the finite-segment basin-local minimum found above.
With Alcubierre as reference, its sub-scores are
\begin{equation}
 s_{\rm NEC}=\min\!\left(1,\frac{|\min(\rho+p_i)|}{|\min(\rho+p_i)|_{\rm Alc}}\right),
 \quad s_{\rm IV}=\phi_{\rm IV},\quad
 s_{\rm seg}=\min\!\left(1,\frac{|I_{\rm seg,min}|}{|I_{\rm seg,min}|_{\rm Alc}}\right),
 \label{eq:exoticity-subscores}
\end{equation}
and
\begin{equation}
 \mathcal E=[\max(s_{\rm NEC},\epsilon_0)\max(s_{\rm IV},\epsilon_0)
                 \max(s_{\rm seg},\epsilon_0)]^{1/3},\qquad\epsilon_0=10^{-4}.
 \label{eq:exoticity-index}
\end{equation}
The affine normalization is equation~\eqref{eq:anec-normalization}; the initial
surface, endpoints and search are specified in Section~\ref{sec:averaged-quantum}.
No rapidity cap enters these three axes; the separate timelike observer
plots use $\zeta_{\max}=5$.  The composite depends on the slice, mask,
normalization, ratio caps and floor.  In particular, when $\phi_{\rm IV}=0$,
its value scales as $\epsilon_0^{1/3}$ if the other sub-scores exceed the floor.
It has no observer-independent severity interpretation.  The individual
sub-scores are reported so that the effects of these choices remain visible.

\paragraph{Speed scaling of the wall NEC deficit.}
The speed sweep suggests simple power laws, but a fitted exponent need not
follow from the field equations.  We first state the exact scaling under fixed
geometric assumptions, then compare it with the numerical fits.
For a flat unit-lapse translating profile $\beta=v_sF(\mathbf x-v_st\mathbf e_x)$,
the Eulerian density and spatial stress scale as $v_s^2$ at fixed comoving points
while momentum scales as $v_s$ (Section~\ref{sec:classification}); when $j=0$, the invariant deficit is
\begin{equation}
  \min(\rho+p_i) = -C\,v_s^2\,,
  \label{eq:wall-deficit}
\end{equation}
with speed-independent $C$, positive where the wall infimum is negative.
The wall domain must be fixed in comoving coordinates.  Under these fixed-profile and $j=0$ hypotheses, the tensor is Type~I and
every $(\rho+p_i)/v_s^2$ is speed-independent for $v_s\ne0$.
Zero momentum alone does not imply $C>0$.
For the full three-dimensional Rodal profile, $C\approx\rodalNECscaling$
is computed numerically from that geometry.
The numerical coefficient depends on the spatial resolution: $C=\rodalNECscaling$
on the velocity sweep
(Table~\ref{tab:speed_scaling}(b)) against $C=0.759$ on the finest rung of the wall-resolved
ladder, where $\min(\rho+p_i)=-0.1898$ at $v_s=\tfrac12$
(Table~\ref{tab:diagnostic_convergence}), an observed $1.2\%$ spread.  This spread is not an uncertainty bound.

For Alcubierre and Nat\'ario, $j$ need not vanish and the Type-I subset itself
can change with speed.  Their fixed-exponent rows in
Table~\ref{tab:speed_scaling}(b) are empirical fits to sampled Type-I wall
slacks.  The deviations can reflect both momentum-dependent eigenvalues and
sampling of a changing subset; the present data do not separate these effects.
Van~den~Broeck's conformally curved slice lies outside the flat-slice hypothesis.
For the flat-slice family at $v_s>0$, the axial null contraction is
$\rho_n+S_\parallel-2|j|=a\,v_s^2-2d\,v_s$. Its speed-linear coefficient
is set by the double curl of the shift profile and vanishes when $j=0$.
For the Nat\'ario geometry, the nonzero momentum follows from localization.  Its shift
is divergence-free ($\nabla\!\cdot\!\beta=0$, zero expansion), so
Lemma~\ref{lem:jcurl} reduces to $8\pi j=\tfrac12\nabla^2\beta$ and the momentum
density vanishes identically only if the shift is harmonic; a harmonic field tending
to a constant at infinity equals that constant, which is Minkowski in translated
coordinates.  A divergence-free drive that approaches a uniform shift at infinity
without being uniform therefore has $|j|\neq0$
somewhere in space, and there the null-witness deficit acquires a nonzero
negative linear coefficient $-2d=-2|j|/v_s<0$.  Vorticity alone would not have
sufficed: a uniform rotation is vortical with $j\equiv0$.  For a gradient shift,
$j\equiv0$ by Lemma~\ref{lem:jcurl}, with the harmonic
sublocus Minkowski under the full R-Warp assumptions of
Barzegar--Buchert--Vigneron.
Equation~\eqref{eq:wall-deficit} is pointwise.  The integrated identity of
Santiago--Schuster--Visser~\cite{santiago2021} applies on a flat unit-lapse slice
$\mathbb R^3$, with integrable squared shift
gradients and boundary terms vanishing at infinity.  Integrating the Hamiltonian
constraint $16\pi\rho_n=K^2-K_{ij}K^{ij}$ by parts gives
\begin{equation}
  \int\rho_n\,dV=-\frac{1}{16\pi}\int\omega_{ij}\omega^{ij}\,dV\le0,
  \qquad \omega_{ij}=\partial_{[i}\beta_{j]},
  \label{eq:ssv-integral}
\end{equation}
with $\omega_{ij}$ the shift vorticity
($|\nabla\times\beta|^2=2\omega_{ij}\omega^{ij}$).  A nonzero shift vorticity makes
the right side strictly negative, so $\rho_n<0$ on a set of positive measure and the
weak energy condition fails.  This integrated identity alone does not prove
NEC violation.  Vorticity is sufficient but not necessary for a
violation: the ideal irrotational Rodal drive has $\int\rho_n\,dV=0$.
Its $\beta=O(r_s^{-1})$ and $\partial\beta=O(r_s^{-2})$ tails make the squared
gradients integrable and the integration-by-parts boundary term vanish.
It nevertheless has pointwise weak- and null-energy violations, with a
negative wall deficit
$\min(\rho+p_i)=-C\,v_s^2$ of equation~\eqref{eq:wall-deficit} at nonzero speed.
The integrated negative Eulerian energy also scales exactly as $v_s^2$,
without the additional $j=0$ assumption (Lemma~\ref{lem:integrated}).

\begin{lemma}[Exact integrated negative-energy scaling]
\label{lem:integrated}
For a drive with $\alpha=1$, a time-independent spatial metric
$\gamma_{ij}=\delta_{ij}$, and shift $\beta^i=v_s\,F^i(\mathbf{x})$
($F$ independent of $v_s$ on the slice) the Eulerian integrated negative energy
$E_-(v_s)=\int_{\rho_n<0}|\rho_n|\,dV$ (with $\rho_n=T_{ab}n^an^b$), on a fixed
integration domain, with $E_-(1)<\infty$, obeys
$E_-(v_s)=v_s^2\,E_-(1)$ identically.  For a translating profile the domain and
spatial coordinates are held fixed relative to the bubble.
\end{lemma}
\begin{proof}
The extrinsic curvature $K_{ij}=\partial_{(i}\beta_{j)}$ is linear in $v_s$, so on
the flat slice the Hamiltonian constraint $16\pi\rho_n=K^2-K_{ij}K^{ij}$ gives
$\rho_n=v_s^2\,\hat\rho(\mathbf{x})$ with $\hat\rho$ independent of $v_s$
(Section~\ref{sec:classification}).  The proper volume element is $v_s$-independent
there.  For $v_s\ne0$ the negative set is fixed and the $v_s^2$
prefactor comes out of the integral.  At $v_s=0$ the negative set is empty and
both sides vanish.
\end{proof}

\noindent The free-exponent fit of this integral returns $p=2.00$, $R^2=1$ for Alcubierre, Nat\'ario, and
Rodal (Table~\ref{tab:speed_scaling}(c)), the integrated companion to the pointwise
deficit and the same slice-integrated measure Rodal reports~\cite{rodal2026}.
Van~den~Broeck is the exception: its conformally curved slice violates the flat-slice
hypothesis, so its negative energy has a $v_s$-independent conformal offset present
even at rest.  These speed laws follow from the scaling of the shift and the ADM constraints.
An exact quadratic dependence also appears in the profile-specific integrated
energy of Jusufi and Lobo~\cite{jusufi2026warp}; it does not remove the
classical energy-condition violations of that effective metric.

\paragraph{Curvature scaling.}
In this paragraph, $K=R_{abcd}R^{abcd}$ denotes the Kretschmann scalar,
distinct from the trace of the extrinsic curvature.  We write
$C^2=C_{abcd}C^{abcd}=K-2R_{ab}R^{ab}+R^2/3$ for the squared Weyl tensor
in four spacetime dimensions.
Table~\ref{tab:speed_scaling}(d) fits curvature magnitudes
to $X_{\max}=Av_s^q$ over the sampled subluminal branch
(Figure~\ref{fig:curvature_scaling}). Here $X_{\max}$ is a local maximum found
by unconstrained polishing from a wall-grid seed; it is not an enclosed maximum
over the wall. These are finite-speed fits: the selected point can move, and
subleading terms can affect the fitted exponent.
Previous studies examined the individual Alcubierre and Nat\'ario curvature
invariants~\cite{rodal2023invariants,rodal2024natario}; the matched panel uses a
common profile scale and wall mask for the seeds.
Rodal's fitted curves are quartic to the reported precision, while Alcubierre
and Nat\'ario are close to quadratic. Van~den~Broeck already has curvature at
rest and admits no comparable single power law over this range.

For the fixed flat-slice, unit-lapse profiles considered above, the leading
coefficients follow from the speed scalings of the stress components.  The
trace-reversed Einstein equation gives
$R_{ab}R^{ab}=64\pi^2\,T_{ab}T^{ab}$ with
$T_{ab}T^{ab}=\rho_n^2-2|j|^2+|S|^2$ in the Eulerian frame, so the contribution of order $v_s^2$ to $R_{ab}R^{ab}$
is $-128\pi^2|j|^2$ (since $|j|\propto v_s$) and vanishes exactly at $j=0$, in particular for
an irrotational shift, leaving the $v_s^4$ growth of
the energy-stress terms.  This
contraction is not sign-definite in Lorentzian signature (its leading vortical term
is negative), so the fitted quantity is the locally maximized magnitude
$|R_{ab}R^{ab}|$. The
Weyl-squared and Kretschmann exponents involve the full Riemann tensor and are
\emph{not} fixed by that identity.  One direction of the mechanism is
elementary: for a fixed smooth stationary or rigidly translating profile on
flat unit-lapse slices, the order-$v_s$ part of the metric is the shift
covector $h_{0i}=\beta_i$, and for an irrotational shift $\beta_i=v_s\,\partial_i\Phi$
this is pure gauge (the time relabeling $t\to t-v_s\Phi$), so the order-$v_s$
linearized Riemann tensor vanishes and every quadratic invariant starts at $v_s^4$ or
higher.  Nonzero shift vorticity rules out this scalar time relabeling, but does not
by itself imply a nonzero curvature coefficient.  Rigid rotation, for example,
is removed by a time-dependent spatial rotation.  The Kretschmann and Weyl-squared coefficients have different kernels.

For the following coefficients write $\beta=v_sF$ and
$\bar\omega_{jk}=\partial_{[j}F_{k]}$, so $\omega_{jk}=v_s\bar\omega_{jk}$.
The notation $[X]^{(2)}$ denotes the coefficient of $v_s^2$.  The only
order-$v_s$ Riemann components that survive are
$R^{(1)}_{0ijk}=\partial_i\bar\omega_{jk}$, and contracting them gives
\begin{equation}
  \bigl[K\bigr]^{(2)} = -4\sum_{ijk}\bigl(\partial_i\bar\omega_{jk}\bigr)^2 ,
  \qquad
  \bigl[C^2\bigr]^{(2)} = \bigl[K\bigr]^{(2)}
     + 4\sum_i\Bigl(\partial^j\bar\omega_{ji}\Bigr)^{2} .
  \label{eq:curv-coeffs}
\end{equation}
The Kretschmann coefficient is \emph{minus a sum of squares}.  On the eight-dimensional
space of admissible $\partial\bar\omega$ (nine components of $\partial_i\bar\omega_{jk}$ less
the one closedness relation $\partial_{[i}\bar\omega_{jk]}=0$) it has signature $(0,8,0)$:
negative definite.  No cancellation is available, and the converse therefore
\emph{holds} for the Kretschmann scalar: $\partial\bar\omega\neq0$ forces
$[K]^{(2)}\neq0$.  The Weyl-squared coefficient has signature $(0,5,3)$, negative
semidefinite with a three-dimensional kernel, so cancellation \emph{is} available and
the converse fails.  The shift $F=(0,0,-x^2-y^2)$ realizes it: $\partial\bar\omega\neq0$
with $[K]^{(2)}=-16$ and $[C^2]^{(2)}=0$.

The kernel has a direct tensor description.  Write
$\bar\omega_{jk}=\epsilon_{jkl}\upsilon_l$ with $\upsilon=\tfrac12\nabla\times F$ and
$M_{il}=\partial_i\upsilon_l$; the closedness relation is exactly $\operatorname{tr}M=0$,
so the admissible space is the eight-dimensional traceless matrices.
Here $\|M\|^2=\sum_{il}M_{il}^2$ and
$\operatorname{sym}M=(M+M^{\mathsf T})/2$.
In those variables \eqref{eq:curv-coeffs} reduces to
\begin{equation}
  \bigl[K\bigr]^{(2)} = -8\,\lVert M\rVert^2 ,
  \qquad
  \bigl[C^2\bigr]^{(2)} = -8\,\lVert \operatorname{sym}M\rVert^2 .
  \label{eq:curv-closed}
\end{equation}
Writing $(\operatorname{skew}M)_{il}=\epsilon_{ilk}a_k$, the antisymmetric part
of $M$ contributes $2\lvert a\rvert^2$ to
$\lVert M\rVert^2$ and $4\lvert a\rvert^2$ to $\sum_i(\partial^j\bar\omega_{ji})^2$, and
the two cancel identically.  The Kretschmann coefficient therefore depends on all
eight modes of $\partial\upsilon$, the Weyl-squared coefficient only on
the five shear modes: the converse fails for $C^2$ precisely when
$\partial\upsilon$ is antisymmetric, as the example above demonstrates.  Equation~\eqref{eq:curv-closed} gives the signatures $(0,8,0)$ and $(0,5,3)$.

Two consequences follow.  The controlling quantity is
$\partial\bar\omega$, not $\bar\omega$: a spatially uniform vorticity has $\bar\omega\neq0$ with
$\partial\bar\omega=0$, so the order-$v_s^2$ curvature coefficient vanishes.
A pure rigid rotation is flat; adding it to a gradient need not remove higher-order
curvature.  Uniform vorticity also gives $j=0$ in Lemma~\ref{lem:jcurl}.  The leading Kretschmann exponent is $q=2$ wherever
$\partial\bar\omega\neq0$.  A finite-speed Kretschmann
scalar is not a
pure power: for the shift $F=(0,x^2,0)$, exact computation gives
$K=4v_s^2\bigl(11v_s^2x^4-2\bigr)$, whose $O(v_s^2)$ coefficient is $-8\neq0$ as
\eqref{eq:curv-coeffs} requires, but which also has an $O(v_s^4)$ term and
therefore vanishes at the isolated speed $v_s=\sqrt{2/11}$ on the surface $x=1$.  What
$[K]^{(2)}\neq0$ excludes is a \emph{higher} leading exponent; the exponents fitted
over a finite $v_s$ range in Table~\ref{tab:convergence_panel}(b) can include
subleading corrections and motion of the maximizing point.  The worst
single-point relative deviation from the fitted law, over the subluminal branch the
fit covers, is
$1.7\times10^{-2}$ (Alcubierre, Weyl), $1\times10^{-3}$ (Nat\'ario, all three) and
$4\times10^{-15}$ (Rodal, all three), rising to $0.25$ for
Alcubierre Ricci-squared at $v_s=0.9$, where the
subleading $|S|^2$ terms of $T_{ab}T^{ab}$ are already comparable to the leading
$-2|j|^2$.  Those terms affect the finite-speed fit $q=2.11$; the leading order
remains $v_s^2$ where its coefficient is nonzero.
The table reports fit quality separately from algebraic type: curvature scalars
exist for every type. The crossing of the fitted Rodal and Alcubierre Ricci
curves is a comparison of those empirical curves, not a bound on either wall maximum.

\begin{table}[tbp]
  \centering
  \footnotesize
  \caption{Kinematically standardized construction diagnostics at $t=0$, on
    $(N_r,N_\mu)=(32,32),(48,48),(64,64)$.
    \emph{(a)}~Common $v_s=0.02$, $R_c=15\,$m, $W=4.419943\,$m,
    and band $0.1\le f\le0.9$.
    \emph{(b)}~Reference parameter sets from Appendix~\ref{app:metric_definitions}.
    For Garattini--Zatrimaylov, (a) uses $H=0.02/15\,$m$^{-1}$, $r_0=15\,$m,
    $\sigma=0.497115373\,$m$^{-1}$; (b) uses $H=0.1\,$m$^{-1}$, $r_0=1\,$m,
    $\sigma=8\,$m$^{-1}$.  Both have $v=Hr_0$ and exact Type~I fractions.
    These are parameter choices made here.  Wall cells count the coarsest and
    finest levels.  Margins are sampled Type-I slacks, with no global enclosure.}
  \label{tab:cross_matched}

  \emph{(a)}\;\begin{tabular}{@{}l cc cc c cc@{}}
  \toprule
  & \multicolumn{2}{c}{Wall cells} & Type~I & Type~IV & $R_c^2\min(\rho+p_i)$ & WEC & NEC \\
  Metric & coarse & fine & (\%) & (\%) & (Type~I) & \multicolumn{2}{c}{miss (\%)} \\
  \midrule
  Alcubierre & 7.3 & 14.9 & 0.0 & 100.0 & -- & 0.0 & 0.0 \\
  Rodal & 7.3 & 14.9 & 100.0 & 0.0 & -2.7e-04 & 71.9 & 8.9 \\
  Fuchs & 7.3 & 14.9 & 100.0 & 0.0 & 0.021 & -- & -- \\
  Garattini & 7.3 & 14.9 & 100.0 & 0.0 & -3.5e-04 & 100.0 & 0.0 \\
  \bottomrule
\end{tabular}

  \medskip
  \emph{(b)}\;\begin{tabular}{@{}l c cc cc c cc@{}}
  \toprule
  & $v_s$ & \multicolumn{2}{c}{Wall cells} & Type~I & Type~IV & $R_c^2\min(\rho+p_i)$ & WEC & NEC \\
  Metric & & coarse & fine & (\%) & (\%) & (Type~I) & \multicolumn{2}{c}{miss (\%)} \\
  \midrule
  Alcubierre & 0.5 & 4.5 & 9.2 & 0.6 & 99.4 & -0.137 & 0.0 & 0.0 \\
  Rodal & 0.5 & 4.5 & 9.2 & 100.0 & 0.0 & -0.194 & 72.6 & 9.6 \\
  Fuchs & 0.02 & 7.5 & 15.4 & 100.0 & 0.0 & 0.022 & -- & -- \\
  Garattini & 0.1 & 4.5 & 9.2 & 100.0 & 0.0 & -0.010 & 100.0 & 0.0 \\
  \bottomrule
\end{tabular}

\end{table}

\begin{table}[tbp]
  \centering
  \footnotesize
  \caption{Composite diagnostic on the specified $t=0$ slice,
    equations~\eqref{eq:exoticity-subscores}--\eqref{eq:exoticity-index}.
    NEC and finite-segment sub-scores are ratios to Alcubierre capped at one.
    The Type-IV column is a proper-volume wall fraction from the $N=100$
    sweep grid.  Each factor is floored at $10^{-4}$ before forming the mean.}
  \label{tab:exoticity_index}
  \begin{tabular}{@{}l cccc@{}}
  \toprule
  Metric & NEC & Type~IV & $I_{\rm seg}$ & Composite \\
  & severity & fraction & $|\min|$ & index \\
  \midrule
  Alcubierre & 1.000 & 0.987 & 1.000 & 0.996 \\
  Natário & 1.000 & 0.834 & 1.000 & 0.941 \\
  Van den Broeck & 1.000 & 0.811 & 0.428 & 0.703 \\
  Rodal & 1.000 & 0.000 & 0.011 & 0.010 \\
  \bottomrule
\end{tabular}

\end{table}

\begin{table}[tbp]
  \centering
  \footnotesize
  \caption{Speed scaling over the subluminal branch.  \emph{(a)}~Free-exponent fit
    $|\min(\rho+p_i)| = A\,v_s^{\,p}$ of the wall NEC deficit on the Type-I branch.
    The Alcubierre and Nat\'ario fits use speed-dependent Type-I subsets,
    as discussed in the text; $A$ is the fitted amplitude, and
    Van~den~Broeck's Type-I remainder admits no single power law ($R^2<0.99$).
    \emph{(b)}~Fixed-exponent empirical fits $\min(\rho+p_i)=-Cv_s^2$
    for Alcubierre and Nat\'ario.  Rodal's quadratic form follows under the
    $j=0$ and fixed-domain hypotheses of equation~\eqref{eq:wall-deficit}; its
    displayed coefficient remains a numerical estimate.  Here the fixed-fit statistic
    is $R^2=1-\sum(y-Cv_s^2)^2/\sum y^2$, with $y=|\min(\rho+p_i)|$.
    \emph{(c)}~Free-exponent fit of the Eulerian slice-integrated negative energy
    $E_-=\int_{\rho_n<0}|\rho_n|\,dV$, an Eulerian quantity throughout since
    $\rho_n=T_{ab}n^an^b$ is defined at every point regardless of Hawking--Ellis
    type. \emph{(d)}~Locally maximized invariant magnitudes, seeded from the wall
    grid, fit to $X=A v_s^q$.
    The listed exponents and $R^2_{\log}$ describe finite-speed regressions,
    independently of the leading coefficient identities in the text.}
  \label{tab:speed_scaling}

  \emph{(a)}\;\begin{tabular}{@{}l ccc@{}}
  \toprule
  Metric & exponent $p$ & coefficient $A$ & $R^2$ \\
  \midrule
  Alcubierre & 2.03 & 0.634 & 1.0000 \\
  Natário & 2.26 & 8.526 & 0.9970 \\
  Van den Broeck & \multicolumn{3}{c}{no clean fit ($R^2 < 0.99$)} \\
  Rodal & 2.00 & 0.768 & 1.0000 \\
  \bottomrule
\end{tabular}

  \medskip
  \emph{(b)}\;\begin{tabular}{@{}l ccccc@{}}
  \toprule
  Metric & $C$ & dev.\ (all $v_s$) & dev.\ ($v_s\ge0.5$) & $R^2$ & NEC at sampled $v_s$ \\
  \midrule
  Alcubierre & 0.631 & 4.09\% & 0.92\% & 1.0000 & violated \\
  Natário & 7.833 & 86.67\% & 2.77\% & 0.9999 & violated \\
  Van den Broeck & \multicolumn{4}{c}{no adequate quadratic fit ($R^2=0.56$)} & violated \\
  Rodal & 0.768 & 0.00\% & 0.00\% & 1.0000 & violated \\
  \bottomrule
\end{tabular}

  \medskip
  \emph{(c)}\;\begin{tabular}{@{}l cc@{}}
  \toprule
  & \multicolumn{2}{c}{$E_-=\int_{\rho_n<0}|\rho_n|\,dV$} \\
  \cmidrule(lr){2-3}
  Metric & exponent $p$ & $R^2$ \\
  \midrule
  Alcubierre & 2.00 & 1.0000 \\
  Natário & 2.00 & 1.0000 \\
  Van den Broeck & 0.04 & 0.7528 \\
  Rodal & 2.00 & 1.0000 \\
  \bottomrule
\end{tabular}

  \medskip
  \emph{(d)}\;\begin{tabular}{@{}l ccc ccc@{}}
  \toprule
  & \multicolumn{3}{c}{Weyl $|C^2|$} & \multicolumn{3}{c}{Ricci $|R_{ab}R^{ab}|$} \\
  \cmidrule(lr){2-4}\cmidrule(lr){5-7}
  Metric & $q$ & $A$ & $R^2_{\log}$ & $q$ & $A$ & $R^2_{\log}$ \\
  \midrule
  Alcubierre & 1.97 & 700.225 & 0.9999 & 2.11 & 454.384 & 0.9905  \\
  Natário & 2.00 & 72703.741 & 1.0000 & 2.00 & 33751.084 & 1.0000  \\
  Van den Broeck & 0.63 & 239.232 & 0.5362 & 0.18 & 453.277 & 0.5273  \\
  Rodal & 4.00 & 393.554 & 1.0000 & 4.00 & 1032.108 & 1.0000  \\
  \bottomrule
\end{tabular}

\end{table}

\begin{figure}[tbp]
  \centering
  \includegraphics[width=\textwidth]{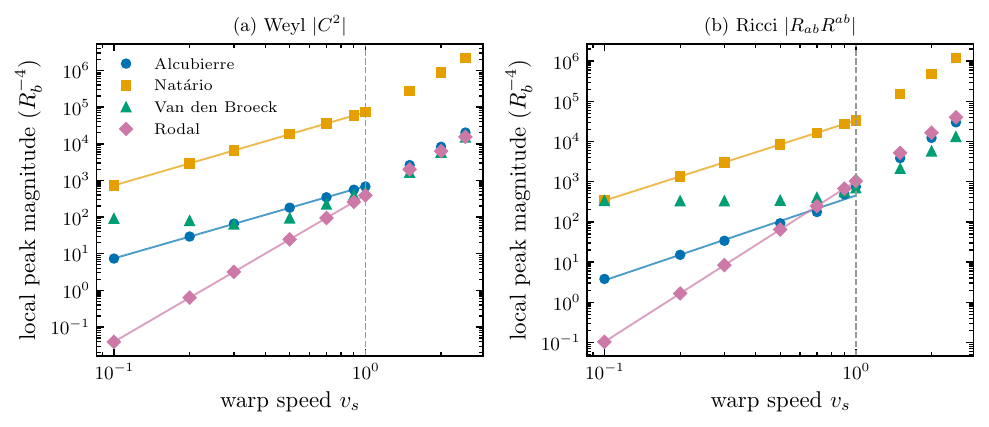}
  \caption{Local curvature maxima found from $N=60$ wall-grid seeds
    ($R_b=1$, $\sigma=8$; log--log; $v_s=1$ dashed).
    (a)~Weyl-squared magnitude $|C^2|$; (b)~Ricci-squared magnitude
    $|R_{ab}R^{ab}|$. Polishing is unconstrained and supplies no global enclosure.
    Lines fit the subluminal samples only where $R^2_{\log}\ge0.9$;
    Van~den~Broeck fails this criterion. Rodal's quartic fits are consistent with
    the absent quadratic coefficient of its ideal irrotational profile.
    The finite-speed fits do not establish the coefficient identities
    in equations~\eqref{eq:curv-coeffs}--\eqref{eq:curv-closed}.}
  \label{fig:curvature_scaling}
\end{figure}

\section{Resolution stability}
\label{sec:convergence}

For a fixed implemented metric, automatic differentiation introduces no
finite-difference step error; refining the spatial sampling grid does not change
the implemented tensor at a fixed point. Roundoff and numerical profile evaluation
remain relevant. Fuchs also has radial-construction discretization error,
which is tested separately at fixed regularization. We assess wall extrema and volume fractions separately. All are
computed on the wall-clustered graded grids of Section~\ref{sec:results-ec}
($N=80,100,120$).

\emph{(i) Wall extrema}, the minimum NEC and DEC margins $\min(\rho+p_i)$,
$\min(\rho-|p_i|)$ and the Type-IV imaginary eigenvalue
$\max|\operatorname{Im}\lambda|$, are numerically evaluated extrema.  Eigenvalue degeneracies and Type-I
restriction can prevent smoothness.  In exact arithmetic, a sampled extremum
is a one-sided bound on the corresponding infimum or supremum. Its residual gap
can oscillate with grid alignment, so we assume no Richardson order and do not
extrapolate it. Seeding from
the deepest sample across the ladder, we instead refine the numerical tensor evaluation on a shrinking
local grid, ending at a $\sim\!10^{-4}$ spatial tolerance within that basin.

That refinement is
\emph{basin-local}: it cannot exclude a deeper basin missed by the grid.
The polished value is a numerical estimate, at least as extreme as the
computed grid values; it supplies no rigorous global bound.  The separate
interval search in Appendix~\ref{app:enclosures} bounds the type-free null
deficit $\mathcal{N}$ over the whole wall for all four drives, using the
axisymmetric reduction.  Table~\ref{tab:enclosures} reports each bracket and
its width.  The Type-I-restricted margins of the
vortical walls, the dominant-condition deficits and $\max|\operatorname{Im}\lambda|$
have no global bound.  \emph{(ii) Volume fractions}, the wall
Type-I/Type-IV fractions and the wall-restricted miss rates, are integrals of a
discontinuous indicator.  Boundary crossings can dominate their grid error; for these we report the measured stability
spread across the wall-resolved ladder.  The separate uniform-grid study uses
the same six-direction Eulerian NEC diagnostic at $N=25,49,97$; these grids
put only $0.66$, $1.32$ and $2.64$ cells across the $0.2747$-wide wall.
Its integrated and $L^2$ diagnostics are non-monotone, while the sampled
minimum is unchanged to the reported precision.  None of these observations
establishes a convergence order or an error bound.\footnote{On a $41\times41$ window carried with the Alcubierre bubble,
translation to $x_s=v_st=1$ moves the measured violation centroid by $1.000$
and preserves the sampled violation count.  The minima differ by orders of
magnitude near the center, revealing numerical regularization instability.
This is a check of the centroid and count only, not verification of scalar
invariance by those computed minima.}

\paragraph{Per-diagnostic convergence.}
On the wall-resolved ladder each principal diagnostic has a measured stability
spread, reported per drive
(Table~\ref{tab:diagnostic_convergence}). The wall Type-IV fraction
deviates from its ladder mean by at most $0.15$ percentage points for Alcubierre and
$1.32$ for Nat\'ario, and is identically zero for Rodal, which is $100\%$ Type~I.
Van den Broeck has the largest spread, $2.90$ percentage points.
The empirical vorticity fraction $\mathcal{R}_\omega$ is grid-stable to
$0.12$ percentage points or better across the ladder. The basin-local searches refine the sampled extrema without a global guarantee.
The
Type-IV imaginary eigenvalue peaks at $\max|\operatorname{Im}\lambda| = 2.588$
(Nat\'ario) down to $0$ (Rodal), and the invariant NEC and DEC deficits reach the
polished values tabulated. They are refined rather than extrapolated, with
no assumed convergence order.

\begin{table}[tbp]
  \centering
  \footnotesize
  \caption{Per-diagnostic convergence and stability of the wall diagnostics on the
    wall-resolved graded grids ($R_b=1$, $\sigma=8$, $v_s=0.5$; $N=80,100,120 \to
    5.9,7.6,8.9$ cells across the 10--90\% wall). The ladder columns are grid
    samples and need not be monotone.  Fractions are reported as ratios; their
    deviations are measured in percentage points.  These entries are numerical
    diagnostics only and are not certifications.}
  \label{tab:diagnostic_convergence}
  \resizebox{\columnwidth}{!}{\begin{tabular}{@{}l l ccc p{7.0cm}@{}}
  \toprule
  Metric & Diagnostic & $N{=}80$ & $N{=}100$ & $N{=}120$ & Numerical status \\
  \midrule
  Alcubierre & Wall Type-IV frac. & 0.989 & 0.987 & 0.990 & maximum absolute deviation from the three-grid mean: 0.145 percentage points \\
   & Wall $\max|\mathrm{Im}\,\lambda|$ & 0.267 & 0.269 & 0.271 & basin-local polished value: 0.272; no global enclosure \\
   & Wall $\min(\rho+p_i)$ & -0.134 & -0.156 & -0.157 & basin-local polished value: -0.158; no global enclosure \\
   & Wall $\min(\rho-|p_i|)$ & -0.174 & -0.177 & -0.164 & basin-local polished value: -0.196; no global enclosure \\
   & Vorticity $\mathcal{R}_\omega$ & 0.333 & 0.333 & 0.333 & maximum absolute deviation from the three-grid mean: 0.000 percentage points \\
   & Wall WEC miss rate & 0.000 & 0.000 & 0.000 & maximum absolute deviation from the three-grid mean: 0.000 percentage points \\
   & Wall DEC miss rate & 0.000 & 0.000 & 0.000 & maximum absolute deviation from the three-grid mean: 0.000 percentage points \\
  \midrule
  Natário & Wall Type-IV frac. & 0.859 & 0.834 & 0.845 & maximum absolute deviation from the three-grid mean: 1.315 percentage points \\
   & Wall $\max|\mathrm{Im}\,\lambda|$ & 2.579 & 2.581 & 2.584 & basin-local polished value: 2.588; no global enclosure \\
   & Wall $\min(\rho+p_i)$ & -1.853 & -1.905 & -1.896 & basin-local polished value: -1.955; no global enclosure \\
   & Wall $\min(\rho-|p_i|)$ & -1.853 & -1.905 & -1.896 & basin-local polished value: -1.955; no global enclosure \\
   & Vorticity $\mathcal{R}_\omega$ & 0.508 & 0.505 & 0.506 & maximum absolute deviation from the three-grid mean: 0.116 percentage points \\
   & Wall WEC miss rate & 0.000 & 0.000 & 0.000 & maximum absolute deviation from the three-grid mean: 0.000 percentage points \\
   & Wall DEC miss rate & 0.000 & 0.000 & 0.000 & maximum absolute deviation from the three-grid mean: 0.000 percentage points \\
  \midrule
  Van den Broeck & Wall Type-IV frac. & 0.869 & 0.811 & 0.839 & maximum absolute deviation from the three-grid mean: 2.902 percentage points \\
   & Wall $\max|\mathrm{Im}\,\lambda|$ & 0.218 & 0.218 & 0.218 & basin-local polished value: 0.219; no global enclosure \\
   & Wall $\min(\rho+p_i)$ & -0.517 & -0.518 & -0.519 & basin-local polished value: -0.521; no global enclosure \\
   & Wall $\min(\rho-|p_i|)$ & -0.872 & -0.872 & -0.872 & basin-local polished value: -0.873; no global enclosure \\
   & Vorticity $\mathcal{R}_\omega$ & 0.440 & 0.439 & 0.440 & maximum absolute deviation from the three-grid mean: 0.021 percentage points \\
   & Wall WEC miss rate & 0.588 & 0.548 & 0.574 & maximum absolute deviation from the three-grid mean: 2.194 percentage points \\
   & Wall DEC miss rate & 0.561 & 0.513 & 0.531 & maximum absolute deviation from the three-grid mean: 2.579 percentage points \\
  \midrule
  Rodal & Wall Type-IV frac. & 0.000 & 0.000 & 0.000 & maximum absolute deviation from the three-grid mean: 0.000 percentage points \\
   & Wall $\max|\mathrm{Im}\,\lambda|$ & 0.000 & 0.000 & 0.000 & basin-local polished value: 0.000; no global enclosure \\
   & Wall $\min(\rho+p_i)$ & -0.192 & -0.192 & -0.190 & basin-local polished value: -0.194; no global enclosure \\
   & Wall $\min(\rho-|p_i|)$ & -0.192 & -0.192 & -0.190 & basin-local polished value: -0.194; no global enclosure \\
   & Vorticity $\mathcal{R}_\omega$ & 0.000 & 0.000 & 0.000 & maximum absolute deviation from the three-grid mean: 0.000 percentage points \\
   & Wall WEC miss rate & 0.745 & 0.726 & 0.722 & maximum absolute deviation from the three-grid mean: 1.408 percentage points \\
   & Wall DEC miss rate & 0.757 & 0.740 & 0.736 & maximum absolute deviation from the three-grid mean: 1.271 percentage points \\
  \bottomrule
\end{tabular}
}
\end{table}

\paragraph{Convergence of the remaining quantities.}
Table~\ref{tab:convergence_panel}(a) reports observed spreads of the composite
across the spatial grid ladder and step refinement of the actual selected
finite-segment ray.  Its impact parameter and endpoints are held fixed;
this checks a local numerical result and supplies no global search or tail bound.
Part~(b) refits the curvature exponents on the same spatial ladder.  The exact
scaling statements apply only under the hypotheses given above; finite-speed
exponents and grid stability remain empirical.  Parts~(c)--(e) give the
native Rodal and observer-search checks.

\begin{table}[tbp]
  \centering
  \footnotesize
  \caption{Resolution stability of the diagnostics that admit no Richardson order.
    \emph{(a)}~Composite diagnostic on the wall-resolved grid ($N=80/100/120$) and the
    selected finite-segment ray at the listed step densities per $16/s$,
    with fixed impact parameter and endpoints.  Grid ``stable'' means a range
    at most $\max\{0.02,0.05(\overline{\mathcal E}+10^{-12})\}$; ray ``stable''
    means a finest-pair change at most $10^{-8}+10^{-4}|I_{\rm finest}|$.
    \emph{(b)}~Fitted curvature-invariant exponents $X_{\rm max}=A\,v_s^{\,q}$
    ($R_b=1$, $\sigma=8$), refit at each level; poor power-law fits are marked
    separately.  \emph{(c)}~Rodal missed fractions at the native
    parameters ($v_s=0.5$, $R_b=100$, $\sigma=0.03$, $N_{\text{starts}}=8$) across
    uniform grids; all points are Type~I at $N\ge50$, so the reference criterion is the algebraic
    slack.  \emph{(d)}~$N_{\text{starts}}$ ablation at $v_s=0.5$ on a $50^3$ grid
    ($\zeta_{\max}=5$).  \emph{(e)}~Alcubierre at $v_s=0.5$, $R_b=1$, $\sigma=8$,
    $t=0$, on uniform $N=25,49,97$ grids in $[-5,5]^3$.
    All levels use the minimum $m_6$ of the six axis-aligned Eulerian null
    contractions.  The integrated and $L^2$ diagnostics sum its negative part
    and squared negative part with coordinate cell-volume weights, taking a
    square root for the latter.  The coordinate center is excluded.
    The last column is the maximum absolute departure from the three-grid mean,
    with no extrapolation or assumed order.}
  \label{tab:convergence_panel}

  \emph{(a)}\;\resizebox{0.95\columnwidth}{!}{\begin{tabular}{@{}l ccc l l l@{}}
  \toprule
  & \multicolumn{4}{c}{Composite diagnostic ($N$)} & \multicolumn{2}{c}{Selected finite ray} \\
  \cmidrule(lr){2-5}\cmidrule(lr){6-7}
  Metric & 80 & 100 & 120 & verdict & step density: integral & verdict \\
  \midrule
  Alcubierre & 0.996 & 0.996 & 0.997 & stable (0.0008) & \shortstack[l]{16384: $-0.146227$\\32768: $-0.146227$\\65536: $-0.146227$} & stable ($\delta=8.7e-14$) \\
  Natário & 0.951 & 0.941 & 0.945 & stable (0.0092) & \shortstack[l]{16384: $-0.540367$\\32768: $-0.540367$\\65536: $-0.540367$} & stable ($\delta=1.5e-09$) \\
  Van den Broeck & 0.719 & 0.703 & 0.711 & stable (0.0162) & \shortstack[l]{16384: $-0.062634$\\32768: $-0.062634$\\65536: $-0.062634$} & stable ($\delta=5.0e-13$) \\
  Rodal & 0.010 & 0.010 & 0.010 & stable (0.0000) & \shortstack[l]{16384: $-0.001671$\\32768: $-0.001671$\\65536: $-0.001671$} & stable ($\delta=9.0e-14$) \\
  \bottomrule
\end{tabular}
}

  \medskip
  \emph{(b)}\;\resizebox{0.8\columnwidth}{!}{\begin{tabular}{@{}l l ccc c@{}}
  \toprule
  Metric & Invariant & \multicolumn{3}{c}{fitted $q$} & Spread \\
   & & $N{=}80$ & $N{=}100$ & $N{=}120$ & \\
  \midrule
  Alcubierre & Weyl $|C^2|$ & 1.97 & 1.97 & 1.97 & 0.00 \\
   & Ricci $|R_{ab}R^{ab}|$ & 2.11 & 2.11 & 2.11 & 0.00 \\
  \midrule
  Natário & Weyl $|C^2|$ & 2.00 & 2.00 & 2.00 & 0.00 \\
   & Ricci $|R_{ab}R^{ab}|$ & 2.00 & 2.00 & 2.00 & 0.00 \\
  \midrule
  Van den Broeck & Weyl $|C^2|$ & \multicolumn{4}{c}{poor power-law fit ($R^2_{\log}<0.99$)} \\
   & Ricci $|R_{ab}R^{ab}|$ & \multicolumn{4}{c}{poor power-law fit ($R^2_{\log}<0.99$)} \\
  \midrule
  Rodal & Weyl $|C^2|$ & 4.00 & 4.00 & 4.00 & 0.00 \\
   & Ricci $|R_{ab}R^{ab}|$ & 4.00 & 4.00 & 4.00 & 0.00 \\
  \bottomrule
\end{tabular}
}

  \medskip
  \begin{minipage}[t]{0.48\textwidth}
    \centering
    \emph{(c)}\\[2pt]
    \resizebox{\linewidth}{!}{\begin{tabular}{@{}r rrrr rrrr@{}}
    \toprule
    & \multicolumn{4}{c}{$f_{\text{miss}}$ (\%)} &
      \multicolumn{4}{c}{Total violated (\%)} \\
    \cmidrule(lr){2-5} \cmidrule(lr){6-9}
    $N$ & NEC & WEC & SEC & DEC & NEC & WEC & SEC & DEC \\
    \midrule
      25 & 1.8 & 15.7 & 28.4 & 28.4 & 87.3 & 87.3 & 87.3 & 100.0 \\
      50 & 1.6 & 15.6 & 28.0 & 28.5 & 87.0 & 87.0 & 87.0 & 99.9 \\
     100 & 1.5 & 15.6 & 27.8 & 28.7 & 86.9 & 86.9 & 86.9 & 99.9 \\
    \bottomrule
\end{tabular}
}
  \end{minipage}\hfill
  \begin{minipage}[t]{0.48\textwidth}
    \centering
    \emph{(d)}\\[2pt]
    \resizebox{\linewidth}{!}{\begin{tabular}{@{}l rrrrr@{}}
    \toprule
    & \multicolumn{5}{c}{$N_{\text{starts}}$} \\
    \cmidrule(l){2-6}
    Metric & 1 & 2 & 4 & 8 & 16 \\
    \midrule
    \multicolumn{6}{@{}l}{\textit{Min robust WEC margin}$^{\dagger}$} \\[2pt]
    Alcubierre & $-3460.2$ & $-3460.2$ & $-3460.2$ & $-3460.2$ & $-3460.2$ \\
    Rodal & $-2.6\times10^{-6}$ & $-2.6\times10^{-6}$ & $-2.6\times10^{-6}$ & $-2.6\times10^{-6}$ & $-2.6\times10^{-6}$ \\
    \midrule
    \multicolumn{6}{@{}l}{\textit{Missed WEC (\%)}} \\[2pt]
    Alcubierre & 0.00 & 0.00 & 0.00 & 0.00 & 0.00 \\
    Rodal & 15.60 & 15.60 & 15.60 & 15.60 & 15.60 \\
    \bottomrule
\end{tabular}
}\\[3pt]
    {\footnotesize $^{\dagger}$Robust margin: algebraic slack at
      Type~I points; capped extremum at non-Type~I points.}
  \end{minipage}

  \medskip
  \emph{(e)}\;\resizebox{0.9\columnwidth}{!}{\begin{tabular}{@{}lcccc@{}}
    \toprule
    Quantity & $N\!=\!25$ & $N\!=\!49$ & $N\!=\!97$
      & $\max_i|Q_i-\overline Q|$ \\
    \midrule
    Sampled min NEC
      & $-0.631$ & $-0.631$ & $-0.631$
      & $0$ \\
    Integrated viol.\
      & $1.590$ & $1.721$ & $1.652$
      & $0.066$ \\
    $L^2$ viol.\ norm
      & $0.75$ & $0.77$ & $0.73$
      & $0.020$ \\
    \bottomrule
\end{tabular}
}

\end{table}

\subsection{Curvature identity validation}
\label{sec:identities}

As an independent check of the autodiff curvature chain, the computed Riemann
tensor satisfies the algebraic symmetries and the first Bianchi identity to
residual $<10^{-10}$, and the contracted Bianchi identity $\nabla_a G^{ab}=0$ to
$<10^{-6}$ (third-derivative cancellation) at representative Schwarzschild and
Alcubierre points; $T_{ab}=G_{ab}/8\pi$ is symmetric to a relative $<10^{-13}$.

\subsection{Optimization stability}
\label{sec:nstarts}

To measure sensitivity to the multi-start budget, we
vary the number of restarts $N_{\text{starts}} \in \{1, 2, 4, 8, 16\}$ at
$v_s = 0.5$ for Alcubierre and Rodal, recomputing only the
observer search on a fixed curvature grid with independent (non-nested)
random draws per setting.  As Table~\ref{tab:convergence_panel}(d) shows, the minimum
margins and missed-violation fractions are insensitive to the multi-start
budget: Alcubierre and Rodal are $N_{\text{starts}}$-independent from a
single start (Rodal, $100\%$ Type~I, bypasses the optimizer entirely).
The default observer-search budget is $N_{\text{starts}} = 8$;
the matched benchmark of Table~\ref{tab:matched_benchmark}(b) uses
$N_{\text{starts}} = 2$.  This ablation covers Alcubierre and Rodal only and
does not establish optimization stability for the other two drives.

\section{Discussion}
\label{sec:discussion}

Single-frame (Eulerian) evaluation can systematically understate the
energy-condition structure of warp drive spacetimes, and the way it does
so is geometry-dependent.  For the
everywhere-Type~I Rodal metric, presented in~\cite{rodal2026} as
predominantly positive in invariant energy density while still violating the
null energy condition, the Eulerian diagnostic misses most wall weak- and dominant-energy
violations (Section~\ref{sec:invariant-benchmark}); the Type-I inequalities define the exact pointwise test, while the reported
fractions and ablation spreads remain numerical
(Appendix~\ref{app:rodal_ablation}).
The sampled Alcubierre and Nat\'ario walls also contain Type-IV regions,
where a rest-frame energy density does not exist.  Van~den~Broeck is intermediate (WEC $55\%$,
DEC $51\%$ wall miss at $v_s = 0.5$).  The contraction magnitude also depends on the observer.  At NEC-violating
Type~I points it tends to $-\infty$ along suitable timelike boosts, with
$\gamma^2\simeq e^{2\zeta}/4$ (Section~\ref{sec:optimization}).  The algebraic
energy-condition verdict requires no rapidity cap.

\paragraph{Physical interpretation.}
In a Type-I rest frame the WEC density along a principal spatial direction is
$\rho+(\rho+p_i)\sinh^2\zeta$.  If $\rho>0$ and $\rho+p_i<0$,
it becomes negative when $\sinh^2\zeta>\rho/|\rho+p_i|$.
The angle of that spatial direction is measured within the rest frame;
an Eulerian rapidity cap is a different restriction unless the two frames
coincide.  The DEC uses the separate energy-flux test in
Section~\ref{sec:optimization}, and its Type-I verdict follows directly
from $\rho\ge|p_i|$.

\paragraph{Scope of the physical conclusion.}
\label{sec:limitations}
Energy-condition satisfaction requires a criterion covering every admissible
causal vector.  The Type-I eigenvalue inequalities and the LMI reductions
provide such criteria.  In particular, positive Eulerian energy density does
not ensure satisfaction, while an established Type-IV tensor violates the
NEC and hence the other three conditions.

The no-go results of Santiago--Schuster--Visser~\cite{santiago2021} apply under
their stated geometric and subsidiary hypotheses.  Our calculation concerns
the stress-energy induced by each specified metric.  Matter realizability,
global causal properties, and curved-space quantum inequalities require
additional analysis.  The finite-segment and composite diagnostics retain
the normalization and sampling limits stated in
Sections~\ref{sec:scope} and~\ref{sec:construction-exoticity}.

\section{Conclusion}
\label{sec:conclusion}

The classical S-lemma gives pointwise tests of all four energy conditions
without a Hawking--Ellis type assumption or rapidity cap.  Interval evaluation
of the metric and curvature connects these algebraic criteria to the induced
stress-energy.  At the reference parameters, the global bounds establish a
negative normalized null-energy infimum in each of the four benchmark walls.
With the specified Eulerian normalization, the bounds place Nat\'ario and
Van~den~Broeck below Alcubierre; they do not order the other pairs.

The momentum constraint explains why an irrotational shift on time-independent
Euclidean slices with unit lapse gives Type-I stress-energy.  For smooth
shifts on all of $\mathbb R^3$ with bounded vorticity, momentum vanishes
identically precisely for a gradient plus rigid rotation.  The momentum-plane
type criterion also requires the stated stress-axis splitting hypothesis.
The ideal Rodal profile is Type~I everywhere, yet the Eulerian diagnostic
misses about $73\%$ of its sampled wall WEC violations and $74\%$ of its DEC
violations.  These fractions and the other drives' vorticity/type association
remain numerical results for the specified slices and masks.

For speed-linear shifts on fixed Euclidean slices, integrated negative
Eulerian energy scales exactly as $v_s^2$ when finite.  A quadratic pointwise
eigenvalue-margin law follows when $j=0$, with a fixed comoving profile and
domain.  The leading Kretschmann coefficient depends on all spatial
derivatives of shift vorticity, while the Weyl-squared coefficient has a
nontrivial kernel.  Finite-speed regressions remain distinct from these
coefficient identities.

The regularized Fuchs and Garattini--Zatrimaylov panel compares kinematic
choices only.  The null-geodesic integrals are finite-segment, basin-local
diagnostics with selected-ray step refinement; they supply no
complete-geodesic ANEC conclusion.  The flat-space quantum-inequality
estimates likewise supply no curved-spacetime bound.

The matrix criteria apply to any symmetric stress-energy tensor on a
four-dimensional Lorentzian spacetime.  \warpax{} implements them for prescribed metrics, making the
distinction between positive energy for one observer and satisfaction of an
energy condition explicit.  Matter realizability and global causal viability
require additional analysis.

\ack{This work is financially supported by VinUniversity under the
Environmental Intelligence (CEI) Grant (No.\ VUNI.CEI.FS\_0009).
The author thanks J.~Rodal for clarifications regarding frame
conventions, the $r_s \to 0$ limit of the irrotational angular
profile, and the Cartesian form of the shift vector.
The author declares no competing interests.}

\data{\warpax{} is freely available under the MIT license at
\url{https://github.com/anindex/warpax}.  The source and numerical data support the calculations reported here.
Archived releases are available at
\url{https://doi.org/10.5281/zenodo.18715933}.  No external datasets are required.}

\appendix

\section{Numerical implementation}
\label{app:toolkit}

\warpax{} is a Python package built on JAX~\cite{jax2018} and
Equinox~\cite{kidger2021equinox}.  Each metric is a map $x^\mu\mapsto g_{ab}(x)$:
the four analyzed here (Alcubierre~\cite{alcubierre1994},
Nat\'ario~\cite{natario2002}, Van~den~Broeck~\cite{vandenbroeck1999},
Rodal~\cite{rodal2026}) and two benchmarks (Minkowski, and Schwarzschild in
isotropic Cartesian coordinates), with the full ADM decompositions in
Appendix~\ref{app:metric_definitions}.  Derivatives along the curvature chain are
taken by forward-mode automatic differentiation, so the derivatives of the
implemented metric are exact up to roundoff.  The JAX calculations use
64-bit floating point; interval calculations use the precision stated in
Appendix~\ref{app:interval-census}.

\paragraph{Metric regularity.}
The warp metrics use $C^\infty$ $\tanh$-based shaping.  Radial expressions divide
by $\sqrt{r_s^2+\varepsilon^2}$ to remain $C^\infty$ at the center, with
$\varepsilon^2=10^{-60}$ in the shape and angular profiles and the coarser
$\varepsilon^2=10^{-12}$ in the divisor that builds the unit radial direction;
the residual curl from the coarser regularization is quantified in
Appendix~\ref{app:rodal_ablation}.  Curvature is computed from derivatives of the implemented profile.

\paragraph{Pointwise evaluation.}
For a given metric and grid we evaluate the curvature chain $g_{ab}\to T_{ab}$,
classify each point by Hawking--Ellis type (I, II, III or IV), compute the algebraic
slacks \eqref{eq:nec-type1}--\eqref{eq:dec-type1} at Type~I points, use the
multi-start observer optimization for non-Type-I diagnostics (Section~\ref{sec:optimization}),
and report the verdicts.  At Type~I the algebraic slacks decide and at Type~IV the
absence of a causal eigenvector forces the null violation; the capped extrema report
severity only.  Every non-Type-I point, whatever its label, takes its Boolean
verdict from Theorem~\ref{thm:lmi}.  The Eulerian baseline (six axis-aligned null
rays for the null condition, the ADM normal for the timelike ones) is, for any fixed
observer $u_0$, an upper bound on the worst-case margin, $m^*(x)\le m(x,u_0)$, so its
violation set lower-bounds the true extent; that inequality is enforced as a
consistency check.

\section{Rodal DEC ablation study}
\label{app:rodal_ablation}

\paragraph{Regularization of the coordinate origin.}
The Rodal metric is implemented with two $C^\infty$ regularizations
$\sqrt{r^2 + \varepsilon^2}$ of the coordinate radius at $r=0$: the shape and angular
profiles take $\varepsilon^2 = 10^{-60}$, and the divisor that builds the unit radial
direction takes the coarser $\varepsilon^2 = 10^{-12}$, without which
$\partial_i n_j$ diverges at the bubble center.  The coarser regularization breaks exact
irrotationality: replacing $x^i/r$ by $x^i/\sqrt{r^2+\varepsilon^2}$ leaves a residual
curl linear in $\varepsilon^2$.  At the sampled wall points,
$|\nabla\times\beta|/|\beta|$ is at most $10^{-10}$ in the coordinates used,
ten orders below the $O(1)$ values of the vortical drives.
Lemma~\ref{lem:rodal} proves Type~I for the ideal analytic profiles.
It does not establish the type of the separately regularized implementation;
the sampled type labels and residual-curl checks supply numerical evidence
for that implementation.  A
controlled ablation with three single-variable sweeps (resolution,
shape-function regularization, wall thickness; Alcubierre as control) shows the
$\sim\!28.5\%$ DEC miss rate is insensitive to resolution (0.32~pp) and to
$\varepsilon^2$ in the shape function (0.0~pp across
$10^{-24}$ to $10^{-6}$).  The separate wall-thickness study compares
proper-volume-weighted conditional rates over the full domain and within the
wall, as defined below; the Alcubierre control misses $0.0\%$ throughout.

The sweeps are run at baseline $v_s = 0.5$, $N = 50$, $\sigma_{\text{Rodal}} =
0.03$ and $\varepsilon^2 = 10^{-24}$ in the swept shape function, with the
direction divisor held at its production value $10^{-12}$.
Table~\ref{tab:rodal_ablation}(a) reports the resolution and regularization sweeps.
The third, over wall thickness, runs $52.9\to62.7\to75.4\to87.3\%$ in the
wall-restricted conditional rate on the uniform grid, whose two sharpest
$\sigma$ span $1.79$ and $0.598$ cells and so fall below the four-cell floor; the
wall-resolved values quoted in the text are $52.9\to63.2\to76.8\to85.5\%$
(Table~\ref{tab:rodal_ablation}(b)).  The trend agrees in direction and shape on
both, but quantitative conclusions use the wall-resolved sweep.

\begin{table}[tbp]
\centering
\footnotesize
\caption{Rodal DEC ablation.  \emph{(a)}~Resolution and shape-function
  regularization stability of the Rodal grid DEC miss rate at fixed wall thickness
  ($\sigma = 0.03$); the rate is insensitive to both, and the Alcubierre control
  misses $0.0\%$ throughout.  \emph{(b)}~Wall-thickness sweep, wall-resolved across
  its ladder.  Every $\sigma$ is evaluated on the exact
  axisymmetric $(r,\mu)$ reduction, on a four-level ladder
  $(N_r,N_\mu)=(48,48)$, $(64,64)$, $(80,80)$, $(96,96)$ with proper-volume weights
  $2\pi r^2$.  Both DEC miss columns are conditional on violating volume:
  ``Grid'' uses the full sampled ball and ``Wall'' restricts to $0.1\le f\le0.9$.
  ``Wall cells'' is the worst-case count across the $10$--$90\%$ wall
  normal at the coarsest and finest level.  The quoted spread is the largest
  departure from the finest level across the whole ladder, not the gap between the
  two finest, which at $\sigma=0.1$ would understate the drift by a factor of
  fifty.}
\label{tab:rodal_ablation}

\emph{(a)}\;\begin{tabular}{@{}l l c c@{}}
  \toprule
  Sweep & Parameter value & Rodal grid DEC miss (\%) & Alcubierre grid DEC miss (\%) \\
  \midrule
  \multirow{3}{*}{Resolution ($N$)}
    & $25$   & 28.36 & 0.0 \\
    & $50$   & 28.53 & 0.0 \\
    & $100$  & 28.68 & 0.0 \\
  \midrule
  \multirow{4}{*}{Regularization ($\varepsilon^2$)}
    & $10^{-24}$ & 28.53 & -- \\
    & $10^{-18}$ & 28.53 & -- \\
    & $10^{-12}$ & 28.53 & -- \\
    & $10^{-6}$  & 28.53 & -- \\
  \bottomrule
\end{tabular}

\medskip
\emph{(b)}\;\begin{tabular}{@{}r r r r r r@{}}
  \toprule
  & \multicolumn{2}{c}{Wall cells} & Wall & Grid DEC & Wall DEC \\
  \cmidrule(lr){2-3}
  $\sigma$ & coarsest & finest & points & miss (\%) & miss (\%) \\
  \midrule
  $0.01$ & 26.1 & 52.4 & $5{,}856$ & 46.98 & $52.88 \pm 0.42$ \\
  $0.03$ & 13.0 & 25.7 & $2{,}976$ & 34.46 & $63.16 \pm 0.78$ \\
  $0.1$ & 7.5 & 15.5 & $2{,}016$ & 32.10 & $76.78 \pm 1.92$ \\
  $0.3$ & 6.3 & 11.7 & $1{,}920$ & 30.67 & $85.46 \pm 0.53$ \\
  \bottomrule
\end{tabular}

\end{table}

In the wall-thickness study, the full-domain conditional DEC miss rate falls
from $47.0\%$ to $30.7\%$ as the wall sharpens.  This is the proper-volume-weighted
fraction of DEC violations missed in the full sampled ball, distinct from the
unconditional Cartesian-grid fraction in equation~\eqref{eq:fmiss}.
Restricting both numerator and denominator to violating wall points gives
the wall rate, which rises from $52.9\%$ to $85.5\%$
(Table~\ref{tab:rodal_ablation}(b)).  Thus the two sampled regions show different
trends, and the Eulerian miss persists within the wall itself.
No miss is detected for the Alcubierre control.

The trend is measured on a resolved wall.  The uniform $N=50$ grid falls below the
four-cell criterion at the two sharpest walls, and
reaching four cells at $\sigma=0.3$ on a uniform Cartesian grid would require
$N\ge329$, that is ${\sim}3.6\times10^7$ tensor evaluations.
The exact axisymmetric reduction costs a few thousand points for the same physics,
so every $\sigma$ spans at least $6.3$ cells at the coarsest ladder level and the
conditional rate is converged to better than $0.3$ percentage points between the
two finest levels.  Against the uniform grid the endpoints move by at most $2$
points and the rise is monotone throughout.

\section{Global enclosures of the wall null deficit}
\label{app:enclosures}

A local search can find a negative null contraction but cannot exclude a
deeper minimum elsewhere.  To bound the global wall infimum, we cover the
whole wall with interval bounds and refine the regions that could contain
the minimum.

\paragraph{Exact dimensional reduction.}
The four matched Alcubierre, Nat\'ario, Van~den~Broeck and Rodal metrics are
invariant under rotations about the propagation axis.  Their Eulerian normal
and radial wall mask share this symmetry, so the Eulerian-normalized null
minimum defined below depends only on $(x,s)$ with $s=(y^2+z^2)^{1/2}$.
Each orbit has a representative with $z=0$, $y=s\ge0$; restricting this
objective and mask to that half-plane is therefore exact.

\paragraph{The bounded objective.}
We bound the null-energy deficit in the form that is well defined at every point
irrespective of algebraic type.  Writing a future null direction as $k^a=n^a+v^a$
with $v^a$ orthogonal to $n^a$ and of unit length in the induced metric,
\begin{equation}
  q(v)=T_{ab}k^ak^b=\rho_n-2\,b\!\cdot\!v+v^{\mathsf T}Sv,
  \qquad b_i=-T_{ab}n^ae_i^b,\quad S_{ij}=T_{ab}e_i^ae_j^b,
  \label{eq:enclosure-objective}
\end{equation}
and the pointwise deficit is $\mathcal{N}(x)=\min_{|v|=1}q(v)$, whose sign decides the
null energy condition at $x$.  At a point with vanishing Eulerian momentum,
$b=0$ and $\mathcal{N}=\min_i(\rho+p_i)$.  This identity applies to the ideal
irrotational Rodal profile.  The enclosed implementation uses the radial
regularizations of Appendix~\ref{app:rodal_ablation}, so its bound concerns
$\mathcal N$ directly and does not assume exact equality with an eigenvalue slack.

\paragraph{Method.}
The bounds combine interval differentiation, null-vector evaluation, and
a matrix lower bound.  Interval forward-mode differentiation through second metric derivatives
encloses the curvature chain
$g\to\Gamma\to R^a{}_{bcd}\to R_{ab}\to G_{ab}\to T_{ab}$.
Repeated occurrences of the same variable can widen these enclosures;
centered mean-value bounds reduce this dependency overestimation as described below.
The only transcendental in the shape function is $\tanh$, and
$\tanh u=1-2/(e^{2u}+1)$ contains $u$ exactly once, so its interval evaluation is
free of dependency overestimation.  Perturbing the Eulerian data changes
$\mathcal{N}$ by at most $|\delta\rho_n|+2|\delta b|+\|\delta S\|_2$.
Thus $\mathcal{N}$ is $1$-Lipschitz with respect to this weighted norm.

To obtain an upper bound, a floating-point secular solve proposes a spatial
direction at a point whose wall membership is established.  We normalize that
direction and evaluate \eqref{eq:enclosure-objective} in interval arithmetic.
Any such null contraction bounds the wall infimum from above; a suboptimal
direction weakens the bound but preserves its validity.  The lower bound
combines two independent inequalities valid throughout each box, as described
under \emph{Coverage} below.

The search uses a Moore--Skelboe branch-and-bound: a priority queue of boxes
ordered by the lower end of the enclosure, splitting the most promising box and
discarding boxes that cannot contain the global minimum.  On exit
\begin{equation}
  \mathcal{N}_{\rm lower}\;\le\;\inf_{\rm wall}\mathcal{N}\;\le\;\mathcal{N}_{\rm upper},
  \label{eq:enclosure-bracket}
\end{equation}
with the lower end valid on the \emph{continuum} rather than at sample points.  A box
is dropped from the wall mask only when its shape-function enclosure lies entirely outside
$f\in[0.1,0.9]$, so no wall point is dropped, including points on the curved mask
boundary; and a box straddling the coordinate axis $r_s=0$, where the spherical form
carries a removable $1/r_s$, is reported as unbounded and split.
The exterior of the search box is excluded by a single interval evaluation.
The wall restriction $0.1\le f\le0.9$ is a condition on the
shape function alone, and $f$ depends on position only through
$r_s=|\mathbf{x}-\mathbf{x}_s(t)|$.  The top hat is strictly decreasing in $r_s$ for
$r_s>0$, since $\mathrm{d}f/\mathrm{d}r_s\propto
\mathrm{sech}^2[\sigma(r_s{+}R)]-\mathrm{sech}^2[\sigma(r_s{-}R)]<0$ there, so one
interval evaluation at $r_s=3$ bounds $f$ on the whole unbounded exterior by
$[0,\,1.3\times10^{-14}]$, disjoint from the band.  The
exterior therefore contains no wall point in \emph{any} direction, $x<-3$ included.

\paragraph{Result.}
Table~\ref{tab:enclosures} reports the two ends of \eqref{eq:enclosure-bracket}.

At the reported search point $x_*$ we separately enclose the pointwise deficit
$\mathcal N(x_*)$.  Its midpoint is printed as an achieved numerical value,
with the local interval width.  This error at a fixed point is distinct from
uncertainty in the global infimum: the latter lies in the full bracket
$[\mathcal N_{\rm lower},\mathcal N_{\rm upper}]$.
The upper endpoint is a rigorous bound obtained from an evaluated null
contraction, and need not equal an exactly attained minimum.  Agreement with
basin-local polishing does not reduce the uncertainty in the global infimum.

The lower endpoint is valid over the whole wall.  Interval dependency and the
finite search budget can leave a wide global bracket even when the local
point evaluation is sharp.  The table reports both scales.

\paragraph{Coverage.}
The interval calculation covers all four drives, including
Van~den~Broeck's conformally flat spatial metric
$\gamma_{ij}=B^2\delta_{ij}$.  All four searches yield finite brackets
(Table~\ref{tab:enclosures}).

The full enclosures establish a negative wall infimum for every drive.
Nat\'ario and Van~den~Broeck each have a strictly lower infimum than Alcubierre.
The other pairs have overlapping brackets and are not ordered by these bounds.
These comparisons use the specified Eulerian normalization and slice.
Every search exhausts its evaluation budget before reaching the chosen target
interval width $10^{-4}$.

The lower bound is the larger of two independently valid bounds.  The first is
$\rho_{\rm lo}-2|b|_{\max}+\lambda_{\min,\rm lo}(S)$, with the spatial eigenvalue
bounded by the Gershgorin inequalities in interval arithmetic.  It bounds the
linear and quadratic terms separately and need not become sharp as a box
shrinks.  The second uses
$\mathcal{N}\ge2\lambda_{\min}(M(\sigma))$ for any $\sigma$, with the eigenvalue
bounded by an interval $LDL^{\mathsf T}$ test.
Refinement can tighten this bound when the interval pivots
remain valid; multiplier choice and interval overestimation both leave slack.
Corollary~\ref{cor:margin} gives equality for an exact best multiplier and exact
eigenvalue.

Where interval gradients are available, each orthonormal stress component
$h\in\{\rho_n,b_i,S_{ij}\}$ also has the centered mean-value enclosure
\begin{equation*}
 h(B)\subseteq h(c)+\partial_x h(B)(B_x-c_x)
                       +\partial_s h(B)(B_s-c_s),
\end{equation*}
where $B$ is a box in $(x,s)$ and $c$ its center.  Differentiating the stress
components requires third metric derivatives on the box; the four enclosed
wall profiles are smooth there.  Both lower bounds
are evaluated with these enclosures as well, and the largest valid bound is
retained.  If the tightening cannot be evaluated, the ordinary interval bounds
remain valid.  The objective uses orthonormal spatial components and contractions
with the Eulerian normal.  Spatial coordinate components coincide with that
orthonormal basis when $\gamma_{ij}=\delta_{ij}$; time and flux components still
require the normal.  On Van~den~Broeck's conformally flat slice the chain
Cholesky-factors $\gamma=LL^{\mathsf T}$ to construct the orthonormal frame.

The bracket bounds $\mathcal{N}$ on the continuum of the wall region at fixed
physical parameters.  Note that it is a statement about $\mathcal{N}$ and not about
the Hawking--Ellis type assignment, which remains numerically classified and
tolerance-dependent (Section~\ref{sec:scope}).

\begin{table}[tbp]
  \centering
  \footnotesize
  \caption{Rigorous global search for the wall null-energy deficit
    $\mathcal{N}=\min_{|v|=1}q(v)$ at matched parameters ($R_b=1$, $\sigma=8$,
    $v_s=0.5$), by interval branch-and-bound on the axisymmetry-reduced domain.
    Evaluation budgets are $120{,}000$ for Alcubierre and $4096$ for the other
    drives; the target width $10^{-4}$ is not reached.
    $\mathcal N(x_*)$ is the midpoint of a separate point enclosure, whose width
    gives the local numerical uncertainty; $\mathcal N_{\rm upper}$ bounds
    the global infimum from above;
    $\mathcal{N}_{\rm lower}$ is a continuum bound over the wall
    band $f\in[0.1,0.9]$.  $\mathcal{N}$ includes all algebraic types; the
    sampled Type-I margins are reported separately.}
  \label{tab:enclosures}
  \begin{tabular}{@{}l r rr r@{}}
  \toprule
  Metric & $\mathcal N(x_*)$ & $\mathcal N_{\rm lower}$ & $\mathcal N_{\rm upper}$ & local interval width \\
  \midrule
  Alcubierre & $-0.631077$ & $-0.631624$ & $-0.631076$ & $2.22e-16$ \\
  Rodal & $-0.194141$ & $-23.496805$ & $-0.194141$ & $1.07e-14$ \\
  Nat\'ario & $-5.212534$ & $-1419.173790$ & $-5.212534$ & $1.78e-15$ \\
  Van den Broeck & $-1.281561$ & $-35.315137$ & $-1.281561$ & $8.88e-16$ \\
  \bottomrule
\end{tabular}

\end{table}

\section{Metric definitions}
\label{app:metric_definitions}

For reproducibility, we collect the ADM~$3+1$ decompositions of the
test metrics.  In each case the line element takes the general form
\begin{equation}
  \dd s^2 = -\alpha^2\,\dd t^2
    + \gamma_{ij}\bigl(\dd x^i + \beta^i\,\dd t\bigr)
                 \bigl(\dd x^j + \beta^j\,\dd t\bigr),
  \label{eq:adm-general}
\end{equation}
where $\alpha$ is the lapse, $\beta^i$ the shift vector, and
$\gamma_{ij}$ the spatial metric.  The four benchmark metrics
share the bubble center $x_s(t) = v_s\,t$ and the comoving distance
$r_s = \sqrt{(x - x_s)^2 + y^2 + z^2}$; the Garattini--Zatrimaylov bubble instead
moves radially with $r(t)=r_0e^{Ht}$ and $v=Hr(t)$, as set out below.  The Fuchs shell uses stationary comoving coordinates
and $r=|\mathbf x|$.
The standard $\tanh$-smoothed top-hat shape function is
\begin{equation}
  f(r_s) = \frac{\tanh\bigl[\sigma(r_s + R_b)\bigr]
                - \tanh\bigl[\sigma(r_s - R_b)\bigr]}
               {2\,\tanh(\sigma R_b)}\,,
  \label{eq:shape-tanh}
\end{equation}
with $f(0) = 1$ and $f(r_s) \to 0$ for $r_s \gg R_b$.

\paragraph{Alcubierre.}
$\alpha = 1$, $\gamma_{ij} = \delta_{ij}$,
$\beta^i = (-v_s\,f(r_s),\;0,\;0)$.
The shape function $f$ is given by
equation~(\ref{eq:shape-tanh}) \cite{alcubierre1994}.

\paragraph{Van~den~Broeck.}
$\alpha = 1$, $\beta^i = (-v_s\,f(r_s),\;0,\;0)$,
$\gamma_{ij} = B^2(r_s)\,\delta_{ij}$, where
$B(r_s) = 1 + \alpha_{\text{vdb}}\,f_B(r_s)$.  The inner shape
function $f_B$ uses a separate radius $\tilde{R}$ and width
$\sigma_B$ \cite{vandenbroeck1999}:
\begin{equation}
  f_B(r_s) = \frac{\tanh\bigl[\sigma_B(r_s + \tilde{R})\bigr]
                  - \tanh\bigl[\sigma_B(r_s - \tilde{R})\bigr]}
                 {2\,\tanh(\sigma_B\tilde{R})}\,.
  \label{eq:vdb-conformal}
\end{equation}

\paragraph{Nat\'ario.}
$\alpha=1$, $\gamma_{ij}=\delta_{ij}$.  Nat\'ario writes the rest-frame metric as
$\dd s^2=-\dd t^2+|\dd\boldsymbol\xi-X(\boldsymbol\xi)\dd t|^2$,
with $X(0)=0$ and $X(\infty)=-v_s\mathbf e_x$~\cite{natario2002}.
Under $\mathbf x=\boldsymbol\xi+v_st\mathbf e_x$, the $+\beta$ laboratory shift is
$\beta_{\rm lab}=-X(\mathbf x-v_st\mathbf e_x)-v_s\mathbf e_x$, hence
\begin{align}
 \beta^x&=v_s[2n(r_s)+r_sn'(r_s)\sin^2\theta-1],\nonumber\\
 \beta^y&=-v_s n'(r_s)\frac{(x-v_st)y}{r_s},\qquad
 \beta^z=-v_s n'(r_s)\frac{(x-v_st)z}{r_s}.
 \label{eq:natario-shift}
\end{align}
Here $n=(1-f)/2$, with $f$ from equation~\eqref{eq:shape-tanh}, and
$\sin^2\theta=(y^2+z^2)/r_s^2$.  This analytic profile approaches the constant
interior/exterior profiles used in Ref.~\cite{natario2002}; it is not compactly
supported.  The origin uses the analytic even extension in $r_s^2$.
The shift is divergence-free, tends to zero at infinity, and equals
$-v_s\mathbf e_x$ at the center.  Thus the center worldline
$C(t)=(t,v_st,0,0)$ has $g(\dot C,\dot C)=-1$ for every speed.
The rest/laboratory change leaves the slices $t=\mathrm{constant}$ unchanged;
the normal is the same geometric observer expressed in different coordinates.
All numerical results here use the laboratory form.

\paragraph{Rodal.}
$\alpha = 1$, $\gamma_{ij} = \delta_{ij}$.
The irrotational (curl-free) shift is specified in a spherical
tetrad aligned with the propagation axis ($\theta$ measured from
$+x$) \cite{rodal2026}:
\begin{equation}
  \hat{\beta}_r = -v_s\,F(r_s)\cos\theta\,,\qquad
  \hat{\beta}_\theta = v_s\,G(r_s)\sin\theta\,,
  \label{eq:rodal-shift}
\end{equation}
where $F = f$ (equation~\ref{eq:shape-tanh}) and
\begin{equation}
  G(r_s) = 1 - \frac{
    2\,r_s\sigma\sinh(R_b\sigma)
    + \cosh(R_b\sigma)\bigl[
      \ln\cosh\sigma(r_s{-}R_b) - \ln\cosh\sigma(r_s{+}R_b)
    \bigr]
  }{2\,r_s\sigma\sinh(R_b\sigma)}\,.
  \label{eq:rodal-G}
\end{equation}

\paragraph{Note on conventions (Rodal).}
Ref.~\cite{rodal2026} presents the irrotational construction in a
Nat\'ario-style ``bubble-at-rest'' convention, with radial profiles
$f_{\rm paper}(0)=g_{\rm paper}(0)=0$ and $f_{\rm paper}(\infty)=g_{\rm paper}(\infty)=1$.
We standardize to a lab-frame convention with vanishing shift at infinity:
\begin{equation}
  F(r_s)=1-f_{\rm paper}(r_s),\qquad G(r_s)=1-g_{\rm paper}(r_s),
\end{equation}
so that $F(0)=G(0)=1$ and $F(\infty)=G(\infty)=0$.
Subtracting a uniform asymptotic translation preserves irrotationality
(a uniform translation field is curl-free).
We use the ``$+\beta$'' ADM convention
$ds^2=-\alpha^2\,dt^2+\gamma_{ij}(dx^i+\beta^i\,dt)(dx^j+\beta^j\,dt)$;
Ref.~\cite{rodal2026} uses ``$-\beta$'', so component-level comparison
requires the corresponding sign flip.

\paragraph{Removable $r_s\to 0$ form in $G(r_s)$.}
The apparent $1/r_s$ factor in equation~(\ref{eq:rodal-G}) is a removable $0/0$
form of the spherical-basis representation.  Defining
$\Delta(r_s) \equiv \ln\cosh\bigl[\sigma(r_s{-}R_b)\bigr]
  - \ln\cosh\bigl[\sigma(r_s{+}R_b)\bigr]$,
one finds $\Delta'(0) = -2\sigma\tanh(\sigma R_b)$, whence
$\lim_{r_s\to 0} G(r_s) = 1$.
The irrotational construction is therefore regular at the bubble center; the
$r_s \mapsto \sqrt{r_s^2 + \varepsilon^2}$ substitution is a numerical stability
device for automatic differentiation, not a physical repair.

\paragraph{Cartesian shift components.}
In Cartesian coordinates, equations~(\ref{eq:rodal-shift}) become
\begin{equation}
  \boldsymbol{\beta} = -v_s\bigl[G(r_s)\,\hat{\boldsymbol{x}}
    + \bigl(F(r_s)-G(r_s)\bigr)\,n_x\,\boldsymbol{n}\bigr],
  \label{eq:rodal-cartesian}
\end{equation}
where $\boldsymbol{n} = (\Delta x,\,y,\,z)/r_s$ is the radial unit
vector from the bubble center and $n_x = \Delta x / r_s$.
The supplied profiles are even analytic functions near $r_s=0$, with
$F(0)=G(0)=1$ and $F-G=O(r_s^2)$.  Thus $(F-G)/r_s^2$ is smooth in $r_s^2$,
and $(F-G)n_xn_i=[(F-G)/r_s^2]\Delta x\,(x_i-x_{s,i})$ extends smoothly to the center.
The shift there is $\boldsymbol{\beta}(0)=-v_s\,\hat{\boldsymbol{x}}$.

\paragraph{Fuchs constant-velocity shell.}
Fuchs \emph{et al.}~\cite{fuchs2024constvel} begin with a static matter-shell
construction before modifying the metric.  A shell of constant density occupies $R_1\le r\le R_2$ with total
mass $M$; the Tolman--Oppenheimer--Volkoff equation
\begin{equation}
  \frac{\dd p}{\dd r} = -\frac{(\rho+p)\,(m+4\pi r^3 p)}{r\,(r-2m)},
  \qquad p(R_2)=0,
\end{equation}
is integrated inward for the isotropic pressure; density and pressure are then
smoothed by four passes of a moving average with width ratio
$s_\rho/s_P\simeq1.72$, the cumulative mass $m(r)$ recomputed from the smoothed
density, and the metric potentials obtained from
\begin{equation}
  e^{2b(r)} = \Bigl(1-\frac{2m}{r}\Bigr)^{-1},
  \qquad
  \frac{\dd a}{\dd r} = \frac{m+4\pi r^3 p}{r\,(r-2m)} ,
\end{equation}
with the Schwarzschild boundary condition $e^{2a}=e^{-2b}=1-2M/r$ outside the
solved grid, which we impose analytically rather than by extrapolating the
numerical solution.  Following equation~(26) of Ref.~\cite{fuchs2024constvel}, the warp
modification changes only covariant $g_{01}=g_{10}$:
\begin{equation}
 g_{00}=-e^{2a(r)},\qquad g_{0x}=-v_sS_{\rm warp}(r),\qquad
 g_{0y}=g_{0z}=0,\qquad
 g_{ij}=\gamma_{ij}=\delta_{ij}+(e^{2b(r)}-1)n_i n_j .
 \label{eq:fuchs-covariant}
\end{equation}
Here $n_i=x_i/r$ in stationary comoving coordinates.  The compatible ADM
quantities are
$\beta^i=-v_sS_{\rm warp}\gamma^{ix}$ and
$\alpha^2=e^{2a}+v_s^2S_{\rm warp}^2\gamma^{xx}$.
For positive-definite $\gamma$ this lapse is positive and
$\det g=-\alpha^2\det\gamma<0$.

We replace the published buffered clamps by the natural-endpoint extension
\begin{equation}
 S_{\rm warp}(r)=
 \begin{cases}
  1,&r\le R_1,\\
  [1+\exp(-\chi(r))]^{-1},&R_1<r<R_2,\\
  0,&r\ge R_2,
 \end{cases}\qquad
 \chi=(R_2-R_1)\left(\frac1{r-R_2}+\frac1{r-R_1}\right).
 \label{eq:fuchs-sigmoid}
\end{equation}
All derivatives of $S_{\rm warp}-1$ vanish at $R_1$, and all derivatives of
$S_{\rm warp}$ vanish at $R_2$.  This $C^\infty$ extension is the
manuscript's regularization of the published prescription.  The extension uses the physical shell endpoints.

The numerical radial potentials use a cubic interpolant with not-a-knot
boundary conditions; its exact spline equations enforce $C^2$ continuity.  To join it to
constant interior potentials $(a_0,0)$ and analytic Schwarzschild exterior
potentials, we blend with $Q(u)=6u^5-15u^4+10u^3$, whose first two derivatives
vanish at both endpoints.  The inner blend occupies $[R_1/4,R_1/2]$; the outer
blend occupies $[(R_2+r_{\max})/2,r_{\max}]$, with $r_{\max}=1.5R_2$.
These joins are also this manuscript's regularization.  In exact arithmetic
the metric is at least $C^2$, including its center and all interpolation and
matching radii.  Stored binary64 spline coefficients leave second-derivative
matching residuals of at most $1.46\times10^{-14}$ on the tested radial grids.

The published parameters are $R_1=10\,$m, $R_2=20\,$m,
$M=4.49\times10^{27}\,$kg ($2M=6.668692\,$m), and $v_s=0.02$.
We use four smoothing passes with fixed physical density smoothing width
$\sigma_\rho=0.06(R_2-R_1)$ and pressure width
$\sigma_P=\sigma_\rho/1.72$.  The cell-integrated moving-average kernel has
span $\sqrt{12}\sigma$; these widths and all joins are held fixed when the radial
grid is refined.  Production uses $2048$ radial points.  Curvature probes use
$512$, $1024$, $2048$ and $4096$ points; the construction-panel diagnostics
use $1024$, $2048$ and $4096$.  The sigmoid's $10$--$90\%$ crossings remain
$12.790029$ and $17.209971\,$m, giving $R_c=15\,$m and $W=4.419943\,$m.

\paragraph{Garattini--Zatrimaylov de~Sitter bubble.}
The construction of~\cite{garattini2025desitter} is a shift-only metric on
\emph{flat} slices (de~Sitter in Painlev\'e--Gullstrand form, with the
expansion carried by the background shift rather than by a scale factor), into
which the bubble is cut by interpolation:
\begin{equation}
  \alpha = 1,\qquad
  \gamma_{ij} = \delta_{ij},\qquad
  \beta^i = -\bigl(1-f(r_s)\bigr)\,\frac{x^i}{L} - f(r_s)\,v^i,
  \qquad L = H^{-1},
  \label{eq:gz-shift}
\end{equation}
with $f$ the same $\tanh$ top-hat used by the compact family and $\Lambda=3H^2$.
The bubble co-moves with the Hubble flow, $r(t)=r_0e^{Ht}$ and $v=Hr(t)$, which is
the paper's matching condition; substituting this relation into
\eqref{eq:gz-shift} gives
$\beta = -x/L + f(r_s)\,(x-x_s)/L$, a sum of two radial gradients and hence
\emph{irrotational}.  By Lemma~\ref{lem:jcurl} the Eulerian momentum density then
vanishes identically and the wall is Hawking--Ellis Type~I, which is the second case
of Lemma~\ref{lem:rodal}.  The wall fractions in Section~\ref{sec:construction-exoticity} therefore use
$100.0\%$ Type~I and $0.0\%$ Type~II.  This geometry has a de~Sitter background
and is not asymptotically flat.

For the standardized row we choose $H=0.02/15\,$m$^{-1}$,
$r_0=15\,$m, $R_b=15\,$m, $\sigma=0.497115373\,$m$^{-1}$ and $t=0$, so
$v=Hr_0=0.02$.  For the reference row we choose $H=0.1\,$m$^{-1}$,
$r_0=1\,$m, $R_b=1\,$m, $\sigma=8\,$m$^{-1}$ and $t=0$, so $v=0.1$.
The relation $v(t)=Hr_0e^{Ht}$ comes from Ref.~\cite{garattini2025desitter};
these numerical parameter values are choices made here.  We use $G=c=1$,
with time measured in meters and $\Lambda=3H^2$.

\section{Deciding the conditions without an eigendecomposition}
\label{app:slemma}

The energy conditions can be expressed as positivity conditions on the causal
cone.  In particular, the DEC requires $-T^a{}_b$ to preserve the future cone.
Lorentz-cone-preserving maps and their semidefinite descriptions are established
subjects~\cite{loewy1975icecream,hildebrand2007lorentz,hildebrand2011lorentzII};
Bergqvist and Senovilla~\cite{bergqvist2001nullcone} characterize tensors with
the dominant property.  Here we apply the classical S-lemma to obtain the
four pointwise criteria, specifying the multiplier sign for each condition.

\paragraph{Proof of Theorem~\ref{thm:lmi}.}

\begin{proof}
By \eqref{eq:lmi-quadratic} the WEC is exactly $q\ge0$ on the open unit ball,
hence by continuity on the closed ball, whose boundary is the NEC.  Put
$g(w)=1-|w|^2$.  Since $g(0)=1>0$, Slater's condition holds for the single
quadratic constraint $g\ge0$; and the S-lemma
\cite{yakubovich1971,polik2007slemma} gives: $q\ge0$ wherever $g\ge0$ if and only
if $q-\sigma g\ge0$ on all of $\mathbb{R}^3$ for some $\sigma\ge0$.  Writing
$q(w)-\sigma g(w)=(1,w)M(\sigma)(1,w)^{\mathsf T}$, global nonnegativity of that
affine-argument quadratic is equivalent to $M(\sigma)\succeq0$ (vectors with
nonzero first component by scaling, the rest by continuity).  This is~(ii).

Part~(i) is the equality-constrained case, and it is obtained by homogenizing.  Put
$x=(x_0,\tilde w)\in\mathbb{R}^4$ and
\begin{equation*}
  \tilde q(x)=\rho_n\,x_0^2-2x_0\,b\!\cdot\!\tilde w+\tilde w^{\mathsf T}S\tilde w,
  \qquad
  \tilde g(x)=x_0^2-|\tilde w|^2 ,
\end{equation*}
homogeneous quadratic forms on $\mathbb{R}^4$ with $\tilde q(1,w)=q(w)$ and
$\tilde g(1,w)=g(w)$.  The cone $\{\tilde g=0\}$ is
$\{c\,(\pm1,\hat s):|\hat s|=1,\ c\in\mathbb{R}\}$, and $\tilde q$ is even, so
``$\tilde q\ge0$ on $\{\tilde g=0\}$'' is exactly the null energy condition at the
point.  The equality-constrained S-lemma states that for homogeneous quadratic
forms with $\tilde g$ indefinite~\cite{polik2007slemma,xia2016slemma},
\begin{equation*}
  \tilde q\ge0 \text{ on } \{\tilde g=0\}
  \iff
  \exists\,\sigma\in\mathbb{R}:\ \tilde q-\sigma\tilde g\ge0 \text{ on }\mathbb{R}^n,
\end{equation*}
the multiplier now free in sign because the constraint is an equality.
Here $\tilde g$ is indefinite since $\tilde g(1,0)=1$ while
$\tilde g(0,e_1)=-1$.  The multiplier must be allowed either sign: for
$\widehat T=\mathrm{diag}(-1,2,2,2)$ the null contraction is $q(s)=1>0$, so the null
energy condition holds, yet $M(\sigma)\succeq0$ only for
$\sigma\in[-2,-1]$.  Finally $\tilde g$ has matrix $-\eta$, so $\tilde q-\sigma\tilde g$ has matrix
$\widehat T+\sigma\eta=M(\sigma)$ of \eqref{eq:lmi}, and nonnegativity of that form
is $M(\sigma)\succeq0$.  This is~(i).

For~(iii), with the unnormalized $u^a=n^a+w^ie_i^a$
one has $g_{ab}u^au^b=-1+|w|^2$ and hence
$\Theta_{ab}u^au^b=q(w)+\tfrac12 T\,g(w)$, so the SEC is the WEC for $\Theta$.
For~(iv), the DEC requires $T_{ab}u^au^b\ge0$ together with
$J^a=-T^a{}_bu^b$ causal; and $J_aJ^a=(T^2)_{ab}u^au^b$, so causality of $J$ is
the WEC for $-(T^2)_{ab}$.  Future-directedness needs no separate test: since
$T_{ab}u^au^b=-J_au^a$, a causal $J$ with the WEC satisfied is future-directed
or zero.
\end{proof}

In a fixed tetrad, the optimized multiplier margin equals half the normalized
null-cone minimum.

\begin{corollary}[The multiplier margin is half the null-cone minimum]
\label{cor:margin}
In a fixed orthonormal tetrad, writing $k^a=n^a+s^ie_i^a$ with $|s|=1$,
\begin{equation}
  2\max_{\sigma\in\mathbb{R}}\lambda_{\min}\bigl(M(\sigma)\bigr)
  \;=\;\min_{|s|=1}T_{ab}k^ak^b ,
  \label{eq:margin-identity}
\end{equation}
and the maximum is attained.
\end{corollary}
\begin{proof}
Write $\nu$ for the right-hand side.  For any $\sigma$ and any unit $s$ the vector
$x=(1,s)/\sqrt2$ has Euclidean norm one, and since $\tilde g(1,s)=0$,
\begin{equation*}
  x^{\mathsf T}M(\sigma)\,x
  =\tfrac12\bigl[\tilde q(1,s)-\sigma\tilde g(1,s)\bigr]
  =\tfrac12\,T_{ab}k^ak^b ,
\end{equation*}
so minimizing over $s$ gives $\lambda_{\min}(M(\sigma))\le\nu/2$ for every $\sigma$.
Conversely $\widehat T-\tfrac{\nu}{2}\mathbb{1}$ satisfies the null condition: a null
$(1,s)$ has Euclidean square $2$, so its contraction drops by exactly $\nu$ and is
nonnegative by the definition of $\nu$.  Part~(i) then supplies a $\sigma$ with
$\widehat T-\tfrac{\nu}{2}\mathbb{1}+\sigma\eta\succeq0$, that is
$\lambda_{\min}(M(\sigma))\ge\nu/2$.  Attainment follows from
$\lambda_{\min}(M(\sigma))\le-|\sigma|+\lambda_{\max}(\widehat T)$, which sends the
concave function $\sigma\mapsto\lambda_{\min}(M(\sigma))$ to $-\infty$ at both ends.
\end{proof}

\paragraph{Multiplier search and exact certificates.}
$M$ is affine in $\sigma$, so $\lambda_{\min}(M(\sigma))$ is \emph{concave} and the
search for a multiplier is a one-dimensional concave maximization.  We use
a ternary search, evaluating a $4\times4$ symmetric eigenvalue problem at each
step.  For rational input components, a rational
$\sigma$ together with an exact
$LDL^{\mathsf T}$ factorization of $M(\sigma)$ is a certificate a reader can check
by hand.  This is a sufficient test.  At a saturated optimum the admissible
multiplier can be a single irrational value; an unsuccessful interval test
remains inconclusive, and exact algebraic evaluation may be needed.

For example, with $\rho_n=3$, $b=(1,0,0)$ and
$S=\mathbb{1}$,
\begin{equation}
  \widehat T=
  \begin{pmatrix} 3&-1&0&0\\ -1&1&0&0\\ 0&0&1&0\\ 0&0&0&1 \end{pmatrix},
  \qquad
  M(1)=\widehat T+\eta=
  \begin{pmatrix} 2&-1&0&0\\ -1&2&0&0\\ 0&0&2&0\\ 0&0&0&2 \end{pmatrix},
  \label{eq:worked-certificate}
\end{equation}
and $M(1)=LDL^{\mathsf T}$ with $L$ unit lower triangular carrying the single entry
$L_{10}=-\tfrac12$ and $D=\mathrm{diag}(2,\tfrac32,2,2)$.  Every entry is rational and
every diagonal entry of $D$ is positive, so $M(1)\succeq0$ and the multiplier
$\sigma=1\ge0$ certifies the null and weak conditions at that point.  The certificate
can be checked by multiplying out $LDL^{\mathsf T}$.  For comparison, the admissible multipliers here
are $\sigma\in[1-\sqrt3,\,1+\sqrt3]$, so the interval of rational choices is wide and
$\sigma=1$ is an interior point rather than a saturating one.

For the WEC, when no nonnegative multiplier exists, a violating observer can
be sought by minimizing $q$ on the ball directly.  A negative-eigenvalue
eigenvector of $M(\sigma_\star)$ need not give a violating observer.
Writing that eigenvector as $(u_0,\mathbf{u})$ and taking
$w=\mathbf{u}/u_0$ when $u_0\ne0$, one has
$q(w)=\sigma(1-|w|^2)+\lambda_{\min}|(u_0,\mathbf{u})|^2/u_0^2$, the first term being nonnegative for
$\sigma\ge0$ and $|w|\le1$, so the negative term has to outweigh it, which is not guaranteed by a negative eigenvalue alone.  That construction also fails
outright at a repeated $\lambda_{\min}$, where the eigensolver returns an arbitrary
basis of the eigenspace: for $\widehat T=\operatorname{diag}(1,-3,10,10)$ the optimum is
$\sigma_\star=2$ with $M=\operatorname{diag}(-1,-1,12,12)$, and the routine can select
$e_0$, returning $w=0$ with $q(0)=+1$, while $w=(-0.999,0,0)$ gives
$q\approx-1.994$.  Minimizing $q$ on the ball
avoids both failures, and it does so without needing the minimum to be global: any
$w$ with $|w|\le1$ and $q(w)<0$ refutes the condition: an interior point
gives a timelike observer, and a negative boundary value persists at nearby
interior points by continuity.  The search is therefore one-sided by construction, and a failed witness search remains inconclusive; satisfaction requires a valid
multiplier certificate.  A quadratic on a ball can have a local
minimizer that is not global: $q=-x^2+x+y^2+z^2$ has a strict local minimum at
$w=+e_1$ with $q=0$, against the global $-2$ at $w=-e_1$.  A descent that stalled
there would return no witness, so it must return an inconclusive result.

A proposed multiplier is checked in exact arithmetic.  A proposed timelike
violating vector can likewise be checked by direct contraction.
The check can also use coordinate components.  With $e$ the matrix of tetrad
vectors, $M(\sigma)=e(T+\sigma g)e^{\mathsf T}$ is congruent to the coordinate
matrix, and congruence preserves positive semidefiniteness.

For the null condition, whose multiplier is free in sign, a matrix certificate
can establish violation without explicitly finding a null vector.
Write $\langle A,X\rangle=\operatorname{tr}(A^{\mathsf T}X)$ for the matrix pairing.
To refute every $\sigma$ at once a dual variable $X\succeq0$ must satisfy
$\langle g,X\rangle=0$ \emph{exactly}, since then
$0\le\langle T+\sigma g,X\rangle=\langle T,X\rangle$ for every $\sigma$, so
$\langle T,X\rangle<0$ leaves no multiplier available.  A rank-one
$X=kk^{\mathsf T}$ would need a rational null $k$, and a rational Lorentzian form
need not represent zero: over $\mathbb{Q}$ the form $\mathrm{diag}(-15,1,1,1)$ has no
nonzero null vector, because $15$ is not a sum of three rational squares.
Indeed, clearing denominators gives $a^2+b^2+c^2=15d^2$ in integers.
After removing common factors of two, $d$ must be odd: if $d$ were even,
reduction modulo four would force $a,b,c$ even.  For odd $d$, reduction
modulo eight would require three squares to sum to seven, which is impossible.  A
rank-two $X$ can instead provide a dual certificate.  Take rational $k,\ell$ whose norms
straddle the cone, $g(k,k)\,g(\ell,\ell)<0$, and assume that the coordinate matrices $g$ and $T$ are rational.  Set
\begin{equation}
  X=\alpha\,kk^{\mathsf T}+\beta\,\ell\ell^{\mathsf T},
  \qquad \alpha=|g(\ell,\ell)|,\quad \beta=|g(k,k)| .
  \label{eq:nec-dual-pair}
\end{equation}
Then $\langle g,X\rangle=\alpha\,g(k,k)+\beta\,g(\ell,\ell)=0$ identically, by the
opposite signs, while $\langle T,X\rangle=\alpha\,T(k,k)+\beta\,T(\ell,\ell)$ is a
rational number whose negativity refutes every multiplier at once; a nonnegative
value refutes nothing.  Straddling alone is not sufficient:
at least one contraction must be negative, and the weighted sum must be negative.  For $\widehat T=\mathrm{diag}(1,-3,0,0)$ and
$\ell=(0,0,1,0)$, the timelike $k=(5,4,0,0)$ gives $\langle T,X\rangle=-23$ and
refutes the null condition, whereas $k=(2,1,0,0)$ gives $+1$ and refutes nothing.

\paragraph{Pointwise and spatial bounds.}
Theorem~\ref{thm:lmi} is \emph{pointwise}: it asserts $\forall x\,\exists\,
\sigma(x)$.  It does not lift to a spatial box, where a single interval LMI would
prove the strictly stronger $\exists\,\sigma\,\forall x$.  That stronger statement
already fails on the saturated family $\rho=a(x)>0$, $b=0$, $S=-a(x)\mathbb{1}$,
for which $q\ge0$ everywhere while $\sigma=a(x)$ is forced pointwise, so no common
multiplier exists on any box over which $a$ varies, and subdivision does not repair
it.  Global coverage over a domain is therefore left with the interval branch and
bound of Appendix~\ref{app:enclosures}, which brackets the sphere minimum $\min_{|w|=1}q$ directly.

\paragraph{Axial and full null minima.}
A single axial witness need not minimize the null contraction.  For the
canonical Type~II block $(\mu,f,p_2,p_3)=(0,1,-2,0)$ it is zero, whereas
$k=(1,0,1,0)$ gives $T_{ab}k^ak^b=-1$.  On the unit sphere the contraction is
$3s_x^2-2s_x-1+2s_z^2$, whose minimum $-4/3$ occurs at $s_x=1/3$, $s_z=0$.
Theorem~\ref{thm:lmi} gives the fixed-tetrad margin $-2/3$, as required by
\eqref{eq:margin-identity}.  The null-dust example has zero for both the axial
witness and the null-cone minimum, so it does not distinguish the two diagnostics.
Similarly, an imaginary-eigenvalue diagnostic supplies no quantitative margin for
Type~III, whose exact eigenvalues are real.

\paragraph{Numerical classification near degeneracies.}
The eigenvalue classifier depends on its tolerance near defective tensors.
For the perturbed Jordan families used here, a perturbation $\delta$ displaces
an eigenvalue by $\delta^{1/m}$ for block size $m$, so the dependence is not
Lipschitz at $\delta=0$.  Fits over $\delta\in[10^{-14},10^{-8}]$ give slopes
$0.509$ for $J_2$ and $0.334$ for $J_3$, close to $1/2$ and $1/3$.
On the exact $J_3(-1)\oplus[3]$ test tensor, float64 gives a complex split of
$5.7\times10^{-6}$; the classifier reports Type~IV at $\mathrm{tol}=10^{-10}$,
Type~II at $2\times10^{-6}$ and Type~III at $5\times10^{-6}$.
These are observations for the stated test tensor and solver.  No fixed numerical
precision resolves every exact degeneracy.  At $50$ digits the test tensors
are correctly separated using the Jordan defect of the repeated eigenvalue.

We therefore report numerical type labels and detected disagreements as
diagnostics.  The energy-condition decision uses the LMI theorem and separately
checked certificates.  Garattini--Zatrimaylov uses its exact Type-I result.
A fraction below $100\%$ for WEC violations among numerically labeled Type~III
points would reveal an inconsistency between the numerical routes; it would not
by itself identify which route failed.

Theorem~\ref{thm:lmi} applies even when numerical classification is
inconclusive; its numerical evaluation remains subject to the stated noise floor.  For the Type~III tensor
\eqref{eq:typeIII-example}, direct contraction gives an exact verdict:
$T_{ab}=\ell_am_b+m_a\ell_b$ gives $T_{ab}k^ak^b=2(\ell\!\cdot\!k)(m\!\cdot\!k)$, so
the rational null direction $k^a=(1,0,1,0)$ returns $-2$, and no multiplier can
certify a condition that a rational null vector violates.

The type also supplies a consistency check.  Types~III and~IV
violate \emph{every} standard energy
condition~\cite{martinmoruno2017lnp,martinmoruno2018core}: for a Type~III tensor the
null condition cannot be satisfied once the off-diagonal Jordan entry is nonzero, so
a nontrivial Type~III point violates all four.  Any point the classifier labels III or IV at which the LMI
numerically reports satisfaction gives a numerically detected label contradiction.  In the separate regularized WarpShell diagnostic at $v_s=0.5$, the
full $50^3$ grid over $[-25,25]^3$ contains $92\,616$ Type~I and $32\,384$ Type~IV points and no Type~II or
Type~III point at all.  The WEC diagnostic is negative at all $32\,384$
Type-IV points, and no NEC margin exceeds the positive noise floor; the check
therefore records zero detected contradictions.  That verdict is reported
and not certified: the multiplier search runs in binary64 and $\lambda_{\min}$
from a symmetric eigensolver carries a backward error of order
$\varepsilon\lVert M\rVert$, so a margin is a verdict only once it clears a noise
floor in either direction.  That floor is
\begin{equation}
  \epsilon_{\rm floor}=10^{-12}\,\max_{IJ}|\widehat T_{IJ}|+10^{-18},
  \label{eq:noise-floor}
\end{equation}
relative to the scale of the tensor, with a
small absolute term that keeps exact vacuum, where the residual ternary bracket
returns about $-6\times10^{-22}$, from being reported as a violation.  Inside the
floor the answer is
\emph{inconclusive}, and a separate exact check may resolve it: a rational
multiplier checked by an exact $LDL^{\mathsf T}$, or a checked violating-vector or dual certificate, as appropriate to the
condition.

\paragraph{The same check over the reported family.}
The same test applied to the four constructions of this paper, at every point of
the wall-clustered $48^3$ benchmark grid and at $v_s=0.5,1,2$, is
Table~\ref{tab:lmi_typefree}(a).  The eigenvalue and LMI routes agree on $100.00\%$ of decisive Type-I
comparisons in every row.  The exact criteria are equivalent, so a numerical
disagreement could arise in either implementation.  The
fraction of points labeled Type~III or~IV at which the inequality numerically reports
\emph{satisfaction}, which by the argument above cannot be reconciled with the
label, is zero over $619\,288$ such points across the family.  Both are
thresholded against \eqref{eq:noise-floor} and not against zero.  This does not make the
classifier exact: the decisive comparisons on these sampled geometries show no
disagreement, and the inequality decides every
wall point at which the type-based route returns no verdict ($1\,304$ of $1\,304$
for Alcubierre at $v_s=0.5$, and likewise at every other entry).

\begin{table}[tbp]
  \centering
  \footnotesize
  \caption{\emph{(a)}~Type-free check of the type-based route on the wall-clustered
    $48^3$ benchmark grid.  \emph{Left}: agreement of Theorem~\ref{thm:lmi} with the
    eigenvalue inequalities among decisive Type-I comparisons.
    \emph{Right}: Type-III and Type-IV
    points at which the inequality numerically reports satisfaction.  Both thresholded against
    \eqref{eq:noise-floor}; Rodal has no non-Type-I point, so its right-hand entries
    are undefined.  \emph{(b)}~Crossing the Type-I/Type-IV transition on the analytic
    families of the text.  ``Labeled II'' counts sweep points the float64 classifier
    returns as Type~II.  The maximum adjacent-grid slope is an observed statistic.
    The last column compares twice the numerical LMI margin with the minimum
    contraction over $400\,000$ sampled null directions; the sampling error is
    not enclosed.}
  \label{tab:lmi_typefree}

  \emph{(a)}\;\begin{tabular}{@{}l ccc ccc@{}}
  \toprule
  & \multicolumn{3}{c}{Type-I verdicts agreeing (\%)} & \multicolumn{3}{c}{Numerically detected label contradictions (\%)} \\
  \cmidrule(lr){2-4}\cmidrule(lr){5-7}
  Metric & $v_s=0.5$ & $1.0$ & $2.0$ & $v_s=0.5$ & $1.0$ & $2.0$ \\
  \midrule
  Alcubierre & 100.00 & 100.00 & 100.00 & 0.00 & 0.00 & 0.00 \\
  Natário & 100.00 & 100.00 & 100.00 & 0.00 & 0.00 & 0.00 \\
  Van den Broeck & 100.00 & 100.00 & 100.00 & 0.00 & 0.00 & 0.00 \\
  Rodal & 100.00 & 100.00 & 100.00 & -- & -- & -- \\
  \bottomrule
\end{tabular}

  \medskip
  \emph{(b)}\;\begin{tabular}{@{}l cccc@{}}
  \toprule
  Family & points & labeled II & max.\ grid slope & $|2\mathcal{N}_{\rm LMI} - \min_{s\in\mathrm{sample}}q(s)|$ \\
  \midrule
  Momentum, $\Delta\to0$ & 401 & 1 & 1 & 2.8e-06 \\
  Momentum, decoupled & 401 & 1 & 0.499 & 2.0e-09 \\
  Transverse ($\Delta>0$) & 301 & 0 & 0.499 & 8.5e-05 \\
  \bottomrule
\end{tabular}

\end{table}

The separate WarpShell comparison has three limits.  First,
$89\,808$ of the $92\,616$ Type-I labels are near-vacuum points assigned Type~I by
the near-vacuum test rather than computed from an eigenstructure, so its
discriminating power lies in the $2\,808$ resolved Type-I points and the $32\,384$
Type-IV ones.  Second, the Type-III case tested \emph{zero} points, so it
supplies no evidence for that case.  Third, the regularized WarpShell lies
outside the matched benchmark family.  A finite spatial grid need not sample
a type-transition locus.  The analytic tensor families below provide known
examples on and across those loci.

\paragraph{Transitions in analytic tensor families.}
On analytic one-parameter families the
algebraic type is known in closed form, an independent brute-force scan over
$400\,000$ null directions is available at every point, and the transition can be
crossed deliberately rather than sampled.  Three families are swept
(Table~\ref{tab:lmi_typefree}(b)),
all in a Minkowski orthonormal frame so that the classifier, the inequality and the
scan operate on the same matrix.

The first is the momentum block itself, $\rho=S_\parallel=1$ with $p_\perp=0$ and
$|j|$ swept through $1$, so that $\Delta=4(1-|j|^2)$ carries the point from Type~I
through the nilpotent $J_2(0)$ at $|j|=1$ to Type~IV.  Across $401$ samples the
maximum observed adjacent-grid slope of the margin is $1$.  The analytic margin
passes continuously through the locus
($+4.99\times10^{-3}$ at $|j|=0.995$, exactly $0$ at $|j|=1$,
$-5.00\times10^{-3}$ at $|j|=1.005$), while the label is discontinuous there.
Twice the margin differs from the finite-sample null minimum by at most
$2.8\times10^{-6}$, an observed discrepancy from the angular scan.

The second family fixes the same $\Delta$ and the same Type~II point but sets
$p_\perp=-3$, which puts the null-cone minimum at $-2-|j|^2/4$.  The label still
runs I\,$\to$\,II\,$\to$\,IV across the sweep while the sign of the margin does not
change at any of the $401$ points, so the algebraic type is not a proxy for an
energy-condition verdict: the type changes while the NEC remains violated.

The third is the transverse family of Corollary~\ref{cor:conformal}, where
$\Delta=3>0$ throughout and Type~IV is opened by the transverse coupling instead.
The measured band endpoints, $m=0.78$ and $m=2.01$ at a sweep spacing of $0.01$,
bracket the predicted $m_\star\approx0.782$ and $m_2\approx2.015$
within one sweep step.

A separate fourth family tests the Type~III case.  On an exact Segre $[3,1]$
family, $121$ points log-spaced over $|f|\in[10^{-6},1]$, the $50$-digit recomputation returns
Type~III at every one.  The float64 classifier returns Type~IV at every one of the
$121$ at $\mathrm{tol}=10^{-10}$; at $\mathrm{tol}=2\times10^{-6}$ it returns
Type~II at $63$, Type~III at $40$ and Type~IV at $18$, and at $5\times10^{-6}$ it
returns $70$, $45$ and $6$.  The label is a function of the tolerance across the
whole family, and it is never right at the tolerance used here.  The binary64 linear
matrix inequality detects null-energy violation at all $121$ tested points.  The verdict
therefore survives six decades of chain strength over which the label does not.

\subsection{Rigorous pointwise bounds from the metric}
\label{app:interval-lmi}

An exact certificate applies to its input tensor.  To establish an
energy-condition verdict for a prescribed metric, the numerical error in that
tensor must also be enclosed.  Appendix~\ref{app:enclosures} provides this
step by interval evaluation of the curvature.

At a point, the interval curvature chain encloses the induced tensor.
The pointwise tests below use the interval $LDL^{\mathsf T}$ lower bound:
a trial spectral bound $\ell$ is accepted when the factorization establishes
$M(\sigma)-\ell\mathbb{1}\succeq0$.  Then $2\ell$ bounds $\mathcal N$ from below.
An interval-normalized null contraction supplies the upper bound.
Together these bounds determine the NEC sign whenever the resulting interval
excludes zero.

Table~\ref{tab:interval_lmi} applies these bounds at three wall points of each
drive at $v_s=\tfrac12$: $(x,s)=(1,0),(0.9,0.35),(0.95,0.5)$ in units of $R_b$.  The brackets close to a width between $1.1\times10^{-16}$
and $3.3\times10^{-14}$ in absolute terms, which is ten significant digits on
Van~den~Broeck's
small on-axis margin and fourteen on the deep wall values, with the input tensor and spectral bounds enclosed by directed rounding.
A floating-point eigensolver proposes $\sigma$; the interval
$LDL^{\mathsf T}$ test establishes the bound for that proposal.
A poor proposal can widen the bracket but cannot invalidate it.  For Van~den~Broeck the NEC is certified satisfied at $(x,s)=(1,0)$ and
$(0.9,0.35)$, and violated at $(0.95,0.5)$.

The composition covers the null, weak, strong and
dominant conditions alike: $\Theta_{ab}$ and $-(T^2)_{ab}$ are interval-algebraic
rearrangements of the components $(\rho_n,b_i,S_{ij})$ returned by the chain, and by
Theorem~\ref{thm:lmi} the strong and dominant cases are the ball condition on
those two tensors.  For a null $k$, $\Theta(k,k)=T(k,k)$, so a strict null violation also implies
SEC violation by continuity toward timelike vectors.  The numerical SEC
certificate evaluates a timelike witness for $\Theta$ directly.  The weak-condition
ball includes $w=0$, so $\rho_n<0$ gives a weak-condition witness without a
direction search.

The strict interval test remains \emph{inconclusive} at saturation.
For exact vacuum, $q\equiv0$ establishes satisfaction analytically.

\begin{table}[tbp]
  \centering
  \footnotesize
  \caption{The null-energy verdict taken back to the metric, at $R_b=1$,
    $\sigma=8$, $v_s=\tfrac12$.  $\mathcal{N}$ is the type-free null deficit
    \eqref{eq:enclosure-objective}.  The interval curvature chain encloses the
    tensor; the lower endpoint follows from the interval $LDL^{\mathsf T}$ test
    and the upper endpoint from an interval-normalized null contraction.  The two
    endpoints agree to the six decimals printed, so the width is given as its own
    column.}
  \label{tab:interval_lmi}
  \begin{tabular}{@{}l l rr r c@{}}
  \toprule
  Metric & wall point & $\mathcal{N}_{\rm lower}$ & $\mathcal{N}_{\rm upper}$ & width & verdict \\
  \midrule
  Alcubierre & on axis & $-0.198944$ & $-0.198944$ & $7.7\times10^{-15}$ & violated \\
  Alcubierre & wall, $s=0.35$ & $-0.343691$ & $-0.343691$ & $2.8\times10^{-16}$ & violated \\
  Alcubierre & wall, $s=0.5$ & $-0.146460$ & $-0.146460$ & $1.1\times10^{-16}$ & violated \\
  Natário & on axis & $-0.557042$ & $-0.557042$ & $3.3\times10^{-16}$ & violated \\
  Natário & wall, $s=0.35$ & $-1.275487$ & $-1.275487$ & $8.9\times10^{-16}$ & violated \\
  Natário & wall, $s=0.5$ & $-1.207662$ & $-1.207662$ & $6.7\times10^{-16}$ & violated \\
  Van den Broeck & on axis & $+0.000195$ & $+0.000195$ & $4.5\times10^{-15}$ & satisfied \\
  Van den Broeck & wall, $s=0.35$ & $+0.034675$ & $+0.034675$ & $1.3\times10^{-16}$ & satisfied \\
  Van den Broeck & wall, $s=0.5$ & $-1.006795$ & $-1.006795$ & $4.4\times10^{-16}$ & violated \\
  Rodal & on axis & $-0.156336$ & $-0.156336$ & $1.1\times10^{-14}$ & violated \\
  Rodal & wall, $s=0.35$ & $-0.169108$ & $-0.169108$ & $1.7\times10^{-16}$ & violated \\
  Rodal & wall, $s=0.5$ & $-0.000911$ & $-0.000911$ & $3.3\times10^{-14}$ & violated \\
  \bottomrule
\end{tabular}

\end{table}

\subsection{Interval tests on a wall grid}
\label{app:interval-census}

Table~\ref{tab:interval_census} applies the interval tests to $4257$ wall
points per drive, on a $33\times129$ product grid in $(r,\theta)$ over
$f\in[0.1,0.9]$.  All four conditions are decided at each sampled point.

The tests require neither a Hawking--Ellis classification nor a rapidity cap.
They also apply on the Type~II locus of \eqref{eq:momentum-trichotomy}, where
numerical classification can be ill-conditioned.  Across $17{,}028$ points and
$68{,}112$ verdicts, no interval evaluation failed and no verdict remained
inconclusive.

These counts describe a finite sample, with band membership certified at each
point; they do not measure wall volume.  The separate branch-and-bound brackets
in Table~\ref{tab:enclosures} bound the NEC infimum over the whole wall.

The census and global search supply upper bounds on the same infimum
from different points.  Either upper bound can be smaller.  For every drive,
the census upper endpoint lies above the global lower endpoint, as required.  Relative to the magnitude of the search upper endpoint, the two
upper bounds differ by $0.17\%$, $0.084\%$, $1.14\%$ and $0.020\%$ for
Alcubierre, Nat\'ario, Van~den~Broeck and Rodal, respectively.
These observed differences do not bound the error in the global infimum.

\begin{table}[tbp]
  \centering
  \footnotesize
  \caption{Every sampled wall point decided from the metric, at $R_b=1$,
    $\sigma=8$, $v_s=\tfrac12$ (Van~den~Broeck also $\tilde R=1$, $\alpha=\tfrac12$,
    $\sigma_B=8$), $80$-bit working precision.  Per construction, $4257$
    points of a $33\times129$ grid in $(r,\theta)$ over the band $f\in[0.1,0.9]$
    ($r\in[0.8627,1.1373]$).  Entries count
    \emph{sampled points}, not wall measure.  ``Inconclusive'' counts unresolved
    verdicts, which are not evidence that the condition holds.  ``Deepest upper''
    is the most negative certified upper endpoint attained.}
  \label{tab:interval_census}
  \begin{tabular}{@{}l l rrr r@{}}
  \toprule
  Construction & condition & certified & certified & inconclusive & deepest \\
   &  & violated & satisfied &  & upper \\
  \midrule
  Alcubierre & NEC & 4257 & 0 & 0 & $-0.6300$ \\
   & WEC & 4257 & 0 & 0 & $-0.6300$ \\
   & SEC & 4257 & 0 & 0 & $-0.6300$ \\
   & DEC & 4257 & 0 & 0 & $-0.6300$ \\
  \cmidrule(l){1-6}
  Natário & NEC & 4257 & 0 & 0 & $-5.2169$ \\
   & WEC & 4257 & 0 & 0 & $-5.2169$ \\
   & SEC & 4257 & 0 & 0 & $-5.2169$ \\
   & DEC & 4257 & 0 & 0 & $-7.1228$ \\
  \cmidrule(l){1-6}
  Van den Broeck & NEC & 3359 & 898 & 0 & $-1.2962$ \\
   & WEC & 3359 & 898 & 0 & $-1.2962$ \\
   & SEC & 4257 & 0 & 0 & $-1.2962$ \\
   & DEC & 3359 & 898 & 0 & $-1.2962$ \\
  \cmidrule(l){1-6}
  Rodal & NEC & 3811 & 446 & 0 & $-0.1942$ \\
   & WEC & 3811 & 446 & 0 & $-0.1942$ \\
   & SEC & 3829 & 428 & 0 & $-0.2235$ \\
   & DEC & 4205 & 52 & 0 & $-0.1942$ \\
  \bottomrule
\end{tabular}

\end{table}

\providecommand{\newblock}{}


\begin{thebibliography}{10}
\expandafter\ifx\csname url\endcsname\relax
  \def\url#1{{\tt #1}}\fi
\expandafter\ifx\csname urlprefix\endcsname\relax\def\urlprefix{URL }\fi
\providecommand{\eprint}[2][]{\url{#2}}

\bibitem{alcubierre1994}
Alcubierre M 1994 {\em Class. Quantum Grav.\/} {\bf 11} L73--L77
  (\textit{Preprint} \eprint{gr-qc/0009013})

\bibitem{santiago2021}
Santiago J, Schuster S and Visser M 2022 {\em Phys. Rev. D\/} {\bf 105} 064038
  (\textit{Preprint} \eprint{2105.03079})

\bibitem{hawking1973}
Hawking S~W and Ellis G~F~R 1973 {\em The Large Scale Structure of
  Space-Time\/} (Cambridge University Press)

\bibitem{gergely2026braiding}
Gergely L~{\'A} 2026 Fluid interpretation, Hawking--Ellis classification, and
  energy conditions of the proper kinetic gravity braiding stress tensor
  (\textit{Preprint} \eprint{2608.15228})

\bibitem{martinmoruno2018core}
Mart{\'\i}n-Moruno P and Visser M 2018 {\em Class. Quantum Grav.\/} {\bf 35}
  125003 (\textit{Preprint} \eprint{1802.00865})

\bibitem{rodal2026}
Rodal J 2026 {\em Gen. Relativ. Gravit.\/} {\bf 58} 1 (\textit{Preprint}
  \eprint{2512.18008})

\bibitem{fuchs2024constvel}
Fuchs J, Helmerich C, Bobrick A, Sellers L, Melcher B and Martire G 2024 {\em
  Class. Quantum Grav.\/} {\bf 41} 095013 (\textit{Preprint}
  \eprint{2405.02709})

\bibitem{bobrick2021}
Bobrick A and Martire G 2021 {\em Class. Quantum Grav.\/} {\bf 38} 105009
  (\textit{Preprint} \eprint{2102.06824})

\bibitem{fell2021}
Fell S~D~B and Heisenberg L 2021 {\em Class. Quantum Grav.\/} {\bf 38} 155020
  (\textit{Preprint} \eprint{2104.06488})

\bibitem{garattini2025desitter}
Garattini R and Zatrimaylov K 2026 Warp drive in a de Sitter universe
  (\textit{Preprint} \eprint{2502.13153})

\bibitem{jax2018}
Bradbury J, Frostig R, Hawkins P, Johnson M~J, Leary C, Maclaurin D, Necula G,
  Paszke A, Vander{P}las J, Wanderman-{M}ilne S and Zhang Q 2018 {JAX}:
  composable transformations of {Python+NumPy} programs
  \url{https://github.com/jax-ml/jax}, version 0.9.0

\bibitem{kidger2021equinox}
Kidger P and Garcia C 2021 Equinox: neural networks in {JAX} via callable
  {PyTrees} and filtered transformations (\textit{Preprint} \eprint{2111.00254})

\bibitem{warpfactory2024}
Helmerich C {\em et~al.\/} 2024 {\em Class. Quantum Grav.\/} {\bf 41} 095009
  (\textit{Preprint} \eprint{2404.03095})

\bibitem{martinmoruno2017lnp}
Mart{\'\i}n-Moruno P and Visser M 2017 Classical and semi-classical energy
  conditions {\em Wormholes, Warp Drives and Energy Conditions\/} ({\em
  Fundamental Theories of Physics\/} vol 189) ed Lobo F~S~N (Springer) pp
  193--213 (\textit{Preprint} \eprint{1702.05915})

\bibitem{martinmoruno2018type3}
Mart{\'\i}n-Moruno P and Visser M 2018 {\em Class. Quantum Grav.\/} {\bf 35}
  185004 (\textit{Preprint} \eprint{1806.02094})

\bibitem{martinmoruno2017rainich}
Mart{\'\i}n-Moruno P and Visser M 2017 {\em Class. Quantum Grav.\/} {\bf 34}
  225014 (\textit{Preprint} \eprint{1707.04172})

\bibitem{martinmoruno2021backreaction}
Mart{\'\i}n-Moruno P and Visser M 2021 {\em Phys. Rev. D\/} {\bf 103} 124003
  (\textit{Preprint} \eprint{2102.13551})

\bibitem{barzegar2026classification}
Barzegar H, Buchert T and Vigneron Q 2026 General formalism, classification, and
  demystification of the current warp-drive spacetimes (\textit{Preprint}
  \eprint{2602.16495})

\bibitem{olum1998}
Olum K~D 1998 {\em Phys. Rev. Lett.\/} {\bf 81} 3567--3570 (\textit{Preprint}
  \eprint{gr-qc/9805003})

\bibitem{lobovisser2004}
Lobo F~S~N and Visser M 2004 {\em Class. Quantum Grav.\/} {\bf 21} 5871--5892
  (\textit{Preprint} \eprint{gr-qc/0406083})

\bibitem{kontou2020}
Kontou E~A and Sanders K 2020 {\em Class. Quantum Grav.\/} {\bf 37} 193001
  (\textit{Preprint} \eprint{2003.01815})

\bibitem{celmaster2025}
Celmaster B and Rubin S 2025 Violations of the weak energy condition for Lentz
  warp drives (\textit{Preprint} \eprint{2511.18251})

\bibitem{natario2002}
Nat{\'a}rio J 2002 {\em Class. Quantum Grav.\/} {\bf 19} 1157--1166
  (\textit{Preprint} \eprint{gr-qc/0110086})

\bibitem{vandenbroeck1999}
Van Den~Broeck C 1999 {\em Class. Quantum Grav.\/} {\bf 16} 3973--3979
  (\textit{Preprint} \eprint{gr-qc/9905084})

\bibitem{krasnikov1998}
Krasnikov S~V 1998 {\em Phys. Rev. D\/} {\bf 57} 4760--4766 (\textit{Preprint}
  \eprint{gr-qc/9511068})

\bibitem{lentz2021}
Lentz E~W 2021 {\em Class. Quantum Grav.\/} {\bf 38} 075015 (\textit{Preprint}
  \eprint{2006.07125})

\bibitem{rodal2023invariants}
Rodal J 2023 {\em Gen. Relativ. Gravit.\/} {\bf 55} 134 (\textit{Preprint}
  \eprint{2512.20738})

\bibitem{rodal2024natario}
Rodal J 2024 {\em Int. J. Theor. Phys.\/} {\bf 63} 168 (\textit{Preprint}
  \eprint{2512.19837})

\bibitem{garattini2024bh}
Garattini R and Zatrimaylov K 2024 {\em Phys. Lett. B\/} {\bf 856} 138910
  (\textit{Preprint} \eprint{2408.04495})

\bibitem{clough2024gw}
Clough K, Dietrich T and Khan S 2024 {\em The Open Journal of Astrophysics\/}
  {\bf 7} \url{https://doi.org/10.33232/001c.121868}
  (\textit{Preprint} \eprint{2406.02466})

\bibitem{rodal2025infeasibility}
Rodal J 2025 On the infeasibility of low-energy warp drive via metamaterial
  gravitational coupling (\textit{Preprint} \eprint{2507.09724})

\bibitem{huey2024}
Huey G 2024 {\em Class. Quantum Grav.\/} {\bf 41} 135007
  (\textit{Preprint} \eprint{2311.07193})

\bibitem{santospereira2025matching}
Santos-Pereira O~L, Abreu E~M~C and Ribeiro M~B 2026 {\em Eur. Phys. J. C\/}
  {\bf 86} 46 (\textit{Preprint} \eprint{2512.12541})

\bibitem{barzegar2024restrictions}
Barzegar H and Buchert T 2025 {\em Universe\/} {\bf 11} 293 (\textit{Preprint}
  \eprint{2407.00720})

\bibitem{buchert2026realizations}
Buchert T and Frackowiak A 2026 {\em Universe\/} {\bf 12} 132
  (\textit{Preprint} \eprint{2605.03653})

\bibitem{white2025nacelle}
White H \textit{et~al} 2025 {\em Class. Quantum Grav.\/} {\bf 42} 235022

\bibitem{martingarcia2008}
Mart{\'\i}n-Garc{\'\i}a J~M 2008 {\em Comput. Phys. Commun.\/} {\bf 179}
  597--603 (\textit{Preprint} \eprint{0803.0862})

\bibitem{gourgoulhon2018}
Gourgoulhon {\'E} and Mancini M 2018 {\em Les cours du CIRM\/} {\bf 6}
  (\textit{Preprint} \eprint{1804.07346})

\bibitem{einsteintoolkit2024}
{Einstein Toolkit Consortium} 2024 The {Einstein Toolkit}: a community
  computational infrastructure for relativistic astrophysics
  \url{https://einsteintoolkit.org}

\bibitem{maccallum2018}
MacCallum M~A~H 2018 {\em Living Rev. Relativ.\/} {\bf 21} 6

\bibitem{warpfactory_toolkit2024}
Helmerich C, Fuchs J, Bobrick A, Melcher B, Sellers L and Martire G 2023 {Warp
  Factory}: A numerical toolkit for the analysis and optimization of warp drive
  geometries {\em AIAA SCITECH 2023 Forum\/} (\textit{Preprint}
  \eprint{2404.10855})

\bibitem{warpfactory_docs2024}
{Applied Physics} 2024 {WarpFactory} documentation: energy conditions analysis
  \url{https://applied-physics.gitbook.io/warp-factory/examples/analysis/a1-energy-conditions}

\bibitem{coogan2024diffjeom}
Coogan A 2021 diffjeom: differential geometry with {JAX}
  \url{https://github.com/adam-coogan/diffjeom}, version 0.0.1

\bibitem{cranganore2025einsteinfields}
Cranganore S~S, Bodnar A, Berzins A and Brandstetter J 2026 {Einstein Fields}: a
  neural perspective to computational general relativity {\em ICLR 2026\/}
  (\textit{Preprint} \eprint{2507.11589})

\bibitem{bara2025}
Bara M 2025 Unexplored opportunities for automatic differentiation in
  astrophysics (\textit{Preprint} \eprint{2507.09379})

\bibitem{optimistix2024}
Kidger P 2024 Optimistix: modular optimisation in {JAX}
  \url{https://github.com/patrick-kidger/optimistix}

\bibitem{kidger2022diffrax}
Kidger P 2022 Diffrax: numerical differential equation solvers in {JAX}
  \url{https://github.com/patrick-kidger/diffrax}

\bibitem{grahamolum2007}
Graham N and Olum K~D 2007 {\em Phys. Rev. D\/} {\bf 76} 064001
  (\textit{Preprint} \eprint{0705.3193})

\bibitem{ford1995}
Ford L~H and Roman T~A 1995 {\em Phys. Rev. D\/} {\bf 51} 4277--4286
  (\textit{Preprint} \eprint{gr-qc/9410043})

\bibitem{pfenning1997}
Pfenning M~J and Ford L~H 1997 {\em Class. Quantum Grav.\/} {\bf 14} 1743--1751
  (\textit{Preprint} \eprint{gr-qc/9702026})

\bibitem{fewster2000qi}
Fewster C~J 2000 {\em Class. Quantum Grav.\/} {\bf 17} 1897--1911
  (\textit{Preprint} \eprint{gr-qc/9910060})

\bibitem{fewstereveson1998}
Fewster C~J and Eveson S~P 1998 {\em Phys. Rev. D\/} {\bf 58} 084010
  (\textit{Preprint} \eprint{gr-qc/9805024})

\bibitem{nocedal2006}
Nocedal J and Wright S~J 2006 {\em Numerical Optimization\/} 2nd ed (Springer)

\bibitem{tsitouras2011}
Tsitouras C 2011 {\em Comput. Math. Appl.\/} {\bf 62} 770--775

\bibitem{tao2016symplectic}
Tao M 2016 {\em Phys. Rev. E\/} {\bf 94} 043303 (\textit{Preprint}
  \eprint{1609.02212})

\bibitem{christian2021fantasy}
Christian P and Chan C~K 2021 {\em Astrophys. J.\/} {\bf 909} 67
  (\textit{Preprint} \eprint{2010.02237})

\bibitem{abellan2024lapse}
Abell{\'a}n G, Bol{\'i}var N and Vasilev I 2024 {\em Class. Quantum Grav.\/} {\bf 41}
  105011

\bibitem{rodal2026screening}
Rodal J 2026 Tamm--Rubilar branch diagnostics for Drummond--Hathrell photon
  propagation: Schwarzschild calibration and a Kerr weak-lensing benchmark
  (\textit{Preprint} \eprint{2603.21352})


\bibitem{loewy1975icecream}
Loewy R and Schneider H 1975 {\em J. Math. Anal. Appl.\/} {\bf 49} 375--392

\bibitem{hildebrand2011lorentzII}
Hildebrand R 2011 {\em Linear Multilinear Algebra\/} {\bf 59} 719--731

\bibitem{bergqvist2001nullcone}
Bergqvist G and Senovilla J M~M 2001 {\em Class. Quantum Grav.\/} {\bf 18}
  5299--5325 (\textit{Preprint} \eprint{gr-qc/0104090})

\bibitem{shoshany2024ctc}
Shoshany B and Snodgrass B 2024 {\em Class. Quantum Grav.\/} {\bf 41} 205005
  (\textit{Preprint} \eprint{2309.10072})

\bibitem{kontou2024qei}
Kontou E-A 2024 {\em Universe\/} {\bf 10} 291 (\textit{Preprint}
  \eprint{2405.05963})

\bibitem{yakubovich1971}
Yakubovich V~A 1971 {\em Vestnik Leningrad Univ.\/} {\bf 1} 62--77 (S-procedure
  in nonlinear control theory)

\bibitem{polik2007slemma}
P\'olik I and Terlaky T 2007 {\em SIAM Rev.\/} {\bf 49} 371--418

\bibitem{hildebrand2007lorentz}
Hildebrand R 2007 {\em Linear Multilinear Algebra\/} {\bf 55} 551--573

\bibitem{maeda2020nonstatic}
Maeda H 2021 {\em Gen. Relativ. Gravit.\/} {\bf 53} 90 (\textit{Preprint}
  \eprint{2001.11335})

\bibitem{xia2016slemma}
Xia Y, Wang S and Sheu R-L 2016 {\em Math. Program.\/} {\bf 156} 513--547
  (\textit{Preprint} \eprint{1403.2816})

\bibitem{jusufi2026warp}
Jusufi K and Lobo F S~N 2026 Quantum-gravity-inspired Alcubierre warp-drive
  geometries (\textit{Preprint} \eprint{2609.05554})

\end{thebibliography}
\end{document}